\documentclass[acmsmall]{acmart}

\AtBeginDocument{%
  }

\newcommand{\paratitle}[1]{
\noindent{\bf #1.}}

\usepackage{appendix}
\usepackage{paralist}
\usepackage{subcaption}
\usepackage{standalone}
\usepackage{color, colortbl}
\usepackage{listings}
\usepackage{scalerel}
\usepackage{tikz}
\usepackage{amsmath}
\usepackage{pifont}
\usepackage[framemethod=TikZ]{mdframed}
\usepackage[capitalize,noabbrev]{cleveref}
\usepackage{enumitem}
\Crefname{algocf}{Algorithm}{Algorithms}

\usepackage[colorinlistoftodos]{todonotes}
\newcommand{\srtodo}[1]{{\reversemarginpar \todo[nolist]{{\tiny sr: #1}}}}

\newcommand{\xmark}{\ding{55}}%

\usepackage[vlined,commentsnumbered,ruled]{algorithm2e}
\let\oldnl\nl
\newcommand{\nonl}{\renewcommand{\nl}{\let\nl\oldnl}}

\SetCommentSty{mycommfont}

\def\HiLi{\leavevmode\rlap{\hbox to \hsize{\color{red!20}\leaders\hrule height .8\baselineskip depth .5ex\hfill}}}

\newcommand\independent{\protect\mathpalette{\protect\independenT}{\perp}}
\def\independenT#1#2{\mathrel{\rlap{$#1#2$}\mkern2mu{#1#2}}}

\usepackage{amsfonts}
\usepackage{amsmath}
\usepackage{comment}
\usepackage{multirow}
\usepackage{paralist}
\usepackage{paralist}
\usepackage{mathtools}

\newtheorem{definition}{Definition}[section]

\newtheorem{example}{Example}[section]

\newcommand{\blue}[1]{{\color{blue} {#1}}}

\newcommand{\red}[1]{{\color{red} {#1}}}

\newcommand{\magenta}[1]{{\color{magenta} {#1}}}

\newcommand{\cut}[1]{}

\newcommand{\cov}{\textsc{Cov}}
\newcommand{\probName}{\textsc{ExPO}}
\newcommand{\algoName}{\textsc{CauSumX}}
\newcommand{\sysName}{\textsc{CauSumX}} 
\newcommand{\Aagg}{\ensuremath{A_{avg}}}

\newcommand{\attrset}{\ensuremath{\mathbb{A}}}
\newcommand{\attrsubset}{\ensuremath{\mathcal{A}_{gb}}}
\newcommand{\dom}{{\tt dom}}

\newcommand{\doop}{{\tt do}}
\newcommand{\db}{\ensuremath{D}}
\newcommand{\Qagg}{\ensuremath{Q}}

\newcommand{\causalG}{\ensuremath{G}}
\newcommand{\exo}{\ensuremath{\mathcal{E}}}
\newcommand{\pattern}{\ensuremath{\mathcal{P}}}
\newcommand{\model}{\ensuremath{\mathcal{M}_\attrset}}
\newcommand{\edvar}{\ensuremath{\mathcal{N}}}
\newcommand{\bruteforce}{\textsf{Brute-Force}}
\newcommand{\ids}{\textsf{IDS}}
\newcommand{\frl}{\textsf{FRL}}
\newcommand{\bruteforcelp}{\textsf{Brute-Force-LP}}
\newcommand{\greedy}{\textsf{Greedy-Last-Step}}
\newcommand{\exptable}{\textsf{Explanation-Table}}
\newcommand{\exptableG}{\textsf{Explanation-Table-G}}
\newcommand{\xinsight}{\textsf{XInsight}}

\newcommand{\german}{\textsf{German}}
\newcommand{\so}{\textsf{SO}}
\newcommand{\adult}{\textsf{Adult}}
\newcommand{\impus}{\textsf{IMPUS-CPS}}
\newcommand{\accidents}{\textsf{Accidents}}

\usepackage{xcolor}
\definecolor{moonstoneblue}{rgb}{0.45, 0.66, 0.76}
\definecolor{oldlace}{rgb}{0.99, 0.96, 0.9}
\definecolor{mintcream}{rgb}{0.96, 1.0, 0.98}
\definecolor{mintgreen}{rgb}{0.6, 1.0, 0.6}
\definecolor{mistyrose}{rgb}{1.0, 0.89, 0.88}
\definecolor{palegold}{rgb}{0.9, 0.75, 0.54}
\definecolor{palechestnut}{rgb}{0.87, 0.68, 0.69}

\newcommand{\reva}[1]{{\leavevmode\color{black}{#1}}}
\newcommand{\revb}[1]{{\leavevmode\color{black}{#1}}}
\newcommand{\revc}[1]{{\leavevmode\color{black}{#1}}}
\newcommand{\common}[1]{{\leavevmode\color{black}{#1}}}

\def\HiLiG{\leavevmode\rlap{\hbox to \hsize{\color{green!30}\leaders\hrule height .8\baselineskip depth .5ex\hfill}}}
\def\HiLiY{\leavevmode\rlap{\hbox to \hsize{\color{yellow!50}\leaders\hrule height .8\baselineskip depth .5ex\hfill}}}
\usepackage{wrapfig}
\definecolor{light-gray}{gray}{0.95}
\usepackage{tcolorbox}

\def\independenT#1#2{\mathrel{\rlap{$#1#2$}\mkern2mu{#1#2}}}

\renewcommand\footnotetextcopyrightpermission[1]{} 
\begin{document}

\title{Summarized Causal Explanations For Aggregate Views}

\author{Brit Youngmann}
\email{brity@technion.ac.il}
\affiliation{%
  \institution{Technion - Israel Institute of Technology}
  \country{Israel}
}

\author{Michael Cafarella}
\affiliation{%
  \institution{CSAIL MIT}
  \country{USA}}
\email{michjc@csail.mit.edu}

\author{Amir Gilad}
\affiliation{%
  \institution{Hebrew University}
  \country{Israel}
}
\email{amirg@cs.huji.ac.il}

\author{Sudeepa Roy}
\affiliation{%
 \institution{Duke University}
 \country{USA}}
 \email{sudeepa@cs.duke.edu}





\renewcommand{\shortauthors}{Brit Youngmann, Michael Cafarella, Amir Gilad, \& Sudeepa Roy}

\begin{abstract}
SQL queries with group-by and average are frequently used and plotted as bar charts in several data analysis applications. 
Understanding the reasons behind the results in such an aggregate view may be a highly non-trivial and time-consuming task, especially for large datasets with multiple attributes. 
Hence, generating automated explanations for aggregate views can allow users to gain better insights into the 
results 
while saving time in data analysis. 
When providing explanations for such views, it is paramount to ensure that they are succinct yet comprehensive, reveal different types of insights that hold for different aggregate answers in the view, and, most importantly, they reflect reality and arm users 
to make informed data-driven decisions, i.e., the explanations do not only consider correlations but are {\em causal}. 
In this paper, we present \sysName, a framework for generating summarized causal explanations for the entire aggregate view. 
Using background knowledge captured in a causal DAG, \sysName\ finds the most effective causal treatments for different groups in the view. 
We formally define the framework and the optimization problem, study its complexity,
and devise an efficient algorithm using the Apriori algorithm, LP rounding, and several optimizations.
We experimentally show that our system generates useful summarized causal explanations compared to prior work 
and scales well for large high-dimensional data. 
\end{abstract}

\maketitle

\section{Introduction}
As database interactions grow in popularity and their user base broadens to data analysts and decision-makers with varied backgrounds, it becomes important to generate insightful and automated explanations for results of the queries users run on the data. One simple yet important class of queries used in data analysis is the class of {\em SQL queries with group-by and average}, which are frequently used and plotted as barcharts in data analysis applications. These queries show how the average varies in different sub-populations in the data by creating an {\em aggregate view} over the input database
 (e.g., average salary per country, occupation, race, or gender;  average severity of car accidents per major city in the USA, etc.). Understanding the {\em causal reasons} behind the high/low values of the average in different groups for such queries can enable sound data-driven decision-making to address unwarranted situations. For instance, if a policymaker knows the possible causal reasons behind lower average salary of a certain race or gender in a certain region in the USA, they can try to improve the situation with corrective measures, which may not be possible with insights that are based on non-causal associational factors. Here we give a running example that we will frequently use in the paper. 
 

\begin{example}
\label{ex:running_example}
Consider the Stack Overflow annual developer survey \cite{stackoverflowreport}, 
where respondents from around the world answer questions about their job.
We consider a subset of the data with 38090 tuples from 20 countries and 5 continents appearing the most in the dataset and augmented the data with additional attributes that describe the economy of each country: {\tt HDI} (Human Development Index, higher values mean more human development), {\tt Gini} (measures income inequity, higher values imply more inequity),  {\tt GDP} (Gross Domestic Product per capita, a measure of country's economic health, higher is better). 
\cref{tab:data} shows a few sample tuples with a subset of the attributes. 
The other attributes are {\tt SexualOrientation}, {\tt EducationParents}, {\tt Dependents}, {\tt Student}, {\tt Hobby}, {\tt HoursComputer}, and {\tt Exercise}.
Now consider the following group-by SQL query measuring the average salary in different countries:
%
%
\begin{center}
\small
    \begin{tabular}{l}
         \verb"SELECT Country, AVG(Salary)"\\
         \verb"FROM Stack-Overflow"\\         
         \verb"GROUP BY Country"\\
        
    \end{tabular}
\end{center}

The results are plotted as a barchart in Figure \ref{fig:so-barchart} 
(the colors will be explained later).
There is a huge variation in average salary (converted to USD) in different countries. 
The user may wonder (i) what are the main factors for this variation across countries, and also (ii) within each country, 
what is {\bf causing} developers to earn more or less. However, the dataset is too big, both in terms of the number of tuples and attributes, to look for a succinct yet informative explanation by manual inspection. While tools like Tableau give highly sophisticated visualizations by slicing and dicing the data across several dimensions, they return the aggregates for these dimensions and do not differentiate between causal and non-causal reasons behind Figure~\ref{fig:so-barchart}. 
\end{example}

Understanding the importance of generating insightful explanations for aggregated query results, several approaches have been proposed in database research on {\em explanations for aggregated query answers}. A simple form is given by the {\em provenance} for aggregate query answers that show how the output was computed using the input tuples \cite{DBLP:conf/pods/AmsterdamerDT11}. However, an aggregate answer over a large dataset uses many input tuples, hence several approaches have focused on providing high-level explanations as {\em predicates} on input tuples that are responsible for producing query answers of interest \cite{wu2013scorpion, roy2014formal, li2021putting} or provide other types of insights explaining them (e.g., the counterbalance approach in \cite{miao2019going}). While these approaches  provide predicates as explanations, which are easy to comprehend, they aim to explain certain answers in the view (e.g., outliers \cite{wu2013scorpion}, high/low values of an answer \cite{miao2019going, li2021putting}, or comparisons of a set of answers \cite{roy2014formal}), and do not provide a summarized explanation for the entire view. Further, although the explanations returned by these approaches reveal many interesting insights, they are not causal. 
\par
{\em Causal inference}, nevertheless, has been studied for several decades in Artificial Intelligence (AI) by 
{\em Pearl's Graphical Causal Model} \cite{pearl2009causal}, and in  Statistics by {\em Rubin's Potential Outcome Framework} \cite{rubin2005causal}. Causal analysis is a vital tool in determining the effect of a treatment on an outcome, and has been used in decision-making in medicine \cite{robins2000marginal}, economics \cite{banerjee2011poor}, biology \cite{shipley2016cause}, and in critical applications like understanding the efficacy of a new vaccine using {\em randomized controlled trials}. While randomized trials cannot be performed in many applications due to ethical or feasibility issues, fortunately, the above causal models provide ways to do sound causal analysis on {\em observed}  datasets under some assumptions (ref. Section~\ref{sec:prelim}).  
\par
Recent works have introduced causality to the field of database research \cite{salimi2018bias,galhotra2022hyper,SalimiPKGRS20,youngmann2022explaining, abs-2207-12718}, allowing users to benefit from this well-founded approach and infer solid causal conclusions from their data and queries. In particular, there has been prior work on extending Pearl's causal model for relational databases \cite{SalimiPKGRS20}, providing causal hypothetical reasoning for what-if and how-to queries \cite{galhotra2022hyper}, and providing explanations for aggregate queries using causal analysis that focused on revealing unobserved factors 
 influencing the results 
\cite{salimi2018bias,youngmann2022explaining}; however, \cite{salimi2018bias,youngmann2022explaining} provide a single explanation of the entire view, and 
do not offer fine-grained explanations for individual groups. A recent work \cite{abs-2207-12718} introduces a framework that searches for predicates that explain the difference in two average outcomes, and marks the patterns as either causal or not, by proposing a new causal discovery algorithm that extends the PC algorithm \cite{spirtes2000causation}. However, they do not search for important treatments affecting the outcomes and do not give causal explanations summarizing the entire view. On the other hand, {\em summarization} techniques for data and query answers form another active topic in database research \cite{el2014interpretable,bu2005mdl,lakshmanan2002generalized,wen2018interactive,sathe2001intelligent,lakshmanan2002quotient,basu2010constructing,DBLP:journals/pvldb/YoungmannAP22}, often with 
{\em diversity} and {\em coverage} factors, however these summaries are also not causal. 


\cut{
One useful way to understand aggregate views is by {\em causal explanations} 
for high and low values in the view. 
However, even obtaining a causal explanation for a single data item can be challenging; aggregate views go further, presenting data items that might be affected by causes in different ways, thereby making a coherent causal story difficult to discern.
}




\smallskip
\noindent
\textbf{Our contributions.~} In this work, {\em we present a novel framework called \sysName\ ({\em \underline{Cau}sal \underline{Sum}marized E\underline{X}planations}) to explain the entire aggregate view} from a query with group-by-average. 
Given a database $D$, causal background knowledge in the form of a causal DAG by Pearl's graphical causal model \cite{pearl2009causal}, a group-by-average query $Q$, 
and parameters $k$ and $\theta$, \sysName\ generates a set of $k$ 
{\em explanation patterns} (predicates) that explain at least $\theta$ fraction of groups in $Q(D)$. 
An explanation pattern contains a {\em grouping pattern} capturing a subset of output groups covered by the explanation, and a {\em treatment pattern} with a high or low value of {\em conditional average treatment effect (CATE)} (ref. Section~\ref{sec:prelim}) on the average attribute $\Aagg$ as the outcome. In standard causal analysis, the goal is to estimate the causal effect of a given treatment on a given outcome, whereas in \sysName\ we search for treatments with high and low causal effects for different subsets of groups defined by the grouping pattern. \sysName\ combines the features of (i) causal inference, (ii) explanation, and (iii) summarization to provide succinct yet comprehensive 
 and causal explanations for group-by-average queries, 
 helping save time and effort in data analysis.


\begin{table*}[]
\footnotesize
    \centering
    	\caption{\textnormal{A subset of the Stack Overflow dataset.}}
         \label{tab:data}
  			\begin{tabular}[b]{|l|l|l|l|c|l|l|l|c|}
  			
				\hline

				\textbf{ID}& \textbf{Country}& \textbf{Continent} &\textbf{Gender} 
    
    &
				\textbf{Age} &\textbf{Role} &
				 \textbf{Education} 
     
    &\textbf{Major}&\textbf{Salary}
				\\ \hline

				1 & US &N. America&Male&26&Data Scientist & PhD& C.S&180k\\
    
    		2 & US &N. America&Non-binary&32&QA developer & B.Sc.& Mech. Eng.&83k\\



 3 & India &Asia&Male&29&C-suite executive  & B.Sc. & C.S&24k\\

  4 & India &Asia&Female&25&Back-end developer  & M.S. & Math.&7.5k\\

  5 & China &Asia&Male&21&Back-end developer & B.Sc. & C.S&19k\\
  

    \hline
			\end{tabular}
\end{table*}

\begin{figure}[t]
\begin{center}
		\includegraphics[scale = 0.4]{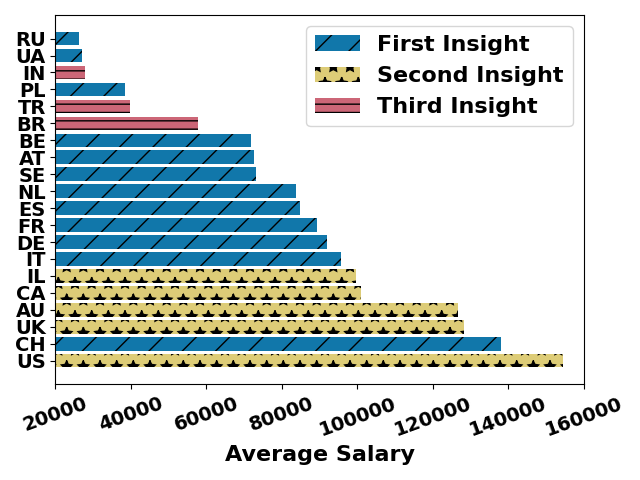}
		\caption{A visualization of the Stack Overflow query results.} 
			\label{fig:so-barchart}
\end{center}

\end{figure}

 \begin{figure}[t]
        \centering
        \begin{minipage}[b]{1.0\linewidth}
            \small
            \begin{tcolorbox}[colback=white]
            \vspace{-2mm}
\textsf{$\bullet$ \magenta{\underline{For countries in Europe}} (blue bars with '/'), the most substantial effect on high salaries (effect size of 36K, $p {<}$ 1e-3) is observed for \textcolor{blue}{\underline{individuals under $35$ with a Master's degree}}. Conversely, \textcolor{red}{\underline{being a student}} has the greatest adverse impact on annual income (effect size: -39K, $p {<}$ 1e-3 ).}\\
                \textsf{$\bullet$ \magenta{\underline{ For countries with a high GDP level}} (yellow bars with '*'), the most substantial effect on high salaries (effect size of 41K, $p {<}$ 1e-3 ) is observed for \textcolor{blue}{\underline{C-level executives}}. Conversely, \textcolor{red}{\underline{being over $55$ with a bachelor's degree}} has the greatest adverse impact on annual income (effect size: -35K,$p {<}$ 1e-4 ).}\\
                \textsf{$\bullet$ \magenta{\underline{For countries with a high Gini coefficient}} (pink bars with '-'), the most substantial effect on high salaries (effect size of 29K, $p {<}$ 1e-4) is observed for \textcolor{blue}{\underline{white individuals under $45$}}. Conversely, being \textcolor{red}{\underline{having no formal degree}} has the greatest adverse impact on annual income (effect size: -28K, $p {<}$ 1e-3).}
            \vspace{-2mm}
            \end{tcolorbox}
        \end{minipage}
        \caption{Causal explanation summary by \sysName.}
        \label{fig:so-explanation}
    \end{figure}

\begin{example}\label{ex:motivating-results}

Reconsider the dataset and query from \cref{ex:running_example}. The user runs \sysName\  to search for an explanation for her query with no more than three insights while covering all groups, and receives the answers shown in ~\cref{fig:so-explanation}. The mapping between countries and insights is visualized using the bars' color and texture in \cref{fig:so-barchart}. Each country can be mapped to more than one insight, but for simplicity, only one color/texture is visualized.  
\sysName\ uses a \emph{causal DAG} (a partial DAG is shown in \cref{fig:causal_dag}), explores multiple patterns, 
and evaluates their causal effect on the salary across different countries. 
There are three parts in each insight: (a) A {\bf grouping pattern}  (first underlined text in \magenta{magenta}), illustrates a property or predicate on the groups or {\tt countries} (group-by attribute in the query) for which this insight holds. (b) A {\bf positive treatment pattern} (second underlined text in \blue{blue}) 
is a predicate on the individuals from the above groups 
with a high positive treatment effect. 
(c) A {\bf negative treatment pattern} (third underlined text in \red{red}) is a predicate on the individuals from the above groups with a high negative treatment effect.
\par
Without having to manually explore this large dataset 
by running many subsequent queries, the user learns the main reasons for high and low salaries in different countries. 
These reasons are not just predicates summarizing tuples in the dataset, they have high and low {\em causal effects} as determined by the causal model. 
Moreover, the user knows that these explanations not only hold for one country but hold for several countries that share the same 
grouping pattern.
The user can continue the exploration by varying parameters in \sysName.
\end{example}

\noindent
Our main contributions are as follows. \\
{\bf (1)} We {\bf develop a framework called \sysName}
    that generates a summarized causal explanation to explain an aggregate view $Q(D)$ for $Q$ with group-by and average. We define explanation patterns that comprise a grouping pattern and a treatment pattern, 
    define an optimization problem to maximize the causal explainability of these explanations subject to a size constraint on the number of explanations and a coverage constraint on the number of output groups covered by them, and show its NP-hardness.
    \par
    \par
{\bf (2)} We {\bf design a three-step algorithm} named \algoName. The first step mines frequent grouping patterns using the seminal Apriori algorithm~\cite{agrawal1994fast}. The second step uses a greedy lattice-based algorithm for mining promising treatment patterns for each grouping pattern from the previous step. In the third step, we model the optimization problem as an Integer Linear Program (ILP) and solve it by randomized rounding of its LP relaxation 
using the grouping and treatment patterns from previous steps.
\par
{\bf (3)} \common{We provide a thorough {\bf experimental analysis} and {\bf multiple case study} that include five datasets, six baselines, and two variations of our solution as additional comparison points. We show that the explanations generated by \sysName\ are of high quality compared to existing approaches and may provide different (and more justifiable) explanations from those given by associational and causal approaches. Additionally, we analyze the runtime and accuracy of the proposed algorithms, and show that \sysName\ is both efficient and useful in providing explanations}.   

\section{Related work}
\label{sec:related}
Table \ref{tab:relatedwork} summarizes the differences between \sysName\ and previous work. Columns in bold highlight our novelty, namely: \sysName\ generates a {\bf summarized} and  \textbf{causal} explanation to the \textbf{entire aggregated view} generated by a SQL query while \textbf{accounting for variations among the groups.} \revc{In contrast, alternative approaches either provide non-causal explanations (as indicated in the Causal column), focus solely on elucidating specific portions of query results (as evident in the Entire View column), or offer a single explanation for query results, neglecting differences among groups (as represented in the Support Groups column).}. 

\begin{table}[t]
\small
\centering
\caption{Positioning of \sysName\ w.r.t. Query Result Explanation, Interpretable Prediction Models, and Data Summarization.}
\label{tab:relatedwork}
\begin{tabular}{|c|c|c|c|c|c|c|c|}
    \hline
 
\multicolumn{2}{|c|}{{Related Work}}  & {Causal} &   {\begin{tabular}{@{}c@{}}Entire\\View\end{tabular}} & {\begin{tabular}{@{}c@{}} Supports\\Groups\end{tabular}} \\
    \hline
\multirow{4}{*}{\begin{tabular}{@{}c@{}}Query Result Explanation\end{tabular}} &\cite{meliou2010complexity, meliou2009so,milo2020contribution,bessa2020effective, miao2019going,roy2014formal,wu2013scorpion,roy2015explaining} &\xmark&\xmark&\xmark\\
&\cite{li2021putting}&{\bf \xmark}&\xmark&\checkmark\\
&\cite{abs-2207-12718}&\textbf{\checkmark}&\xmark&\checkmark\\
&\cite{salimi2018bias, youngmann2022explaining}&\textbf{\checkmark}&\xmark&\xmark\\    \hline
\multirow{2}{*}{\begin{tabular}{@{}c@{}}Interpretable Prediction Models\end{tabular}} 
& \cite{chen2018optimization,lakkaraju2016interpretable}&\xmark&{\bf \checkmark}&\xmark\\
&&&&\\
\hline
\multirow{2}{*}{\begin{tabular}{@{}c@{}}Data Summarization\end{tabular}} &\cite{el2014interpretable} &\xmark&{\bf \checkmark}&\xmark\\
&\cite{wen2018interactive,sathe2001intelligent, DBLP:journals/pvldb/YoungmannAP22}&\xmark&\checkmark&\xmark\\
\hline
\multicolumn{2}{|c|}{\sysName}
\textbf{}&\textbf{\checkmark}&\textbf{\checkmark}&\textbf{\checkmark}\\

    \hline
\end{tabular}
\end{table}

\paratitle{Query Result Explanation}
A substantial body of research has been dedicated to query result explanations. 
Multiple works used \emph{data provenance} to obtain explanations for (possibly missing) query results~ \cite{bidoit2014query,chapman2009not,meliou2010complexity, meliou2009so,DeutchFG20,milo2020contribution,lee2020approximate,ten2015high,li2021putting}. 
Other forms of explanations include (non-causal) interventions \cite{wu2013scorpion,roy2014formal,roy2015explaining,tao2022dpxplain,DBLP:journals/pvldb/DeutchGMMS22}, entropy \cite{Gebaly+2014-expltable}, Shapley values \cite{LivshitsBKS20,ReshefKL20}, and counterbalancing patterns \cite{miao2019going}. 
Those works are orthogonal to our work, as we aim to explain an entire aggregated view via a small set of causal explanations (as explained above). 
Recent works \cite{salimi2018bias, youngmann2022explaining,youngmann2023nexus} propose using causal inference to explain query results. 
\common{In particular, \cite{salimi2018bias} detects bias in a query in the form of confounder variables as explanations for Simpson’s Paradox, while \cite{youngmann2022explaining} finds confounders that explain the correlation between the grouping attribute and AVG queries. In both, the treatment is set as the grouping attribute and is fixed, and the same explanation is provided for all groups. Herein, we mine attribute-value pairs as treatments and search for ones with high causal effects for sets of groups.}
\common{
Another work~\cite{abs-2207-12718} introduced a framework that identifies causal and non-causal patterns to explain the differences between two groups of tuples. 
In Section \ref{subsec:quality_eval} we empirically demonstrate that this framework aims to solve a different task, and hence is unsuitable to solve the problem studied in this work.}

\paratitle{Causal Inference}
There is an extensive body of literature on causal inference over observational data in AI and
Statistics~\cite{greenland1999epidemiology,pearl2009causal,rubin2005causal,tian2000probabilities}. We employ standard
techniques from this literature to compute causal effects.
In related research, estimating heterogeneous treatment effects has been explored \cite{wager2018estimation, xie2012estimating}. This refers to variations in treatment effects across different population subgroups. However, this research differs from our framework. They assume known treatment and outcome variables and focus on identifying subpopulations with varying treatment effects. In contrast, we assume only the outcome variable is given and aim to identify treatments that influence the outcome for each subpopulation, potentially leading to different treatments for each subgroup.

\paratitle{Interpretable Prediction Models}
Previous work developed models that offer both high predictive
accuracy and interoperability~\cite{sagi2021approximating,schielzeth2010simple,lou2013accurate,kim2014bayesian}. 
Rule-based interpretable prediction models~\cite{yang2017scalable,chen2018optimization,lakkaraju2016interpretable} 
often utilize association rule mining processes to produce predictive rules. For instance, \cite{lakkaraju2016interpretable} introduced IDS, a model for binary classification that aims to optimize both the accuracy and interpretability of the chosen rules. This model generates a short, non-overlapping rule set encompassing the entire feature space and classes. 
Similarly, \cite{chen2018optimization} devised FRL, an ordered rule list that includes probabilistic if-then rules. We compare against both \cite{lakkaraju2016interpretable,chen2018optimization} in \cref{sec:experiments}.




\paratitle{Data Summarization}
Data summarization is the process of condensing an input dataset into interpretable and representative subsets~\cite{wen2018interactive,yu2009takes,kim2020summarizing}. A broad spectrum of approaches
have been proposed for data and view summarization~\cite{bu2005mdl,lakshmanan2002generalized,wen2018interactive,sathe2001intelligent,lakshmanan2002quotient,basu2010constructing,DBLP:journals/pvldb/YoungmannAP22}. 
Unlike our work, none of these methods specifically aim to uncover causal explanations for aggregated views. 
In~\cite{el2014interpretable}, the authors use \emph{explanation tables}, which is one of the baselines in \cref{sec:experiments}.

\section{Background on Causal Inference}
\label{sec:prelim}
We use Pearl's model for {\em observational causal analysis} on collected datasets \cite{pearl2009causal} and present the following concepts according to it. 
\par
\paratitle{Causal inference, Treatment, ATE, and CATE}
The broad goal of {\em causal inference} is to estimate the effect of a {\em treatment variable} $T$ on an outcome variable $Y$ (e.g., what is the effect of higher \verb|Education| on \verb|Salary|). The gold standard of causal inference is by doing {\em randomized controlled experiments}, 
where the population is randomly divided into a {\em treated} group that receives the treatment (denoted by $\doop(T = 1)$ for a binary treatment) and the {\em control} group ($\doop(T = 0)$). One popular measure of causal estimate is {\em Average Treatment Effect} (ATE). In a randomized experiment, ATE is the difference in the average outcomes of the treated and control groups~\cite{rubin2005causal, pearl2009causal}
\begin{equation}
    {\small ATE(T,Y) = \mathbb{E}[Y \mid \doop(T=1)] -  
    \mathbb{E}[Y \mid \doop(T=0)]}
\label{eq:ate}
\end{equation}
The above definition assumes that the treatment assigned to one unit does not affect the outcome of another unit (called the {Stable Unit Treatment Value Assumption (SUTVA)) \cite{rubin2005causal}}\footnote{This assumption does not hold for causal inference on multiple tables and even on a single table where tuples depend on each other, which we discuss in Section~\ref{sec:conc}.}. 
\par
In our work on causal explanations for SQL group-by-average queries, where the treatment with maximum effect may vary among different tuples in the query answer, we are interested in computing the \emph{Conditional Average Treatment Effect} (CATE), which measures the effect of a treatment on an outcome on \emph{a subset of input units}~\cite{rubin1971use,holland1986statistics}. 
Given a subset of units defined by (a vector of) attributes $B$ and their values $b$, 
we can compute $CATE(T,Y \mid B = b)$ as:
{
\begin{eqnarray}    
    \mathbb{E}[Y \mid \doop(T=1), B = b] -  
    \mathbb{E}[Y \mid \doop(T=0), B = b]\label{eq:cate}
\end{eqnarray}
}
\par
However, randomized experiments where treatments are assigned at random cannot be done in many practical scenarios due to ethical or feasibility issues 
(e.g., effect of 
higher education on salary). In these scenarios, {\em Observational Causal Analysis} still allows sound causal inference under additional assumptions. Randomization in controlled trials mitigates the effect of {\em confounding factors} or {\em covariates}, i.e., attributes that can affect the treatment assignment and outcome. Suppose we want to understand the causal effect of \verb|Education| on \verb|Salary| from the SO dataset.  
We no longer apply Eq. (\ref{eq:ate}) since the values of \verb|Education| were not assigned at random in this data, and obtaining higher education largely depends on other attributes like \verb|Gender|, \verb|EducationParents|,and \verb|Country|. 
Pearl's model provides ways to account for these confounding attributes $Z$ to get an unbiased causal estimate from observational data under the following assumptions ($\independent$ denotes independence):
\begin{eqnarray}
    \textbf{Unconfoundedness}: & Y \independent T | Z {=} z \label{eq:unconfoundedness}\\
    \textbf{Overlap}: & 0 < Pr(T {=} 1 |Z {=} z)< 1 \label{eq:overlap}
\end{eqnarray}
The unconfoundedness assumption, Eq. (\ref{eq:unconfoundedness}), states that if we condition on $Z$, then treatment $T$ in the dataset and the outcome $Y$ are independent. In SO, assuming that only $Z$ =\{\verb|Gender|, \verb|EducationParents|, \verb|Country|\} affects $T = $ \verb|Education|, if we condition on a fixed set of values of $Z$, i.e., consider people of a given gender, from a given country, and with a given education level of parents, then $T = $ \verb|Education| and $Y = $ \verb|Salary| are independent. For such confounding factors $Z$,  Eq. (\ref{eq:ate}) and (\ref{eq:cate}) respectively reduce to the following form 
(Eq. (\ref{eq:overlap})
gives feasibility of the expectation difference): 
{\small
\begin{flalign}    
    & ATE(T,Y) = \mathbb{E}_Z \left[\mathbb{E}[Y \mid T=1, Z = z] -  
    \mathbb{E}[Y \mid T=0, Z = z] \right] \label{eq:conf-ate}\\
 & CATE(T,Y \mid B = b) = \nonumber
    \mathbb{E}_Z \left[\mathbb{E}[Y \mid T=1, B = b, Z = z] -  
    \mathbb{E}[Y \mid T=0, B = b, Z = z]\right]\label{eq:conf-cate}
\end{flalign}
}
The above equations no longer have the $\doop(T = b)$, and can be estimated from an observed dataset. 
Pearl's model gives a systematic way to find such a $Z$ when a causal DAG is available.


\paratitle{Causal DAG}
Pearl's Probabilistic Graphical Causal Model model \cite{pearl2009causal} can be written as a tuple $(\exo, \edvar, Pr_{\exo}, \psi)$, where $\exo$ is a set of {\em unobserved exogenous (noise)} variables, $\Pr_{\exo}$ is the joint distribution of \exo, and $\mathcal{N}$ is a set of {\em observed endogenous variables}. 
Here $\psi$ is a set of structural equations that encode dependencies among variables. The equation for $A \in \edvar$ takes the following form:
$$\psi_{A}: 
\dom(Pa_{\exo}(A)) {\times} \dom(Pa_{\edvar}(A)) \to \dom(A)$$
Here $Pa_{\exo}(A) {\subseteq} {\exo}$ and $Pa_{\edvar}(A) {\subseteq} \edvar \setminus \{A\}$ respectively denote the exogenous and endogenous parents of $A$. A causal relational model is associated with a {\em causal DAG}, $G$, whose nodes are the endogenous variables $\edvar$ and whose edges are all pairs $(X,Y)$ (directed edges from $X$ to $Y$) such that $Y {\in} \edvar$ and $X {\in} Pa_{\edvar}(Y)$. The causal DAG obfuscates exogenous variables as they are unobserved. Any given set of values for the exogenous variables completely determine the values of the endogenous variables by the structural equations (we do not need any known closed-form expressions of the structural equations in this work). 
The probability distribution $\Pr_{\exo}$ on exogenous variables $\exo$ induces a probability distribution  
on the endogenous variables $\mathcal{N}$ by the structural equations $\psi$.

  Figure \ref{fig:causal_dag} depicts a causal DAG for the SO dataset over the attributes in Table \ref{tab:data} as endogenous variables (we use a larger causal DAG with all 20 attributes for the SO dataset in our experiments). 
  Given this causal DAG, we can observe that the \verb|Role| that a coder has in their company depends on the values of their \verb|Education|, \verb|Age|, \verb|Major|, and \verb|YearsCoding| attributes.

A causal DAG can be constructed by a domain expert as in the above example, or using existing {\em causal discovery} \cite{pearl2009causal} algorithms, which we study further in our experiments (Section~\ref{sec:experiments}). 
\par
In Pearl's model, a treatment $T = t$ (on one or more variables) is considered as an {\em intervention} to a causal DAG by mechanically changing the DAG such that the values of node(s) for $T$ in $G$ are set to the value(s) in $t$, which is denoted by $\doop(T = t)$. Following this operation, the probability distribution of the nodes in the graph changes as the treatment nodes no longer depend on the values of their parents. Pearl's model gives an approach to estimate the new probability distribution by identifying the confounding factors $Z$ described earlier using conditions such as {\em d-separation} and {\em backdoor criteria} \cite{pearl2009causal}, which we do not discuss in this paper. 

\section{Framework for Summarized Causal Explanations for Aggregate Queries}\label{sec:franework}

\begin{figure}
\centering
\includegraphics[scale = 0.2]{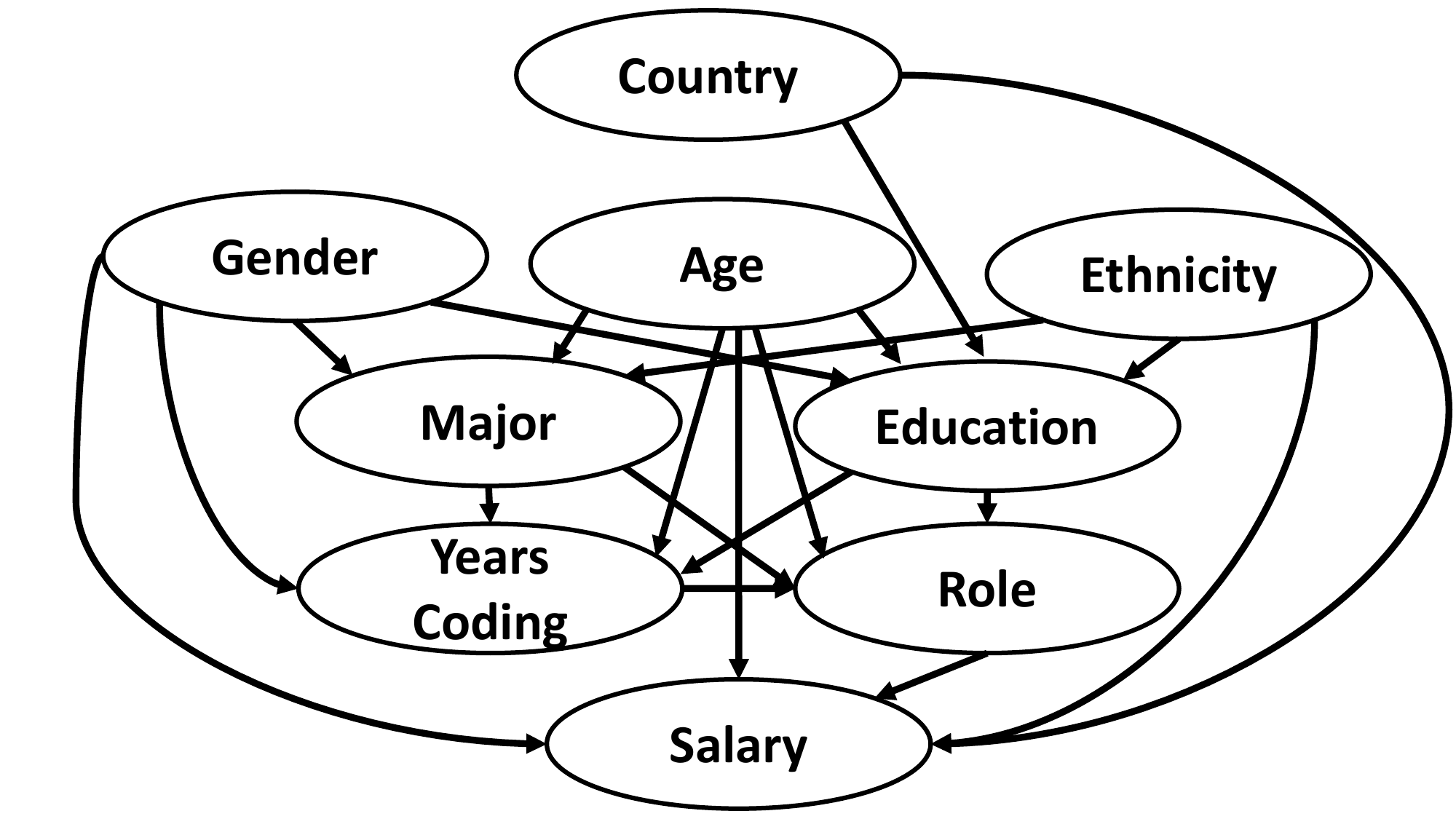}
\caption{Example causal DAG.} \label{fig:causal_dag}
\end{figure}

\paratitle{Databases and Queries}
We consider a single-relation database\footnote{We discuss adjustments for supporting multi-dimensional datasets in \cref{sec:conc}.} over a schema $\attrset$. The schema is a vector of attribute names, i.e., $\attrset {=} (A_1, \ldots, A_s)$, where each $A_i$ is associated with a domain $\dom(A_i)$, which can be categorical or continuous. 
A database instance \db, populates the schema with a set of tuples $t {=} (a_1, \ldots, a_s)$ where $a_i {\in} \dom(A_i)$. We use $t[A_i]$ to denote the value of attribute $A_i$ of tuple $t$.  
In this paper, we will only consider the {\em active domain} of every $A_i$ as $\dom(A_i)$, i.e., the set of values of $A_i$ in the given $\db$.

We consider an important class of SQL queries for data analysis, with group-by and average as the aggregate function: 
\begin{equation*}
\small
\Qagg = \texttt{SELECT } \attrsubset, \texttt{ AVG }(\Aagg) \texttt{ FROM D WHERE } \phi \texttt{ GROUP BY } \attrsubset;
\end{equation*}
Here, $\attrsubset {\subseteq} \attrset$ is a set of categorical group-by attributes, $\Aagg {\in} \attrset$ is the average attribute, and $\phi$ is a predicate. 
The result of evaluating $\Qagg$ over \db\ is denoted by $\Qagg(\db)$. We denote $|\Qagg(\db)|$, i.e., the number of groups in $Q(D)$, by $m$. The {\tt WHERE} condition $\phi$ simply reduces the table $D$ to tuples satisfying $\phi$ before the techniques in this section can be applied, so we do not discuss $\phi$ further in this section.

Consider again the query presented in Example \ref{ex:running_example}. 
Here $\attrsubset {=} $  \verb|Country|, $\Aagg {=} $ \verb|Salary|, and $\phi$ is the empty predicate. The 
query results $Q(D)$ are shown in Figure~\ref{fig:so-barchart}, where $m {=} |Q(D)| {=} 20$. 
\subsection{Explanation Patterns}
The {\em patterns} in our framework of summarized explanations are {\em conjunctive predicates} on attribute values that are prevalent in previous work on explanation, e.g.,  \cite{roy2014formal,wu2013scorpion,Gebaly+2014-expltable}.
\begin{definition}[Pattern]\label{def:pattern}
Given a database instance \db\ with schema \attrset,
a simple predicate is an expression of the form $\varphi {=} A_i ~{\tt op }~ a_i$, 
where $A_i {\in} \attrset$, 
$a_i {\in} \dom(A_i)$, and ${\tt op} {\in} \{=, <, >, \leq, \geq\}$.  
A {\em pattern} is a conjunction of simple  predicates $\pattern {=} \varphi_1 \land \ldots \land \varphi_k$.
\end{definition}



\paratitle{Grouping and treatment patterns}
Our explanations consist of pairs of patterns: 
(i) A {\bf grouping pattern} $\pattern_g$ 
captures a subset of groups in $\Qagg(\db)$ and must be well-defined over $\Qagg(\db)$, i.e., each query answer $s\in Q(D)$ is either \emph{covered} or not by $\pattern_g$. Therefore, $\pattern_g$ can only contain attributes $W$ such that the Functional Dependency (FD) from the grouping attributes, 
$\attrsubset \rightarrow W$ holds for all $W$ in $\pattern_g$. 
(ii) A {\bf treatment pattern}, $\pattern_t$, is defined over the dataset $\db$ (as opposed to $\pattern_g$ that is defined over $\Qagg(\db)$) and partitions the input tuples into treated ($T = 1$ if $\pattern_t$ evaluates to true for a tuple) and control groups ($T = 0$ if $\pattern_t$ evaluates to false). This partition is then used to assess the causal effects of the treatment pattern on the outcome $Y = \Aagg$, the attribute for average in the query $Q$.
\par
A pair of grouping and treatment pattern $(\pattern_g, \pattern_t)$ together define an {\bf explanation pattern}.  
Intuitively, $\pattern_g$ specifies the subpopulation of interest (equivalent to the condition $B = b$ in Eq. (\ref{eq:cate}) for CATE), while $\pattern_t$ (equivalent to treatment $T$) explains the observed outcome $Y$ within that subpopulation as per the CATE value. 

\begin{definition}[Explanation Pattern]\label{def:explanation-pattern}
Given a database instance \db\ with schema \attrset\ and a query $\Qagg$ with group-by attributes \attrsubset\ and attribute \Aagg for average, 
an \emph{explanation pattern} $(\pattern_g, \pattern_t)$ where $\pattern_g$ is a grouping pattern on $Q(D)$, i.e., the FD $\attrsubset \rightarrow W$ holds in \db for all attributes $W$ in $\pattern_g$, 
and $\pattern_t$ is a pattern defined over $\db$. 
\end{definition}
\begin{example}\label{ex:pattern}
In the first insight in Figure~\ref{fig:so-explanation}, 
one explanation pattern is $(\pattern_g, \pattern_t)$ with $\pattern_g:$ 
\newline$({\tt Continent = Europe})$ and $\pattern_t:$ $({\tt Age} < 35) \wedge$ ${\tt(Education = Master's~degree)}$. Note that the FD from the group-by attribute ${\tt Country} \rightarrow {\tt Continent}$ holds. 
\end{example}
\par
\paratitle{Partitioning attributes for grouping and treatment patterns} We partition the attributes in $\attrset$ into two disjoint sets. All attributes $W {\subseteq} \attrset$ s.t the FD $\attrsubset {\rightarrow} W$ holds in $D$ are considered for grouping patterns. All other attributes $U$
are considered for treatment patterns. The necessity for FD for grouping patterns is explained above. 
Further, any attribute $W$ where $\attrsubset {\rightarrow} W$ holds, and in general, any grouping pattern, cannot be a valid treatment pattern. By the overlap condition in Eq.~(\ref{eq:overlap}), conditioned on a grouping pattern $\pattern_g$, there should be at least one unit with $T {=} 1$ and at least one unit with $T {=} 0$. If we pick a pattern $W {=} w$ with $\attrsubset {\rightarrow} W$ or a pattern with multiple such attributes as the treatment, for all tuples in $D$ contributing to a query answer in $Q(D)$, either it evaluates to true or to false, so we do not get both $T {=} 1$ and $T {=} 0$ to estimate the CATE value.

\paratitle{Explainability} To evaluate the effectiveness of an explanation pattern $(\pattern_g, \pattern_t)$, we define its explainability.

\begin{definition}[Explainability]\label{def:explainability}
Given a database instance \db\ with schema \attrset, a query $\Qagg$ group-by attributes \attrsubset and attribute \Aagg, and a causal model \model\ on \attrset associated with a causal DAG, 
the explainability of an explanation pattern $( \pattern_g, \pattern_t)$ is defined as:
%
$$Explainability(\pattern_g, \pattern_t) {:=} CATE_{\model}(\pattern_t, \Aagg~|~\pattern_g)$$
using the definition of CATE as in Eq.~(\ref{eq:cate}) with treatment $T {=} \pattern_t$, outcome $Y {=} \Aagg$, and subpopultion 
defined by $\pattern_g$. The subscript $\model$ denotes that the CATE is estimated using the causal model $\model$ as explained in Section~\ref{sec:prelim}.
\end{definition}
In particular, CATE given by Eq.~(\ref{eq:cate}) in the above definition is reduced to Eq.~(\ref{eq:conf-ate}) using confounding variables $Z$ obtained from the causal DAG of \model, which then can be estimated from the data $D$. 
Our focus is on computing CATE (Eq.~\ref{eq:cate}) rather than ATE (Eq.~\ref{eq:ate}) to 
understand the causal factors that influence outcomes for groups in $Q(D)$. 
For instance, the effect of having a {\tt Master's degree} on \verb|Salary| in countries in Europe can be different from that in countries in Asia. Therefore, computing the ATE while considering all individuals from across the globe may not provide meaningful insights.
For $\pattern_g: ({\tt Continent = Europe})$ and $\pattern_t: {\tt (Education =}$ ${\tt MA~degree)}$, we define the treatment group as individuals with a master's degree from European countries and the control group as individuals without a Masters degree from European countries. This enables us to draw relevant conclusions. 

\begin{example}
    In Example~\ref{ex:motivating-results} and Figure~\ref{fig:so-explanation}, there are two explanation patterns using the same grouping pattern. Here $\pattern_g: ({\tt Continent =}$ ${Europe})$, $\pattern_{t1}: ({\tt Age < 35) \wedge (Education = Master's~degree)}$, and $\pattern_{t2}: (Student = yes)$. The explainability of $(\pattern_g, \pattern_{t1}) = 36K$ and that of $(\pattern_g, \pattern_{t2}) = -39K$, indicating that age below 35 and having a Master's degree has a high positive causal effect for individuals from European countries while being a student has a high negative effect. 
\end{example}

\subsection{Problem Definition and Hardness}
\label{subsec:problem}

\par
Our goal is to obtain a succinct yet comprehensive set of explanation patterns for the groups in $Q(D)$ from the huge search space of possible explanation patterns (e.g., the number of explanation patterns for Example~\ref{ex:running_example} was 22350). 
To achieve this, we frame a constrained optimization problem. 
We apply three constraints: (1) the number of explanation patterns should not exceed a specified threshold, (2) the number of groups in $Q(D)$ explained by the patterns must be at least a specified $\theta$-fraction of all the groups in $Q(D)$, and (3) an explanation pattern should not explain the same set of groups explained by another explanation pattern.
Finally, our goal is to find the set of explanation patterns that abide by these constraints and whose overall explainability is maximized. First, we define the {\em coverage} of a grouping pattern $\pattern_g$.
\begin{definition}[Coverage]
Given a database instance \db\ and a query $\Qagg$ with group-by attributes \attrsubset, a grouping pattern $\pattern_g$, and a group $s {\in} Q(D)$, $\pattern_g$ is said to {\bf cover} $s$ if for any tuple $t {\in} D$ such that $t[\attrsubset] = s[\attrsubset]$, it holds that $t \models \pattern_g$, i.e., $t$ satisfies the predicate $\pattern_g$. 
The set of groups in $Q(D)$ covered by $\pattern_g$ is denoted by 
$\cov(\pattern_g)$.
\end{definition}

Next, we define the problem of {\em Summarized Causal Explanations}.
\begin{definition}[Summarized Causal Explanations]
\label{def:problem}
Given a database $\db$, a causal model \model, a query $\Qagg$ with group-by attributes \attrsubset\ and attribute \Aagg for average, a collection of explanation patterns $\{\pattern_i\}_{i=1}^l$, an integer $k {\in} [1, m]$ where $m = |Q(D)|$, and a threshold $\theta {\in} [0,1]$, 
we aim to find a set $\Phi {\subseteq} \{\pattern_i\}_{i=1}^l$ 
of explanation patterns 
such that the following conditions hold:
\begin{itemize}
    \item \textbf{(Size constraint)} $|\Phi| \leq k$.
\item \textbf{(Coverage constraint)} at least $\theta {\cdot} m$ groups from $\Qagg(\db)$ are covered by 
    $\Phi$, i.e., 
    $\cup_{\pattern_g \in \Phi} \cov(\pattern_g) \geq \theta \times m$. 
\item \textbf{(Incomparability constraint)} There are no pairs of explanation patterns, $(\pattern_g, \pattern_t)$ and $(\pattern_g', \pattern_t')$ in $\Phi$ such that $\cov(\pattern_g) = \cov(\pattern_g')$.
\end{itemize}
\textbf{(Objective)} The objective is to maximize the total explainability of $\Phi$ under the above constraints, i.e., maximize $\sum_{\pattern \in \Phi}Explainability(\pattern)$.
\end{definition}

\revb{The size and incomparability constraints ensure that the size of the explanation will be small with no redundancy and therefore it will be more easily grasped by the user. The coverage constraint ensures that the explanation will be extensive and capture a significant part of the view, and the objective aims to maximize the validity of the explanation as the cause of the trends in the view.}

\begin{example}
In Example~\ref{ex:motivating-results} and Figure~\ref{fig:so-explanation}, the parameters $k {=} 3$ and $\theta {=} 1$ are used, i.e., we aim to find a set of at most $3$ explanation patterns that reveal the causes of the outcome for all groups in $\Qagg(\db)$. As Proposition~\ref{prop:hardness} shows, even deciding whether there is a set $\Phi$ of explanation patterns for a given size constraint $k$ and coverage constraint $\theta$ simultaneously is NP-Hard, although for this example, we find a solution 
covering all 20 countries in $Q(D)$ as shown in Figure~\ref{fig:so-barchart}.
\cut{
The candidate grouping patterns include patterns based on the group-by attribute \verb|country| (e.g., \verb|{country = India}|) and those based on the \verb|continent, HDI, GDP, GINI| attributes that depend on \verb|country| (e.g., \verb|{Continet = North America}|). We first must find the treatment pattern with the highest explainability for each grouping pattern, resulting in the space of possible explanation patterns. Here, the number of possible grouping patterns is $75$, and the maximum number of treatment patterns per grouping pattern is $298$.
}
\end{example}

\paratitle{Positive and negative explanation patterns}
\revb{Given an outcome variable (e.g., income), positive explanations correspond to treatments that have an impact on making its value higher (what increases income), and negative explanations are treatments that make its value lower (what reduces income). These positive/negative treatments can vary for different groups (as illustrated in Figure~\ref{fig:so-explanation}). In an application, both or one of them may be valuable. Consequently, our proposed framework supports both. In a prototype with a UI, analysts have the flexibility to choose whether they want to view one or both types of explanations and even top-k positive/negative treatments for a grouping pattern.
This helps understand the cause of both high and low values of outcomes for different groups without explicitly asking for explanations for high and low values as done in previous work \cite{wu2013scorpion, roy2014formal, miao2019going, li2021putting}.}

To generate positive and negative patterns for a grouping pattern, we slightly vary the optimization objective in Definition~\ref{def:problem}. 
For a 
grouping pattern $\pattern_g$, we find a treatment pattern $\pattern^+_{g, tp}$ with the {\em highest explainability value of $(\pattern_g, \pattern^+_{g, t})$} ($\pattern^+_{g, t}$ is called a {\em positive treatment pattern} for $\pattern_g$). A {\em negative treatment pattern} $\pattern^-_{g, t}$ is defined similarly using the lowest explainability value. 
In our system, for each grouping pattern $\pattern_g$ we compute the sum of absolute values of two explainabilities: 
$|Explainability(\pattern_g, \pattern^+_{g, t})| + |Explainability(\pattern_g, \pattern^-_{g, t})|$.
We treat this sum as the weight of the explanation pattern combination $(\pattern_g, \pattern^+_{g, t}, \pattern^-_{g, t})$ and return top-$k$ explanations that satisfy the constraints in Definition~\ref{def:problem}.


\paratitle{Hardness Result and enumeration of search space}
Since in our optimization problem, we want to cover a certain fraction of answer tuples in $Q(D)$ (`elements') with at most $k$ patterns (`sets'), we can show the following NP-hardness result even if we ignore the optimization objective (proof in the full version \cite{fullversion}).

\begin{proposition}\label{prop:hardness}
It is NP-hard to decide whether the Summarized Causal Explanations problem is feasible (i.e., has any solution satisfying the constraints) for a given $k$ and $\theta$.
\end{proposition}
Definition~\ref{def:problem} assumes that the search space of explanation patterns is given, while in practice it is not efficient to enumerate all explanation patterns ahead of time. In Section~\ref{sec:algo} we give an efficient algorithm that give good solutions for the optimization problem without explicitly enumerating all explanation patterns upfront.

\section{The \algoName\ Algorithm}
\label{sec:algo}
Proposition~\ref{prop:hardness} 
shows that even deciding the feasibility of the Summarized Causal Explanations problem for a given $k$ and $\theta$ is NP-hard. Further, it is not practical to enumerate all explanation patterns and compute their explainability upfront. Given four-dimensional desiderata in Definition~\ref{def:problem}
(unlike the standard set-cover or max-cover problems that have two), it is non-trivial to design a good approximation algorithm or heuristics for this problem. In this section, we present the \algoName\ algorithm (for {\em \underline{Cau}sal \underline{Sum}marized E\underline{X}planations}) that aims to address these challenges.

\revb{A brute-force approach considers all grouping and treatment patterns and results in long runtimes (as we demonstrate in Section \ref{subsec:exp-runtime}). Instead, \algoName\ employs the Apriori algorithm~\cite{agrawal1994fast} to mine frequent grouping patterns (the coverage of such patterns is monotone) as a heuristic for finding promising grouping patterns that are short (we ignore the rest). To mine treatment patterns (which are non-monotone for CATE) we adopt a lattice traversal approach and use a greedy heuristic materializing only promising treatment patterns. Both steps can lose optimality. In Section \cref{sec:experiments} we compare \algoName\ against \bruteforce\ and show that while \algoName\ is much faster, the difference in results (in terms of explainability, coverage, and accuracy) is modest.}




\vspace{1mm}
\noindent
\textbf{Overview}.
The pseudo-code in Algorithm \ref{algo:full_algo} outlines the operation of \algoName. It takes a database \db, an integer $k$, a query \Qagg\ with an outcome $\Aagg$ and grouping attributes $\attrsubset$, and a coverage threshold $\theta$ as input. The output is a set of explanation patterns $\Phi$. The algorithm consists of three steps: (1) extracting candidate grouping patterns (line \ref{line:init}). 
(2) Focusing on promising treatment patterns for each extracted grouping pattern, \emph{materializing and evaluating only them} (lines \ref{l:iterate-candidates}--\ref{l:top-treatment}). (3) Utilizing Linear Programming (LP) to obtain a set of explanation patterns (line \ref{l:step3}). 





\begin{algorithm}[t]
\footnotesize
\SetKwInOut{Input}{input}\SetKwInOut{Output}{output}
\LinesNumbered
\Input{A database relation \db, an integer $k$, a query \Qagg\ with an aggregate attribute $\Aagg$ and grouping attributes $\attrsubset$, and a coverage threshold $\theta$}
\Output{A set $\Phi$ 
of explanation patterns.} \BlankLine
\SetKwFunction{NextBestExplanation}{\textsc{NextBestExplanation}}
\SetKwFunction{SolveLP}{\textsc{SolveLP}}
\SetKwFunction{GetGroupingPatterns}{\textsc{GetGroupingPatterns}}
\SetKwFunction{ExplanationSummary}{\textsc{ExplanationSummary}}
\SetKwFunction{GetTopTreatment}{\textsc{GetTopTreatment}}
$\Phi \gets \emptyset$\;
$Candidates \gets \GetGroupingPatterns (\db, \Aagg, \attrsubset)$\tcp*{\cref{subsec:grouping_patterns}} \label{line:init}   
 \For{$\pattern_g \in Candidates$}{\label{l:iterate-candidates}
    $\pattern_g.\pattern_t \gets \GetTopTreatment(\pattern_g, k, \Aagg, \db)$\tcp*{\cref{subsec:treatment_patterns}}\label{l:top-treatment}
 }
  $\Phi \gets \SolveLP(Candidates, k, \theta)$\tcp*{\cref{subsec:step_3}}\label{l:step3}

\Return $\Phi$\label{line:returnTop}\\
\caption{The \algoName\ algorithm}\label{algo:full_algo}
\end{algorithm}

\subsection{Mining Grouping Patterns}
\label{subsec:grouping_patterns}
Considering every possible grouping pattern (\cref{def:explanation-pattern}) is infeasible as their number is exponential ($O(agrmax_{A_i {\in} \attrset}|\dom(A_i)|^{|\attrset|}$). 
Instead, our approach utilizes the Apriori algorithm~\cite{agrawal1994fast} to generate candidate patterns. The Apriori algorithm gets a threshold $\tau$, and ensures that
the mined patterns are present in at least $\tau$ tuples of $\db$.
Formally, given a set of attribute 
$W {\subseteq} \attrset$
s.t the FD $\attrsubset {\rightarrow} W$
holds in $\db$, we apply the Apriori to extract frequent patterns defined solely by these attributes. The algorithm guarantees that each mined pattern covers at least $\tau$ tuples from $\db$ and is well-defined over $\Qagg(\db)$, making them promising candidates for covering the necessary number of groups (see item (2) in \cref{def:problem}).





\paratitle{Post-Processing}
Certain extracted grouping patterns may be superfluous, i.e., the set of groups they define from $\Qagg(\db)$ could be indistinguishable.
\revc{Redundant grouping patterns can emerge even in the absence of FDs among attributes. To illustrate, consider the patterns} \verb|{HDI = High}| and \verb|GDP = High|, \revc{which both identify the same set of countries. However, it is possible that two countries with medium HDI values exhibit different GDP values. Thus, there is no functional dependency between} \verb|HDI| and \verb|GDP|. \revc{Nevertheless, two grouping patterns defined by these attributes can still define the same set of tuples, rendering the consideration of both redundant.}
Therefore, following the mining stage, we remove redundant grouping patterns to ensure the obtained solution will satisfy the incomparability constraint (item (3) in \cref{def:problem}). Additionally, we favor more succinct patterns as they are easier to comprehend. To achieve this, we utilize a hash table that records the groups from $\Qagg(\db)$ associated with each mined grouping pattern. In each group set, we retain only the shortest grouping pattern.



\vspace{-2mm}
\subsection{Mining Treatment Patterns}
\label{subsec:treatment_patterns}
As opposed to a standard causal analysis setting where a causal question of the form ``What is the effect of treatment $T$ on outcome $Y$?'' is posed, in our setting, {\bf we aim to find $T$ that yields the highest effect on $Y$.} 
Our subsequent goal, as discussed at the end of Section~\ref{subsec:problem}, is to identify a positive treatment pattern $\pattern^+_{tg}$ and a negative treatment pattern $\pattern^-_{tg}$ for each mined grouping pattern $\pattern_g$, which respectively explain the cause for high and low outcomes of tuples in $D$ that satisfy $\pattern_g$. However, for simplicity and without loss of generality, we 
present the approach for finding positive explanation patterns only denoted by $\pattern_t$ (i.e., treatments that have the highest positive CATE for a grouping pattern $\pattern_g$). 

\revb{Since the number of potential treatment patterns for $\pattern_g$ can be large (exponential in $|\attrset|$), 
we propose a greedy heuristic approach to materialize and assess the CATE only for promising treatment patterns. This is done by leveraging the notion of lattice traversal~\cite{DeutchG19,asudeh2019assessing}. 
In particular, the set of all treatment patterns
can be represented as a lattice where nodes correspond to patterns and there is an edge between $\pattern_t^1$ and $\pattern_t^2$
if $\pattern_t^2$
can be obtained
from $\pattern_t^1$ by adding a single predicate. This lattice can be traversed in a top-down fashion while generating each node at
most once. 
Since not all nodes correspond to treatments that have a positive CATE, we only materialize nodes if all their parents have a positive CATE. 
The primary distinction from existing solutions (e.g., ~\cite{asudeh2019assessing}) lies in the non-monotonic nature of CATE (meaning that adding a predicate to a treatment pattern can either increase or decrease its CATE value). Consequently, exclusively materializing nodes where all of their parents exhibit a positive CATE might lead the algorithm to overlook certain relevant treatment patterns.}
For example, we observe that for the pattern $\pattern_{t1} {=} $\verb|(role = QA)|, by adding the predicate \verb|(education = MA)|, its CATE value increases, while the CATE decreases by adding the predicate \verb|(education = no degree)|.
\revb{
However, according to our experiments (\cref{subsec:quality_eval}), combining treatment patterns that exhibit a positive CATE is highly likely to result in a treatment with a positive CATE as well. As a result, we 
account for most treatments with a positive CATE, and the accuracy of this algorithm is relatively high.}  

Algorithm \ref{algo:treatment_patterns} describes the search for the treatment pattern $\pattern_t^{max}$ with the highest positive CATE value for a given grouping pattern $\pattern_g$. It traverses the
pattern lattice in a top-down manner, starting from patterns with a single atomic predicate in line 2. 
In particular, the function {\tt GenChildren} generates all atomic predicates of the form $A_i~ {\tt op}~ a_j$. 
For each treatment pattern, it evaluates its CATE, 
and discards the ones with CATE values that do not have the same sign $\sigma$, either $+$ or $-$ (line~3). 
The algorithm stores the pattern with the highest positive CATE identified thus far (line~4). It then proceeds to traverse the pattern lattice only for nodes that satisfy the condition of having all their parents with a positive CATE (line~7-8), and updates $\pattern_t^{max}$ if a pattern with a higher CATE was found (lines~9-12). It terminates at the first level, which does not include the maximum value recorded (lines~13-14).


      

\setlength{\textfloatsep}{2px}
\SetInd{1.3ex}{1.3ex}
\begin{algorithm}[t]
\footnotesize
\SetKwInOut{Input}{input}\SetKwInOut{Output}{output}
\LinesNumbered
\Input{A grouping pattern $\pattern_g$, the outcome attribute $\Aagg$, the dataset $\db$, a causal DAG \causalG, and direction $\sigma {\in} \{+, -\}$.}
\Output{A treatment pattern $\pattern_t$.} \BlankLine
    \SetKwFunction{GenChildren}{\textsc{GenChildren}}
    \SetKwFunction{GenChildrenNextLevel}{\textsc{GenChildrenNextLevel}}
     \SetKwFunction{ComputeCATEnFilter}{\textsc{ComputeCATEnFilter}}
          \SetKwFunction{GetTopTreatment}{\textsc{GetTopTreatment}}

\tcc{Get all single-predicate patterns.}
$\boldsymbol{C} \gets \GenChildren(\db, \Aagg, \pattern_g)$\; \label{line:init}  

$\boldsymbol{C} \gets \ComputeCATEnFilter(\boldsymbol{C}, \db, \Aagg, \pattern_g, \causalG)$\; \label{line:filter}  

$\pattern_t^{max} \gets \GetTopTreatment(\boldsymbol{C})$\; \label{line:initR}
      
\While{True}
{\label{line:whileStart}
\tcc{Get patterns in the next level.}
     $\boldsymbol{C} \gets \GenChildrenNextLevel(\boldsymbol{C}, \sigma)$\;

  $\boldsymbol{C} \gets \ComputeCATEnFilter(\boldsymbol{C}, \db, \Aagg, \pattern_g, \causalG)$\;

     $\pattern_t \gets \GetTopTreatment( \boldsymbol{C})$\;
     \If{$\pattern_t.CATE > \pattern_t^{max}.CATE$}{
     $\pattern_t^{max} \gets \pattern_t$\;
     } 
     \Else{
     Break\;
     }
}
\Return $\pattern_t^{max}$\;\label{line:returnTop}
\caption{Top treatment pattern for a grouping pattern 
}\label{algo:treatment_patterns}

\end{algorithm} 


\begin{figure}[t]
\centering
\vspace{-5mm}
\includegraphics[scale = 0.25]{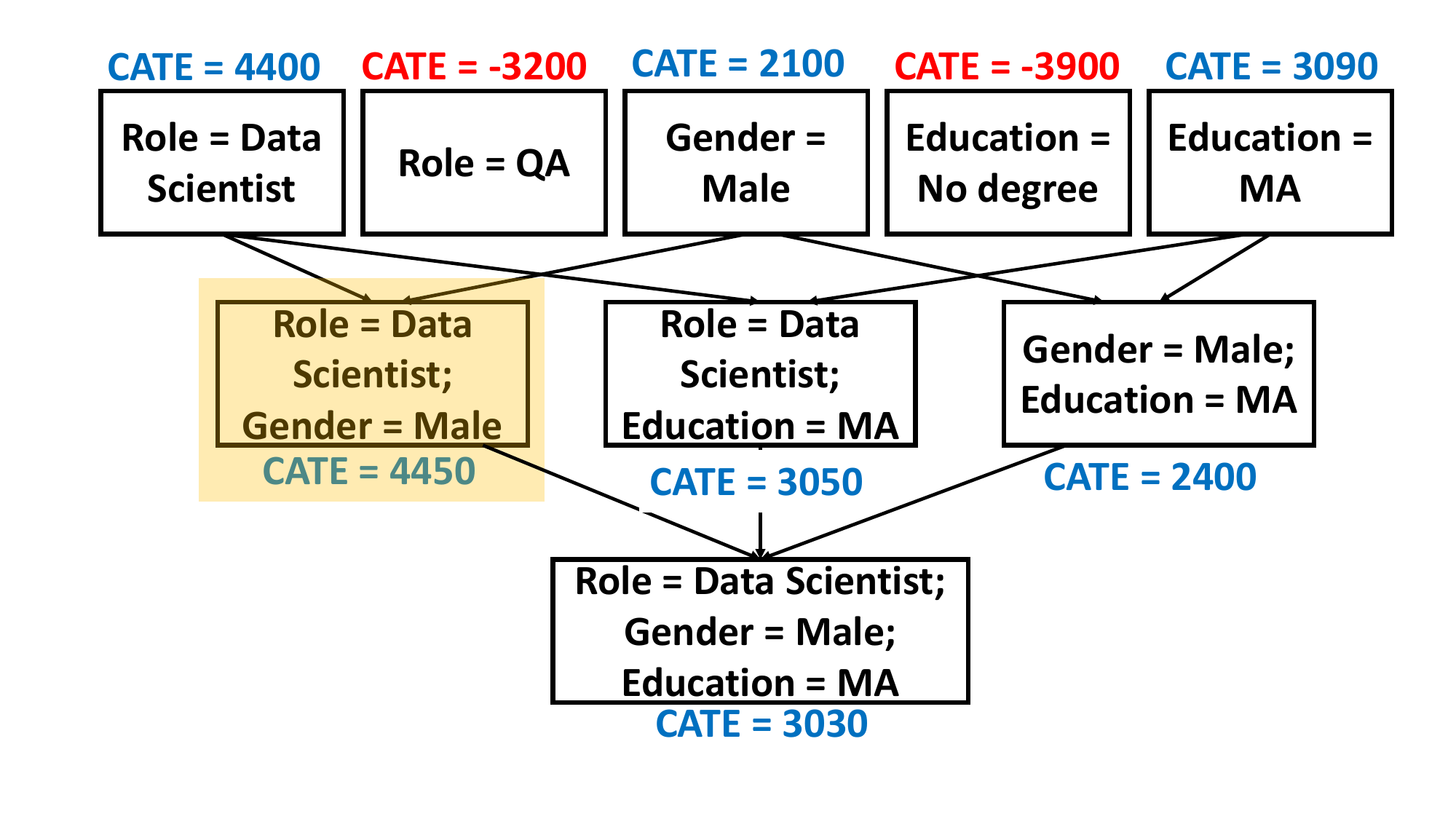}
 \vspace{-2mm}
\caption{Partial treatment-patterns lattice (Example~\ref{ex:lattice}).} \label{fig:treatments_graph}
\end{figure}

\begin{example}\label{ex:lattice}
We illustrate the operation of Algorithm \ref{algo:treatment_patterns} using Figure \ref{fig:treatments_graph}. Initially, it considers all treatment patterns with a single predicate (line 2). Assuming $\sigma {=} +$ (searching for the treatment pattern with the highest CATE), it moves to the next level (line 5). At this level, patterns with two predicates are considered only if both of their parents have a positive CATE. Thus, the pattern \verb|{Role = QA, Gender = Male}| is excluded as the CATE of \verb|Role = QA| is negative (line 8). The treatment pattern with the maximum CATE ($\pattern_t^{max}$) is found at the second level (marked in yellow) (line 9). After another iteration and generating the third level of the lattice, the algorithm does not proceed to the fourth level because $\pattern_t^{max}$ is not present in the third level (line 10).
\end{example}


By changing $\sigma$, 
Algorithm \ref{algo:treatment_patterns} can also find the treatment patterns with the lowest negative CATE. It can also be generalized to output the highest explainability (i.e., the highest absolute CATE value), or both the patterns with the highest and lowest CATE values.  

\paratitle{Optimizations}
We implement more optimizations for \cref{algo:treatment_patterns}.\\
\textbf{(a) Pruning attributes:} We eliminate attributes that do not have a causal relationship with the outcome attribute (lines 2, 7). Since these attributes have no impact on CATE values, they can be disregarded. We can detect such attributes by utilizing the input causal DAG or by removing attributes with low correlation to the outcome.\\
\textbf{(b) Pruning treatments:} In constructing the lattice (lines 2, 7), we exclude patterns with a near-zero CATE.
We observe that combining patterns with a CATE close to zero value often yields a similar result.
Consequently, when advancing to the next lattice level, we only consider the top 50\% patterns with the highest or lowest CATE.\\
\textbf{(c) Parallelism:} The process of extracting treatment patterns for each grouping pattern (lines 2, 7) can be performed in parallel since this procedure is dependent only on the grouping pattern. \\
\textbf{(d) Estimating CATE Values \revb{by Sampling}:}
\revb{To decrease runtime, we use a fixed-size random sample of tuples to estimate CATE values (lines 3, 8) and investigate the impact of sample size on runtime and accuracy (\cref{subsec:exp-sensitivity}). Random sampling of tuples is done only for evaluating CATE values. Our study shows that using a random sample size of $1$ million tuples yields CATE estimates that closely match those obtained from the entire dataset.}

\vspace{-2mm}
\subsection{Linear Program Formulation and Rounding}
\label{subsec:step_3}
\cut{Given a collection of explanation patterns, a threshold $\theta$, and an integer $k$, our goal is 
to solve the \probName\ problem. However, as discussed in \cref{subsec:problem}, this problem is NP-hard (in fact, even finding a solution that satisfies the constraints is NP-hard). 
}
Proposition~\ref{prop:hardness} shows that finding any feasible solution to the optimization problem in Definition~\ref{def:problem} is NP-hard. Our study shows that intuitive combinatorial greedy algorithms targeting the size and coverage constraints and maximizing explainability are unable to find good solutions; hence, we use an LP-rounding algorithm. 
\par
Given 
a collection of explanation patterns $\{\pattern_j\}_{j=1}^l$ with weights $w_j$  corresponding to their explainability, an integer $k$, and a threshold $\theta$, we construct the Integer Linear Program (ILP) shown in \cref{fig:lp} (which extends the ILP for the \emph{max-k-cover} problem). Here $g_j$ are the variables for patterns $\pattern_j$. $t_i$, $i = 1$ to $m$ are the variables for $m$ groups in $Q(D)$. $s_i {\in} \pattern_j$ denotes that $s_i$ satisfies $\pattern_j$.
\par
We consider the LP-relaxation of this ILP with variables in $[0, 1]$ instead of $\{0, 1\}$, then use an LP-solver to find solutions. If no solution is returned, we know that the original ILP had no feasible solution either. If any solution is returned by the LP, the original ILP may or may not have a feasible solution. We use the standard randomized rounding algorithm for max-k-cover \cite{raghavan1987randomized} (sample $k$ patterns with probability $\frac{g_j}{k}$), which guarantees at most $k$ sets, an $(1 - \frac{1}{e})$-approximation to the coverage constraint $\sum_{i = 1}^m t_i \geq \theta \cdot m$, and a $\frac{1}{k}$-approximation to the maximization objective $\sum_{j = 1}^l g_j \cdot w_j$ in the LP as well as the ILP since $OPT \ LP \geq OPT \ ILP$ (details in the full version \cite{fullversion}). 
\revb{However, the approximation guarantees of the proposed LP formulation are solely theoretical since the approximation to the objective holds when \emph{all explanation patterns are considered}. In practice, we use this LP-rounding algorithm in conjunction with the grouping and treatment pattern mining procedures described in Sections \ref{subsec:grouping_patterns} and \ref{subsec:treatment_patterns}, therefore incur a trade-off between value and efficiency and lose the theoretical guarantees. }


\cut{
\textbf{Variables}: 
Let $g_j$, $j {=} 1,\ldots, l$ be
indicator variables for explanation pattern $\pattern_j$ with fixed weights $w_j$. 
Let $t_i, i {=} 1 \ldots, m$ be indicator variables for covering the $i$-th group $s_i$ in $\Qagg(\db)$. We denote $s_i {\in} \pattern_j$ when $s_i$ satisfies $\pattern_j$. 
\textbf{Constraints}: The size constraint (1) ensures that no more than $k$ explanation patterns are selected. 
The consistency constraints (2) ensure that a group in $\Qagg(\db)$ is explained (covered) if at least one explanation pattern that covered it is selected. 
The coverage constraint (3) ensures that at least $\theta$ fraction of the groups are covered. Constraint (4) determines the binary domain of the variables. \\
\textbf{Objective}: maximize the overall explainability of chosen patterns.
}

\begin{figure}
    \centering
    \begin{tcolorbox}[colback=white]
    \vspace{-2mm}
    \begin{equation*}
    \begin{aligned}
    \max \quad & \sum_{j = 1}^l g_j \cdot w_j \quad 
    \textrm{s.t.} \quad (1)~\sum_{j = 1}^l g_j \leq k, \quad 
      (2)~t_i \leq \sum_{j: T_i \in \pattern_j} g_j ~\forall i = 1 \text{ to }m,     \\
      &(3)~\sum_{i = 1}^m t_i \geq \theta \cdot m, \quad 
      (4)~t_i, g_j \in \{0, 1\} ~\forall i = 1 \text{ to }m, ~\forall j = 1 \text{ to }l
    \end{aligned}
    \end{equation*}
    \vspace{-2mm}
    \end{tcolorbox}
    \caption{ILP for optimization problem  (line \ref{l:step3} in \cref{algo:full_algo}).}
    \label{fig:lp}
\end{figure}

\paratitle{Time complexity analysis for \cref{algo:full_algo}} 
The maximum number of explanation patterns in a database $D$ with attributes $\attrset$ is bounded by $|D|^{|\attrset|}$ (considering both grouping and treatment patterns and active domain of attributes), which is polynomial in \emph{data complexity} assuming a fixed schema \cite{Vardi82}. The number of patterns dominates the number of variables in the ILP in Figure~\ref{fig:lp}. Hence even if all patterns are enumerated and evaluated explicitly, the LP relaxation of the ILP can be solved (\sysName\ uses z3 \cite{MouraB08}) and rounded to an integral solution in polynomial time in $|D|$. The additional operations in this section (e.g., computation of CATE in \cref{algo:treatment_patterns}) are polynomial in $D$, giving a worst-case polynomial data complexity of \sysName. However, $|D|^{|\attrset|}$ can have a large value for large $|D|$ and $\attrset$. The suite of optimizations developed in this section reduces the running time of \sysName\ (ref. Section~\ref{sec:experiments}).
\section{Experimental Evaluation}
\label{sec:experiments}
We present experiments that evaluate the effectiveness and
efficiency of our proposed framework. We aim to address the following research questions. $\mathbf{Q1}$: What is the
quality of our explanations, and how does it compare to that of existing methods? $\mathbf{Q2:}$ How does each phase of \algoName\ contributes to its ability to find an explanation summary that satisfies our optimization goal and constraints?  
$\mathbf{Q3:}$ What is the efficiency of the \algoName\ algorithm? $\mathbf{Q4}$: How sensitive are the explanations to various parameters?

\paratitle{Prototype implementation}
\algoName\ was written in Python, and is publicly available in~\cite{fullversion}. 
CATE values computation was performed using the DoWhy library~\cite{dowhypaper}, utilizing their linear regression approach. 
\sysName\ generates a solution in natural language using predefined templates, as shown in Figure \ref{fig:so-explanation}. Those templates were generated via prompt questions to ChatGPT~\cite{chatgpt}, asking it to transform predicates into a human-readable text.

\subsection{Experimental Setup}
\label{subsec:setup}

\begin{table}
	\centering
		\caption{Examined datasets.}
			\label{tab:datasets}
	\begin{tabular} 
 {ccccc}
		\toprule
	\textbf{Dataset} & \textbf{tuples}& \textbf{atts}& \textbf{max values per att}&\textbf{grouping patterns}

	 \\
		\midrule

    German~\cite{asuncion2007uci} &1000&20&53&10\\
     Adult~\cite{adult} &32.5K&13&94&13\\
SO~\cite{stackoverflowreport}&38K&20&20&75\\
IMPUS-CPS~\cite{flood2015integrated}&1.1M&10&67&9\\
   Accidents~\cite{moosavi2019countrywide}&2.8M&40&127&15\\
   	

				\bottomrule
	\end{tabular}
\end{table}


\paratitle{Datasets}
We examine multiple commonly used datasets: \\
\textbf{\german:} This dataset contains details of bank account holders, including demographic and financial information, along
with their credit risk. The
causal graph was used from \cite{chiappa2019path}.\\
\textbf{\adult:} This dataset comprises demographic information of individuals along with their education, occupation, annual income, etc. We used the causal graph from \cite{chiappa2019path}. \\
\textbf{\so}: This is discussed in Example~\ref{ex:running_example}. 
\cut{
Stack Overflow's annual developer survey is a survey of people who code around the world.
It contains information
about developers' such as their age, income, and 
country. As there is a significant difference in the economies of the countries, we expanded this dataset by including attributes that describe the economy of each country, such as the \textsc{Gini}, \textsc{HDI}, and \textsc{GDP}.
}
The 
causal DAG was constructed using 
\cite{youngmann2023causal}.\\
\textbf{\impus:} 
This dataset is derived from the Current Population Survey conducted by the U.S. Census Bureau, 
which includes 
demographic details 
for individuals, e.g., education, occupation, and annual income. 
We adopted the causal dag from~\cite{chiappa2019path}.\\
\textbf{\accidents:} This dataset provides comprehensive coverage of car accidents across the USA. It includes numerous environmental stimuli features that describe the conditions surrounding the accidents, such as visibility, precipitation, and traffic signals. To construct a causal DAG, we followed the methodology outlined in~\cite{youngmann2023causal}.

\noindent\textbf{Synthetic:} 
\common{We constructed a schema comprising the attributes $G, G_1, \ldots, G_i, T_1, \ldots, T_j, O$, where: $G$ serves as the grouping attribute in the query, with each tuple taking a unique value ranging from 1 to $n$ (where $n$ is the number of tuples). Attributes $G_1, \ldots, G_i$ are used to establish grouping patterns. Each attribute divides the values of $G$ into varying numbers of buckets. Attributes $T_1, \ldots, T_j$ are employed to define treatment patterns. Each tuple is assigned a random value between 1 and 5 for each attribute independently. The outcome $O$ is defined as: $T_1 - T_2 + T_3 -\ldots+T_j$. 
There is a large number of grouping and treatment patterns to be considered.
For each grouping pattern, the treatment patterns that exhibit the highest causal effects encompass all the attributes $T_1, T_2, \ldots, T_j$. For instance, treatment patterns with high positive causal effect on $O$ are characterized by high values for odd $T$ attributes and low values for even $T$ attributes (e.g., ${T_1 = 5, T_2 = 1, T_3 = 5, \ldots}$). 
}

\vspace{1mm}
\paratitle{Baselines}
We compare \algoName\ with the following baselines:\\
 \textbf{\bruteforce}: The optimal solution according to \cref{def:problem}. This algorithm implements an exhaustive search over all possible grouping and treatment pattern combinations.\\
\textbf{\ids}: The authors of \cite{lakkaraju2016interpretable} have proposed a framework for generating Interpretable Decision Sets (\ids) for prediction tasks. The framework incorporates parameters restricting the percentage of uncovered data tuples and the number of rules. These parameters were assigned the same values used in our system. \\
\textbf{\frl}: The authors of \cite{chen2018optimization} introduce a framework for producing Falling Rule Lists (\frl) as a probabilistic classification model. FRLs consist of a sequence of if-then rules, with the if-clauses containing antecedents and the then-clauses containing probabilities of the desired outcome. The order of rules in a falling rule list reflects the order of the probabilities of the outcome.\\
\textbf{\exptable}: \revc{The authors of \cite{el2014interpretable} introduced an efficient method to generate \emph{explanation tables} for multi-dimensional datasets. The proposed algorithm employs an information-theoretic approach to select patterns that provide
the most information gain about the distribution of the outcome attribute. 
Since this algorithm does not consider an input SQL query, we have included another variant, denoted as \exptableG, to consider the query. This variant considers the grouping patterns found by \algoName\ and reports the pattern found by \exptable\ separately for each group. }\\
\common{\textbf{\xinsight}: The authors of~\cite{abs-2207-12718} introduced a framework that employs causal discovery and identifies causal patterns to explain the differences between two groups in a SQL query result. We generate an explanation by comparing all $\binom{m}{2}$ pairs of groups in $Q(\db)$. For a fair comparison, we skipped the causal discovery phase and gave \xinsight\ the causal DAG \algoName\ uses. }

\vspace{1mm}
Since \ids, \frl, and \exptable\ assume a binary outcome attribute, we binned the outcome variable in each examined scenario using the average outcome values. \reva{We gave \frl, \ids, and \exptable\ the input table, as they do not consider a SQL query, obtaining a single set of rules for each dataset.}

We also considered ChatGPT~\cite{chatgpt} as a baseline. We provided it with multiple prompts comprising a task description, SQL query, causal DAG, and dataset link. However, ChatGPT's responses consistently indicated its lack of direct access to external datasets. Consequently, it couldn't provide specific insights on the datasets, offering only general domain-related insights instead. Although these insights included identifying attributes with strong causal effects on the outcome, they did not provide specific patterns. Consequently, we excluded this baseline from presentation. 

\vspace{1mm}
\paratitle{Variations of \algoName} We consider the following variations:
 \textbf{\bruteforcelp}: As in \bruteforce, all grouping and treatment patterns are examined. In the final step, an approximated solution is obtained by employing the LP formulation described in \cref{subsec:step_3}.\\
\textbf{\greedy}: This baseline utilizes the approaches described in \cref{sec:algo} to generate promising grouping and treatment patterns. In the final step, 
it employs a greedy strategy instead of solving an LP. The strategy involves iteratively selecting explanation patterns based on their explainability and the increase in coverage they offer. 

\vspace{1mm}
Unless otherwise specified, the size constraint is set to 5, and the coverage threshold is set to 0.75. The threshold of the Apriori algorithm is set to 0.1. 
\ids, \frl, and \exptable\ use their default parameters. The time cutoff is set to $3$ hours. The experiments were executed on a PC with a
$4.8$GHz CPU, and $16$GB memory.

\subsection{Quality Evaluation ($Q_1$)}
\label{subsec:quality_eval}
We assess the quality of our explanations relative to the baselines when examining a range of queries on diverse real-world datasets. Since no ground truth is available, we examine the consistency of our findings with insights from prior research (as was done in \cite{galhotra2022hyper, salimi2018bias, youngmann2022explaining}). For each dataset, we present the result for a representative aggregated SQL query. Our queries are inspired by real-life sources, such as the Stack Overflow annual reports~\cite{stackoverflowreport}, media websites (e.g., The 19th Newsletter~\cite{dads}), and academic papers (e.g., \cite{salimi2018bias,aljaban2021analysis}). 
More details and additional use cases are provided in \cite{fullversion}.

 \begin{figure}[t]
        \centering
        \begin{minipage}[t]{1.0\linewidth}
            \small
            \begin{tcolorbox}[colback=white]
\vspace{-2mm}
  \textsf{$\bullet$ \magenta{\underline{For countries in Europe}} (e.g., Spain, Italy), the most substantial effect on high salaries (effect size of 34K, $p {<}$ 1e-3) is observed for \textcolor{blue}{\underline{individuals under $45$}}. Conversely, \textcolor{red}{\underline{being under $25$}} has the greatest adverse impact on salary (effect size of -32K, $p {<}$ 1e-3 ).}\\
                \textsf{$\bullet$ \magenta{\underline{For countries with a high GDP level}} (e.g., Sweden, Spain), the most substantial effect on high salaries (effect size of 40K, $p {<}$ 1e-3 ) is observed for \textcolor{blue}{\underline{white individuals under $35$}}. Conversely, \textcolor{red}{\underline{being over $55$} }has the greatest adverse impact on salary (effect size of -34K,$p {<}$ 1e-4 ).}\\
                \textsf{$\bullet$\magenta{\underline{For countries with a high Gini coefficient} }(e.g., Turkey, Brazil), the most substantial effect on high salaries (effect size of 29K, $p {<}$ 1e-4) is observed for \textcolor{blue}{\underline{white individuals under $45$}}. Conversely, being \textcolor{red}{\underline{being over $55$} }has the greatest adverse impact on salary (effect size of -28K, $p {<}$ 1e-3).}
            \vspace{-2mm}
            \end{tcolorbox}
            
        \end{minipage}
        \caption{\revb{Causal explanation summary by CauSumX for \so\ use-case (sensitive attributes only)}.}
        \label{fig:so-output-sensitive}
    \end{figure}
    
     \begin{figure}[t]
        \centering
        \begin{minipage}[b]{1.0\linewidth}
            \small
            \begin{tcolorbox}[colback=white]
            \vspace{-2mm}
      \textsf{$\bullet$ \magenta{\underline{For cities in the Northeast region}} (e.g., Boston, Albany), the combination of \textcolor{blue}{\underline{overcast weather conditions and low visibility}} has been found to have a substantial positive effect on severity (effect size of 0.55, p < 1e-3). \textcolor{red}{\underline{The presence of traffic signals}} has been identified as the factor with the largest adverse impact on severity (effect size of -0.42, p < 1e-4).} \\
      \textsf{$\bullet$ \magenta{\underline{For cities in the Midwest region}} (e.g., Chicago, Detroit), the combination of \textcolor{blue}{\underline{cold temperatures and snow}} has been found to have a substantial positive effect on severity (effect size of 0.61, p < 1e-3). \textcolor{red}{\underline{Clear weather}} has been identified as the factor with the largest adverse impact on severity (effect size of -0.31, p < 1e-3).} \\
      \textsf{$\bullet$ \magenta{\underline{For cities in the South region}} (e.g., Huston, Miami), \textcolor{blue}{\underline{rain}} has been found to have a substantial positive effect on severity (effect size of 0.3, p < 1e-3). The presence of \textcolor{red}{\underline{traffic calming measures}} has been identified as the factor with the largest adverse impact on severity (effect size of -0.44, p < 1e-3).} \\
      \textsf{$\bullet$ \magenta{\underline{For cities in the West region}} (e.g., Phoenix, Los Angeles), \textcolor{blue}{\underline{the absence of traffic signals and traffic}} \textcolor{blue}{\underline{calming measures}} has been found to have a substantial positive effect on severity (effect size of 0.53, p < 1e-4). \textcolor{red}{\underline{City roads}} (as opposed to highways) has been identified as the factor with the largest adverse impact on severity (effect size of -0.25, p < 1e-4).}
            \end{tcolorbox}
        \end{minipage}
        \caption{\revb{Causal explanation summary by CauSumX for \accidents use-case}.}
        \vspace{-3mm}
        \label{fig:accidents}
    \end{figure}

\begin{figure*}[t]
  \centering
  \begin{subfigure}[b]{0.32\textwidth}
    \centering
    \includegraphics[width=\textwidth]{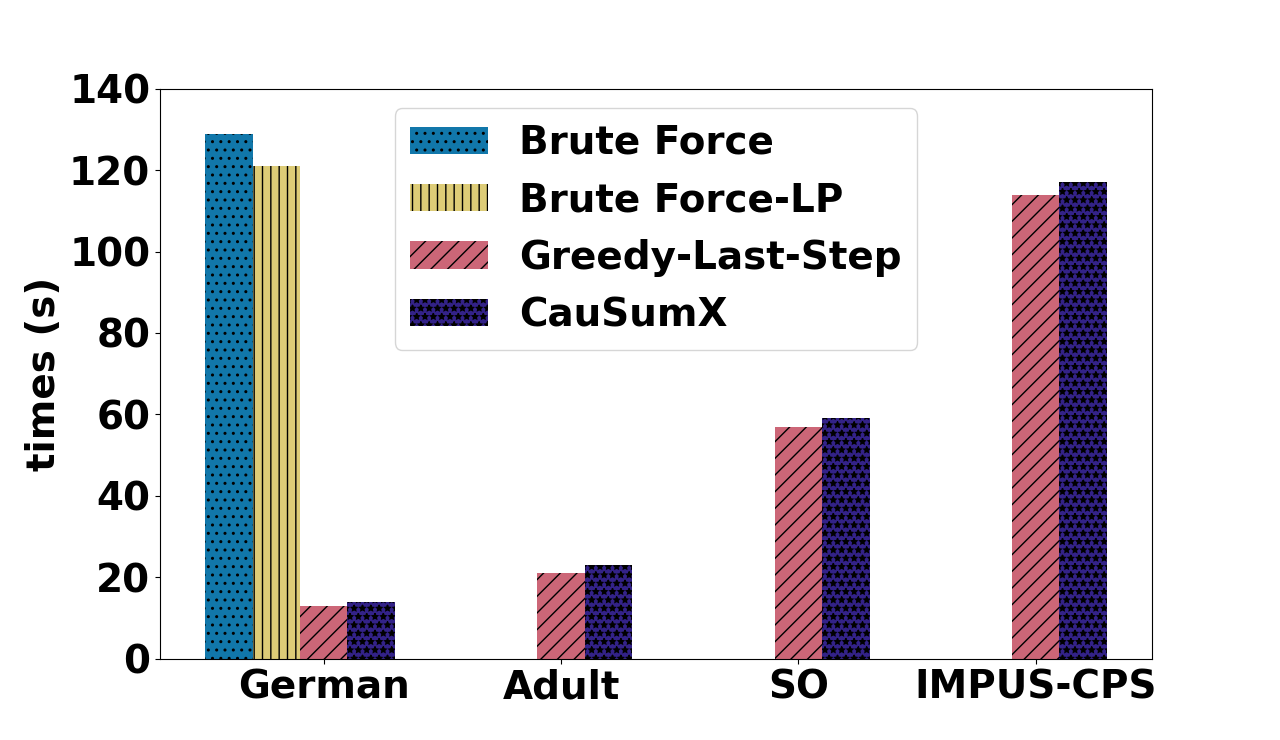}
    \vspace{-7mm}
    \caption{Running times}
    \label{fig:figure1}
  \end{subfigure}
  \hfill
    \begin{subfigure}[b]{0.32\textwidth}
    \centering
    \includegraphics[width=\textwidth]{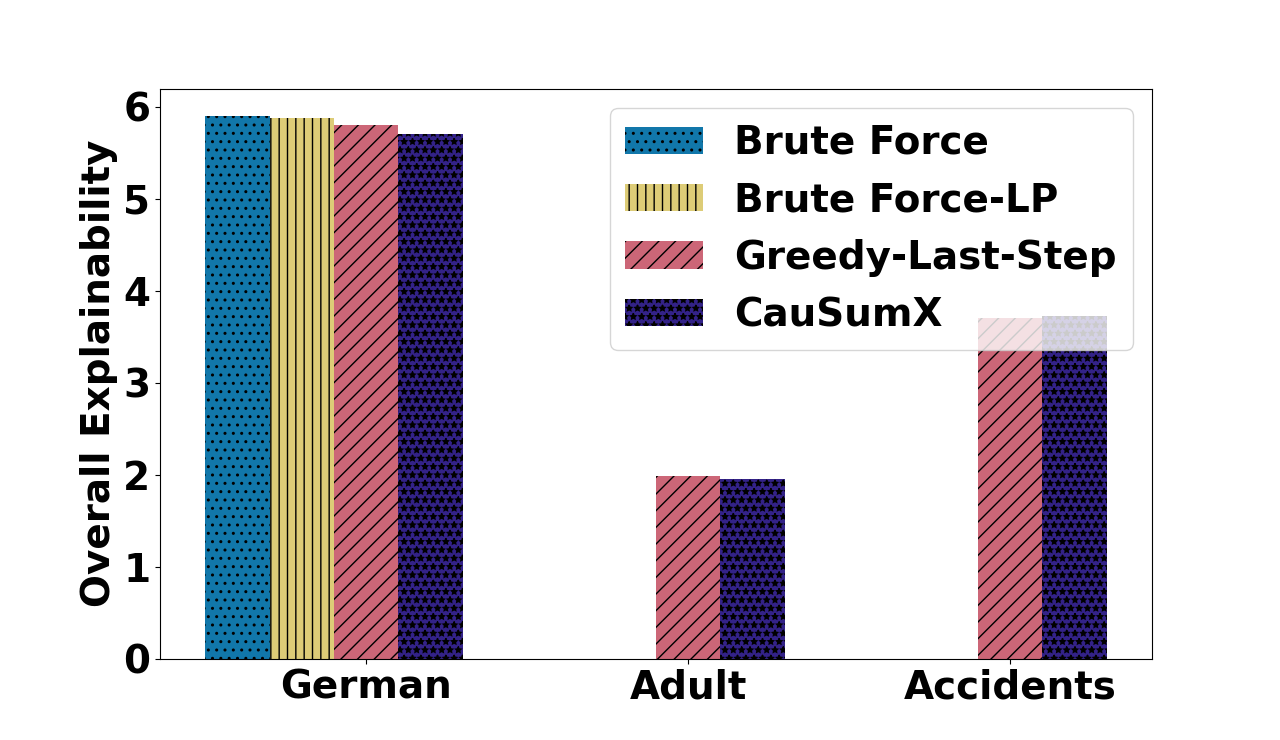}
        \vspace{-7mm}
    \caption{Overall explainability}
    \label{fig:figure3}
  \end{subfigure}
  \hfill
  \begin{subfigure}[b]{0.32\textwidth}
    \centering
    \includegraphics[width=\textwidth]{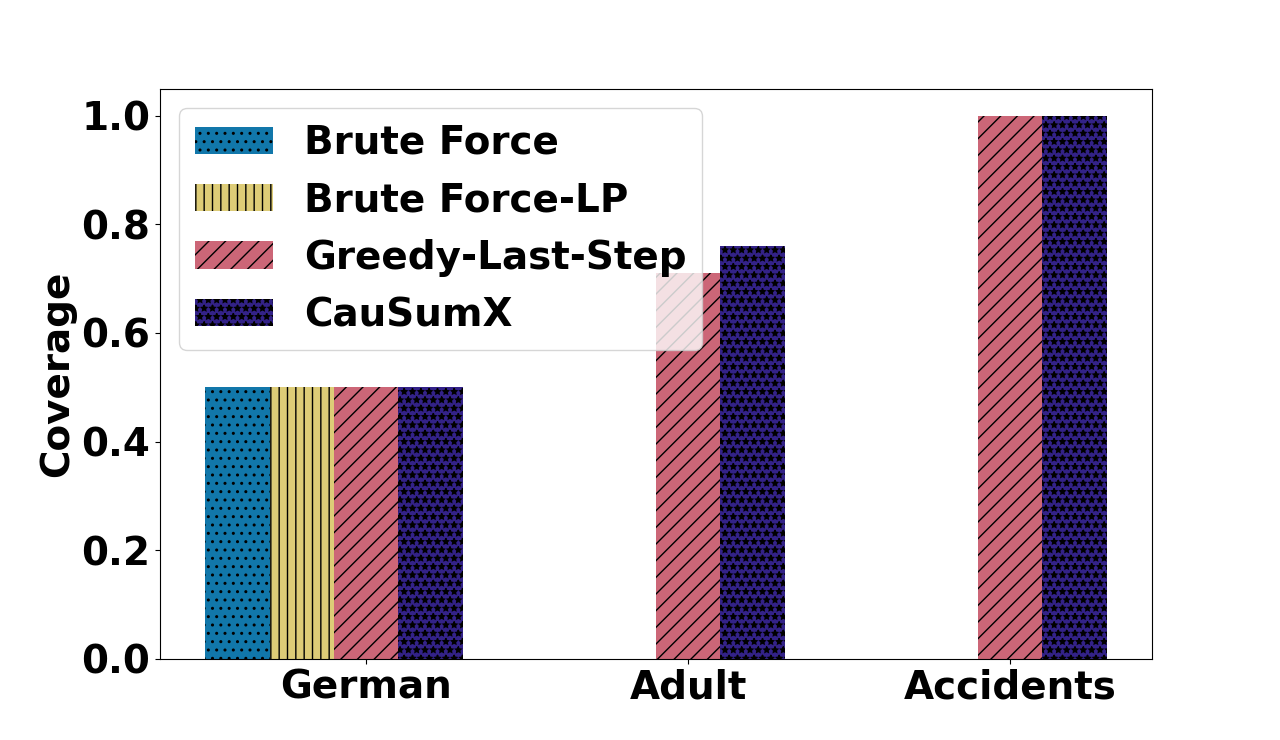}
        \vspace{-7mm}
    \caption{Coverage}
    \label{fig:figure2}
  \end{subfigure}
  \caption{Performance of different variants of \algoName. Baselines that exceed the time cutoff are excluded.}
  \label{fig:variants}
  \vspace{-4mm}
\end{figure*}

\begin{figure}[t]
  \centering
  \begin{subfigure}[b]{0.43\textwidth}
    \centering
    \includegraphics[width=\textwidth]{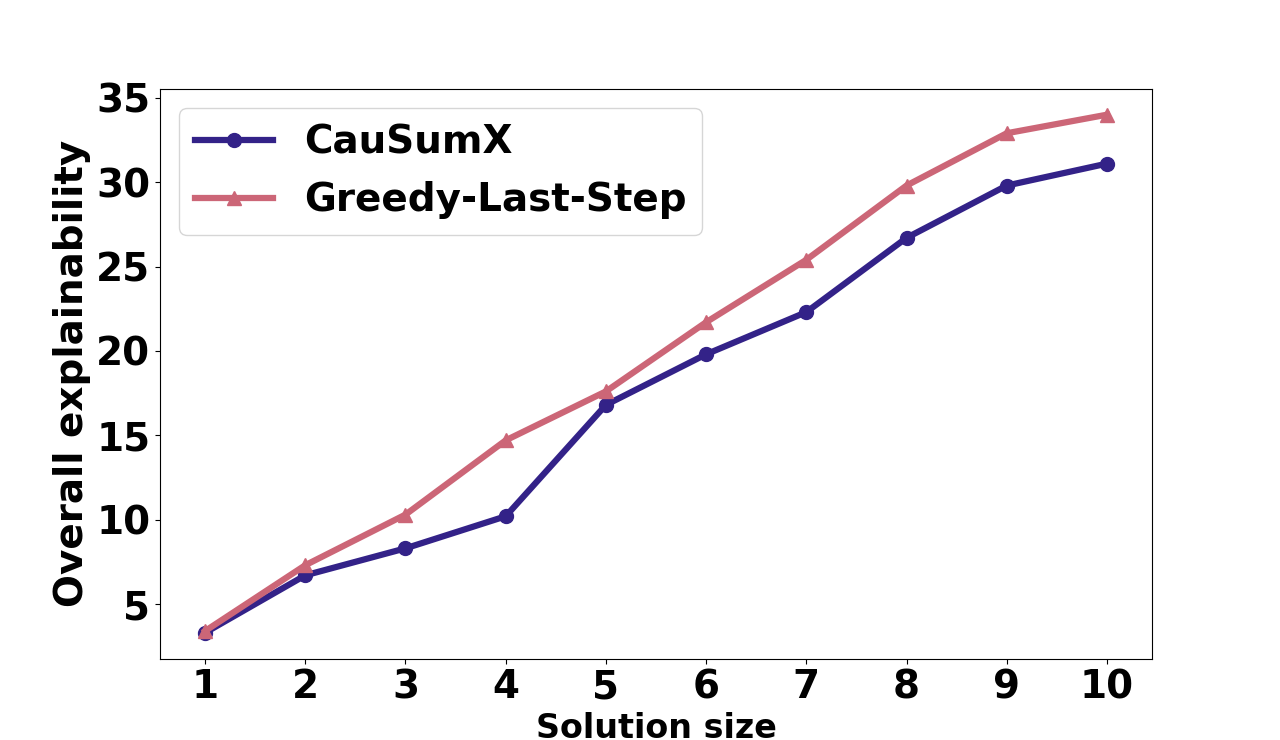}
    \vspace{-6mm}
    \caption{Explainability}
    \label{fig:figure4}
  \end{subfigure}
     \begin{subfigure}[b]{0.43\textwidth}
    \centering
    \includegraphics[width=\textwidth]{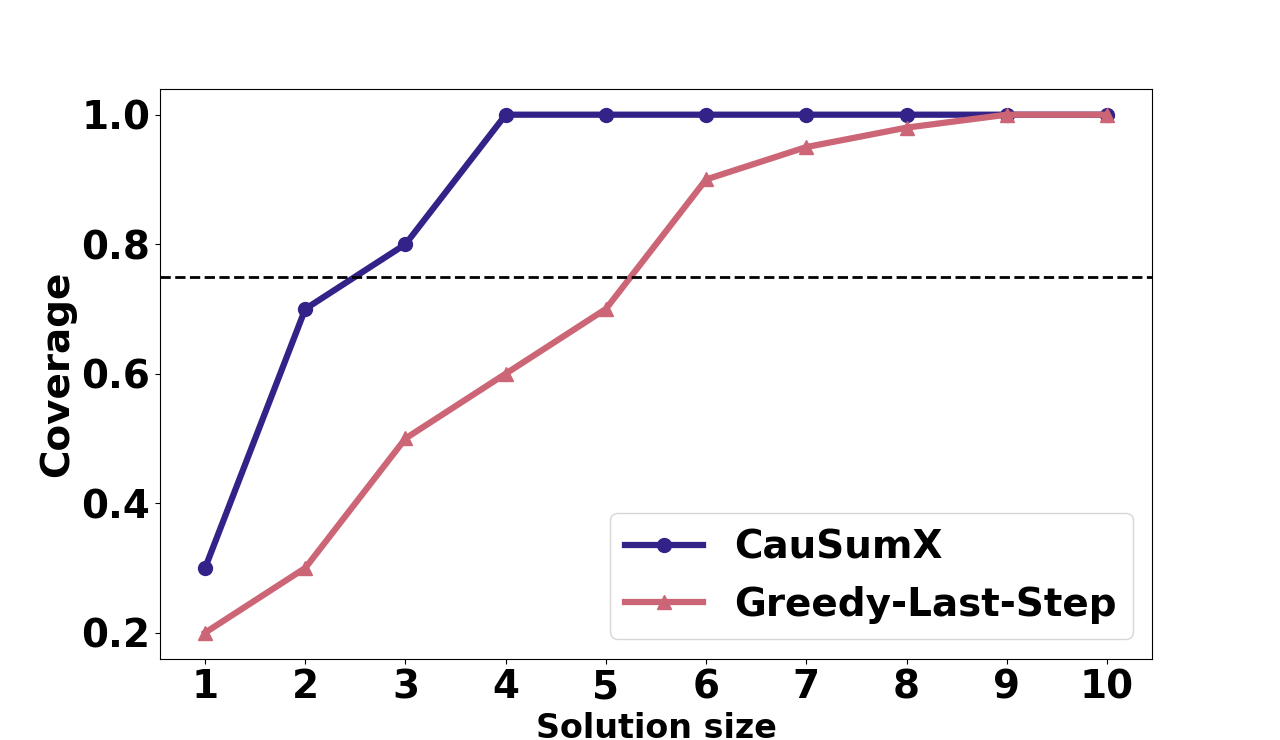}
     \vspace{-6mm}
    \caption{Coverage}
    \label{fig:figure2}
  \end{subfigure}
  \caption{Analysis of \algoName\ and \greedy.}
  \label{fig:cpe_vs_greedy}
\end{figure}
\begin{figure}[t]
  \centering
  \begin{subfigure}[b]{0.43\textwidth}
    \centering
    \includegraphics[width=\textwidth]{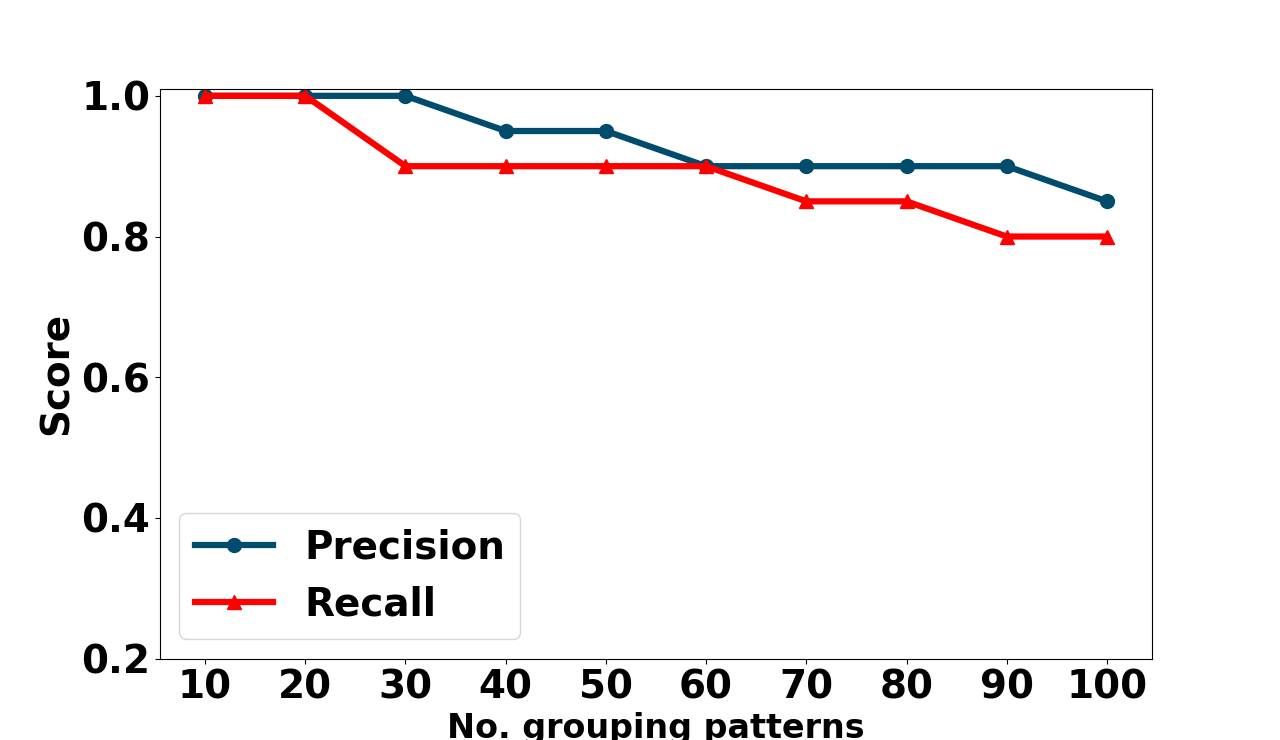}
    \vspace{-6mm}
    \caption{\common{Grouping patterns}}
    \label{fig:figure4}
  \end{subfigure}
     \begin{subfigure}[b]{0.43\textwidth}
    \centering
    \includegraphics[width=\textwidth]{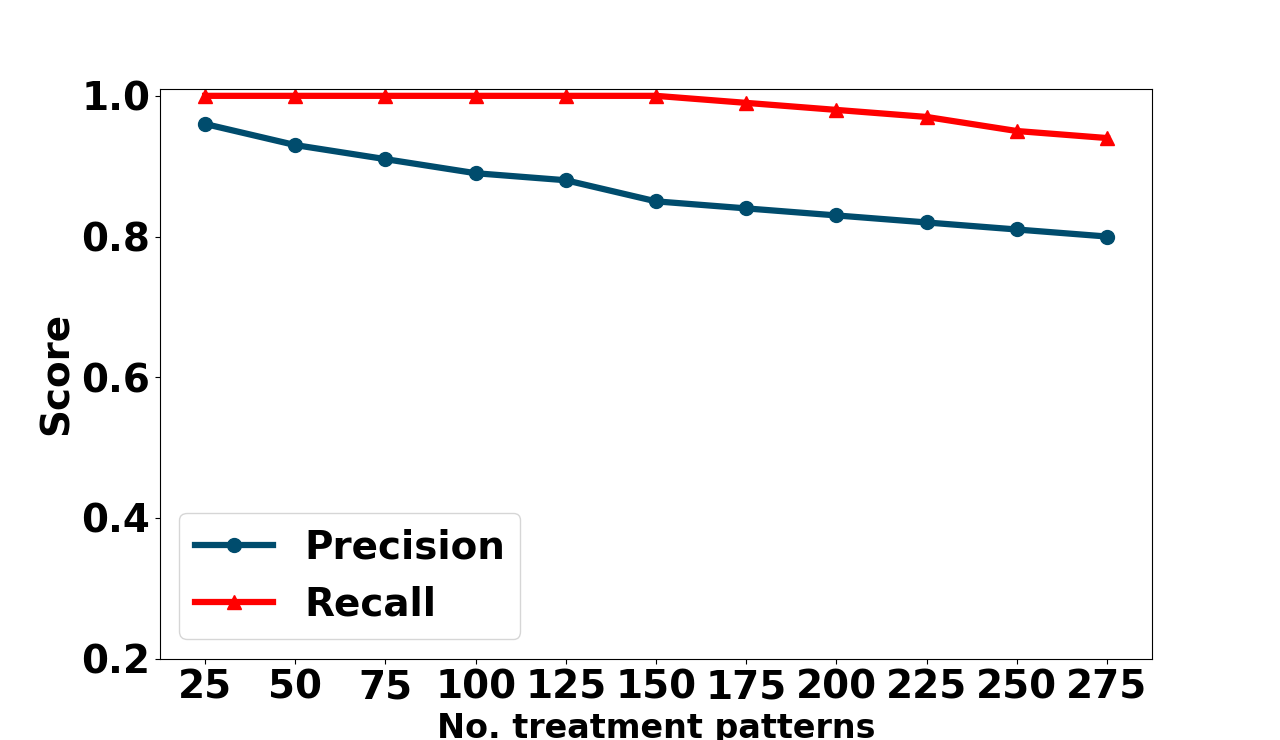}
     \vspace{-6mm}
    \caption{\common{Treatment patterns}}
    \label{fig:figure2}
  \end{subfigure}
  \caption{\common{Analysis of precision and recall} }
  \label{fig:precision_recall}
\end{figure}

\paratitle{\so}
As in our running example, we consider a SQL query that computes the average salary of developers across the 20 most commonly mentioned countries among respondents (accounting for more than 85\% of the \so\ dataset). To define grouping patterns, we considered attributes having FDs with \textsc{country}: \textsc{Continent, HDI, Gini}, and \textsc{GDP}. Here, for brevity, we set the solution size to $3$. 
As mentioned, the explanation summary generated by \algoName (shown in Figure \ref{fig:so-explanation}), reveals key insights regarding the factors influencing salary across different countries. It highlights that job role, age, and education level are the primary determinants of income in the examined countries. Notably, individuals in C-level positions tend to earn significantly more than students. These results align with previous research~\cite{stackoverflowreport,vizajobs}, which emphasized the significant influence of education level and job responsibilities on salaries in high-tech.
Additionally, our results indicate that being below the age of $35$ positively impacts salary, while being over $55$ has a negative impact on income. Age discrimination toward people in the IT industry was identified in the literature~\cite {bdtechtalks,lin2021detecting}. 

\common{XInsight is designed to identify patterns that account for variations in the average salary across pairs of countries ($\binom{20}{2}$), and thus, it results in an extensive explanation (of size exceeding 500KB). As, \algoName, it focuses on patterns involving attributes that causally influence salary, such as \textsf{Formal Education} and \textsf{Gender}. However, these explanations offer distinct perspectives on the query results. For instance, XInsight's explanation sheds light on the substantial variance in average salaries between the US and Poland, highlighting differences in the distribution of role types. In the US, there is a prevalence of machine learning specialists and C-level executives, while Poland has a higher concentration of Back-end developers and DevOps specialists.
Conversely, \algoName\ explanation uncovers commonalities among various groups, such as the positive impact of being under the age of 35 on salaries across all European countries. Therefore, these two solutions can be viewed as complementary efforts, each revealing different causal insights behind aggregate views. 
Nevertheless, if the query returns only two groups, \algoName\ 
can find treatments (patterns) that explain the difference between their outcomes. 
In the US and Poland case, the returned treatments by \algoName\ were similar to that obtained from \xinsight (e.g., \algoName\ also recognized that having a C-level executive position is the treatment with the highest positive effect on respondents from the US). 
In contrast, in the presence of multiple groups in the query result, 
\xinsight\ lacks a straightforward extension for generating a summarized explanation for the entire aggregate view.}

\revc{Both \exptable\ and \exptableG\ aim to discover patterns related to high or low salaries. Notably, \textsf{YearsCoding} emerged as a significant factor in these patterns, indicating that individuals with over 30 years or less than 2 years of coding experience tend to have lower salaries. This aligns with our observation related to being a student or being over 55, which also leads to reduced salaries. However, \textsf{Age} and \textsf{Role} have a more substantial causal effect on salary than \textsf{YearsCoding}.
One limitation of \exptable\ is its inability to capture variations among different groups (countries). It either identifies country-specific patterns (e.g., low income for individuals from India) or universal patterns (e.g., \textsf{YearsCoding}). In contrast, \exptableG\ is designed to capture variations among groups but similarly highlights \textsf{YearsCoding} as a highly informative factor. This underscores the fundamental difference between \exptable, which prioritizes patterns with high information gain, and our approach, which focuses on identifying patterns with strong causal effects.}



\paratitle{Focusing on Sensitive Attributes}
To identify potential biases, we focused exclusively on sensitive attributes (such as ethnicity, gender, and age) when examining treatment patterns. The generated explanation is depicted in \cref{fig:so-output-sensitive}.
Our results indicate that demographic factors significantly influence salary in all countries examined. Specifically, being under $35$ positively impacts income, while being over $55$ has a negative effect.
Similarly, being a white male correlates with higher salary.
These findings align with previous research on demographic impact on income, such as the gender wage gap~\cite{dads} and disparities based on ethnicity~\cite{bertrand2004emily,chetty2020race}.
This showcases \sysName's versatility in identifying causal explanations across different attributes. It also emphasizes its potential in uncovering disparities among demographic groups, supporting efforts to combat discrimination, and promoting equality.


\paratitle{\accidents} 
We investigated the average severity of car accidents across cities in the US.
Our generated explanation is shown in \cref{fig:accidents}. 
It indicates that adverse weather conditions, such as cold temperatures and snow, tend to escalate the severity of car accidents. Conversely, the presence of traffic signals and calming measures appears to mitigate severity. Our results highlight variations in weather conditions across different regions. For instance, in the Midwest, cold temperatures and snow commonly contribute to severe accidents, whereas in the South, rainy weather emerges as a more significant factor for severe car accidents as snow and cold temperatures are less common. 
Previous studies \cite{pardon2013effectiveness, nilsson1982effects} have provided evidence supporting the effectiveness of traffic signals and calming measures in reducing accident severity. 
This analysis underscores the potential value of our system in extracting insights that go beyond common patterns and correlations. Such insights can be valuable for decision-makers, enabling a deeper understanding of causal factors and aiding in implementing effective road safety measures.
\common{The vast number of cities (over 50,000) made it unfeasible to generate an explanation using XInsight within our time constraints.}
\revc{The \ids, \frl, \exptable\, and \exptableG\ baselines also exceeded our time cutoff}.


\common{\subsection{Accuracy Evaluation}}
\common{To evaluate the accuracy of the \algoName\ algorithm, we conducted experiments over syntactic data where the ground truth is known.
}
\common{
Here, we set $n {=} 1k$ and used the default system parameters.

\paratitle{Metrics of evaluation} We report the precision and recall of the grouping and treatment pattern mining algorithms compared to the exhaustive \bruteforce\ baseline.
To assess the precision and recall of the grouping mining algorithm (Section \ref{subsec:grouping_patterns}), we compare the tuples covered by the grouping patterns chosen by \algoName\ against those covered \bruteforce.
To assess the performance of the treatment mining algorithm (Section \ref{subsec:treatment_patterns}), we consider a fixed count of 20 grouping patterns, where both \algoName\ and \bruteforce selected the same grouping patterns. To evaluate the precision and recall of the treatment mining algorithm in comparison to the treatments chosen by \bruteforce, we compared the tuples identified as the treated group by \algoName\ with those defined as the treated group according to \bruteforce. We report average precision and recall across all grouping-treatment pattern combinations.

\paratitle{Grouping patterns} We manipulate the count of attributes utilized for grouping patterns, thereby controlling the number of grouping patterns. 
The results are shown in Figure \ref{fig:precision_recall}(a). Our results show that when the count of potential grouping patterns is relatively low (up to 20), our grouping patterns mining algorithm selects patterns that cover the same tuples as \bruteforce. As the number of grouping patterns to consider increases, the algorithm prunes certain patterns, leading to a decrease in precision and recall. Nevertheless, these scores consistently remain high (above 0.78), signifying that \algoName\ and \bruteforce\ cover nearly the same tuples. 

\paratitle{Treatment patterns} 
The results are shown in Figure \ref{fig:precision_recall}(b). Observe that the recall remains consistently high, irrespective of the number of treatments. This suggests that the treated group according to \bruteforce\ is always encompassed by the treated group identified by \algoName. In contrast, precision decreases as the count of treatment patterns increases. This signifies that the treated group according to \algoName\ may contain irrelevant tuples. This is due to the pruning optimizations employed by \algoName, which prevent it from materializing "long" patterns. However, even with many treatment patterns, the precision consistently exceeds 0.75.
}

\begin{figure}[t]
  \centering
  \begin{subfigure}[b]{0.43\textwidth}
    \centering
    \includegraphics[width=\textwidth]{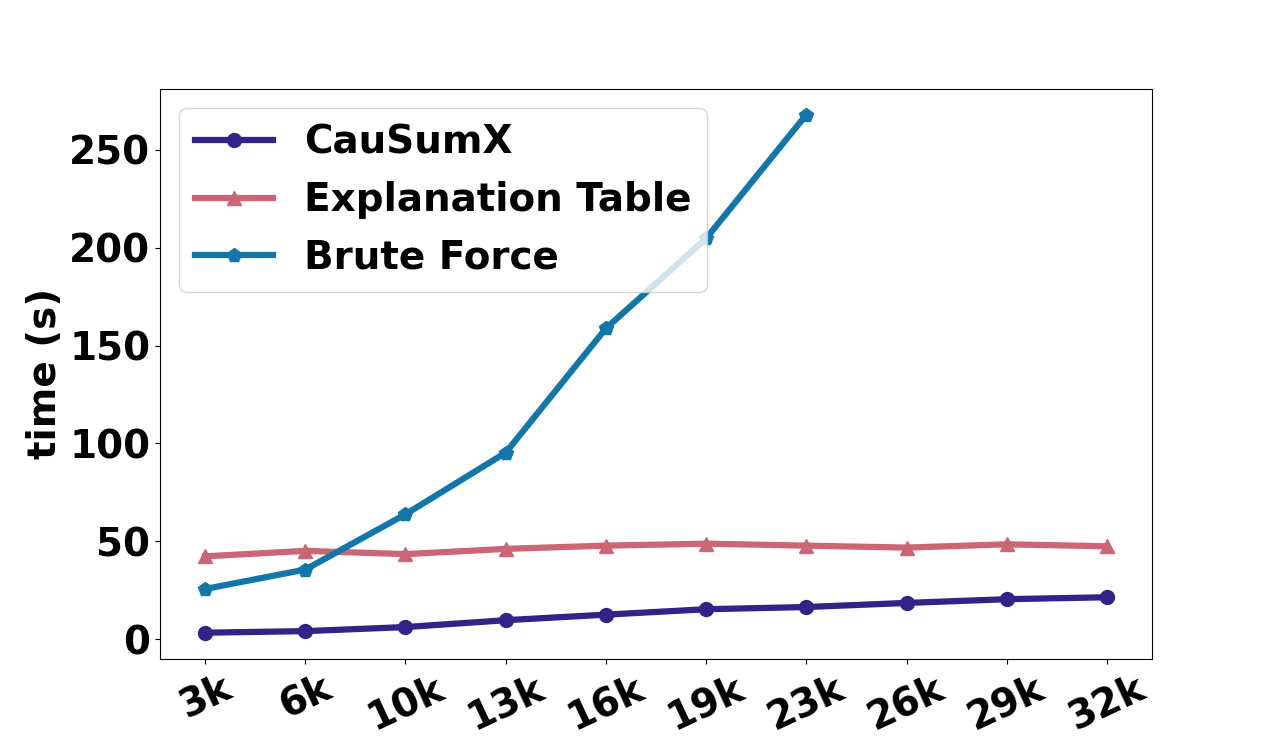}
    \caption{Adult}
    \label{fig:figure1}
  \end{subfigure}
  \begin{subfigure}[b]{0.43\textwidth}
    \centering
    \includegraphics[width=\textwidth]{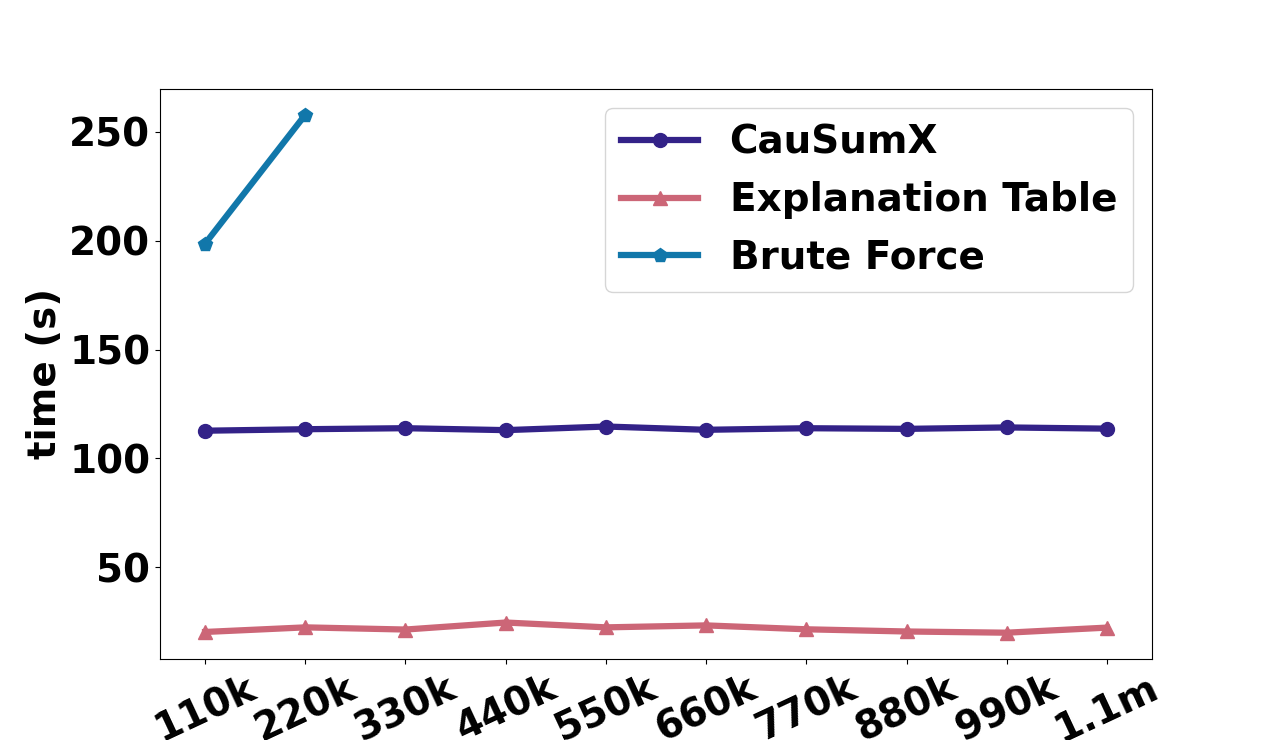}
    \caption{IMPUS-CPS}
    \label{fig:figure2}
  \end{subfigure}
  \caption{Time vs. dataset size.}
  \label{fig:data_size}
\end{figure}


\begin{figure}[t]
  \centering
  \begin{subfigure}[b]{0.43\textwidth}
    \centering
    \includegraphics[width=\textwidth]{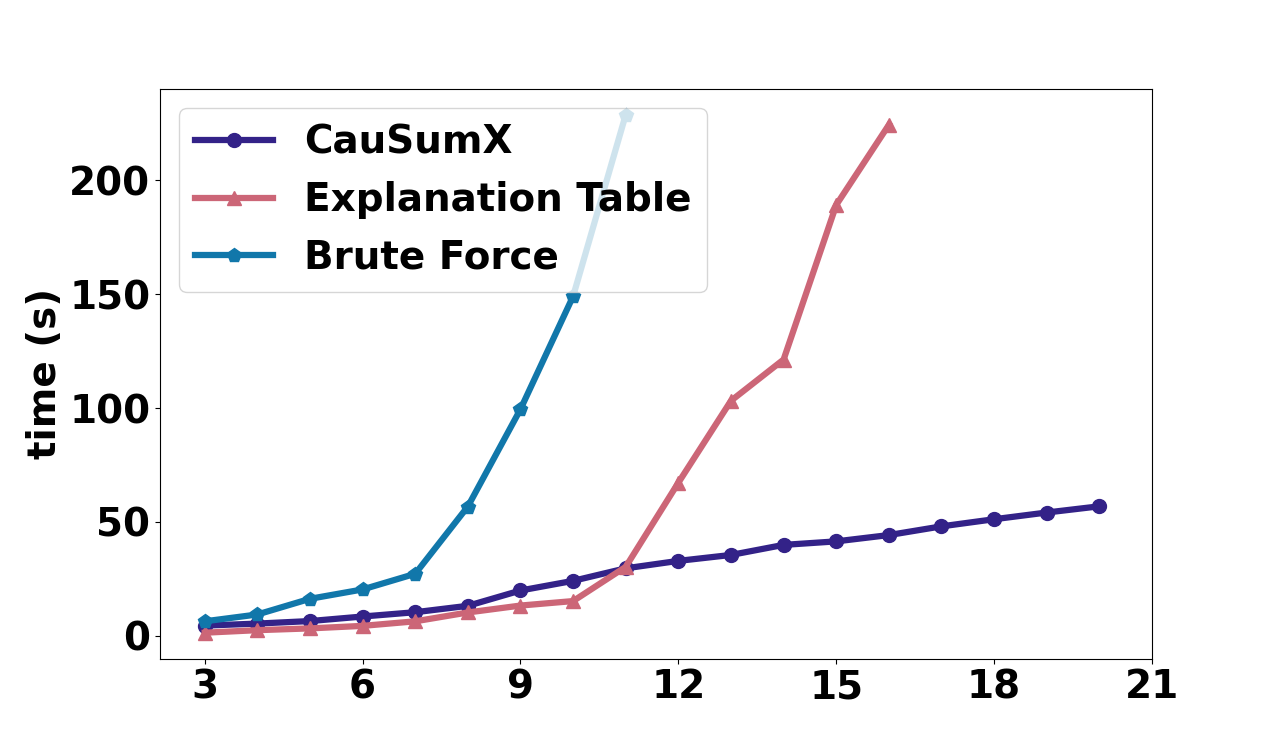}
    \caption{Stack Overflow}
    \label{fig:figure3}
  \end{subfigure}
     \begin{subfigure}[b]{0.43\textwidth}
    \centering
    \includegraphics[width=\textwidth]{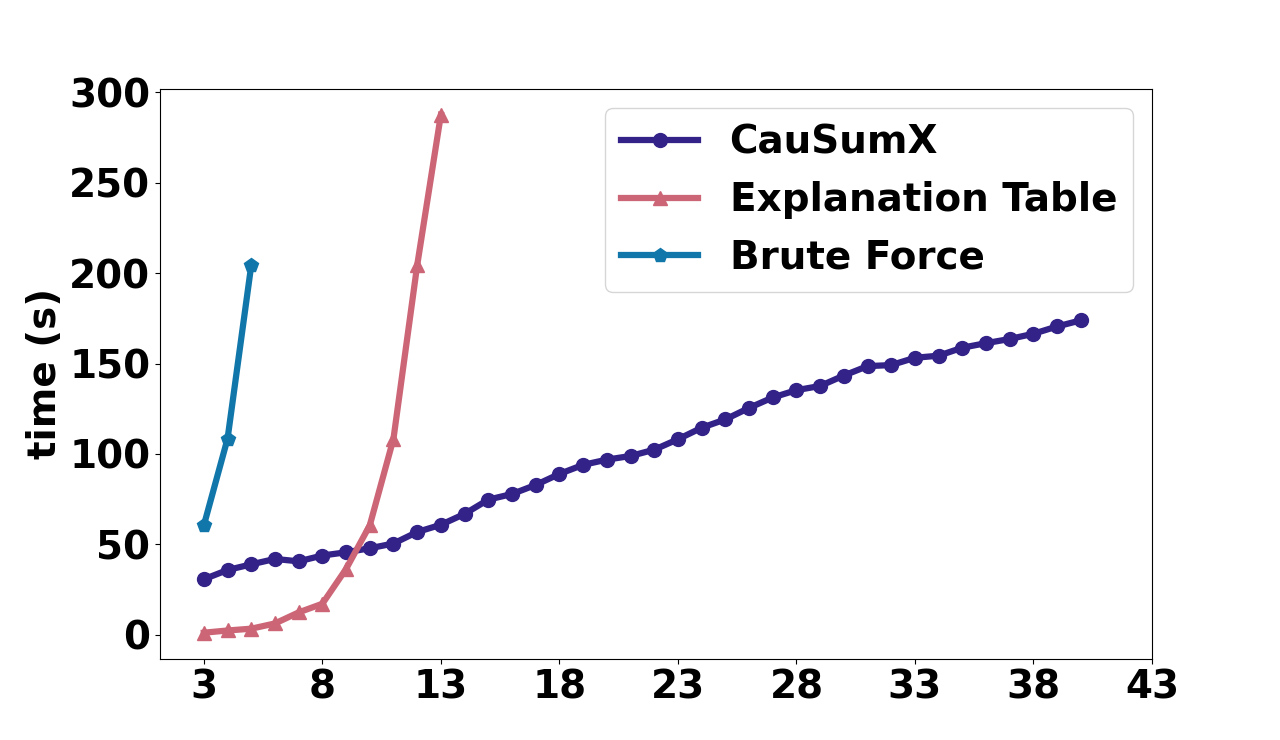}
    \caption{Accidents}
    \label{fig:figure2}
  \end{subfigure}
  \caption{Time vs. the number of attributes.}
  \vspace{-5mm}
  \label{fig:atts_num}
\end{figure}


\begin{figure}[t]
  \centering
  \begin{subfigure}[b]{0.43\textwidth}
    \centering
    \includegraphics[width=\textwidth]{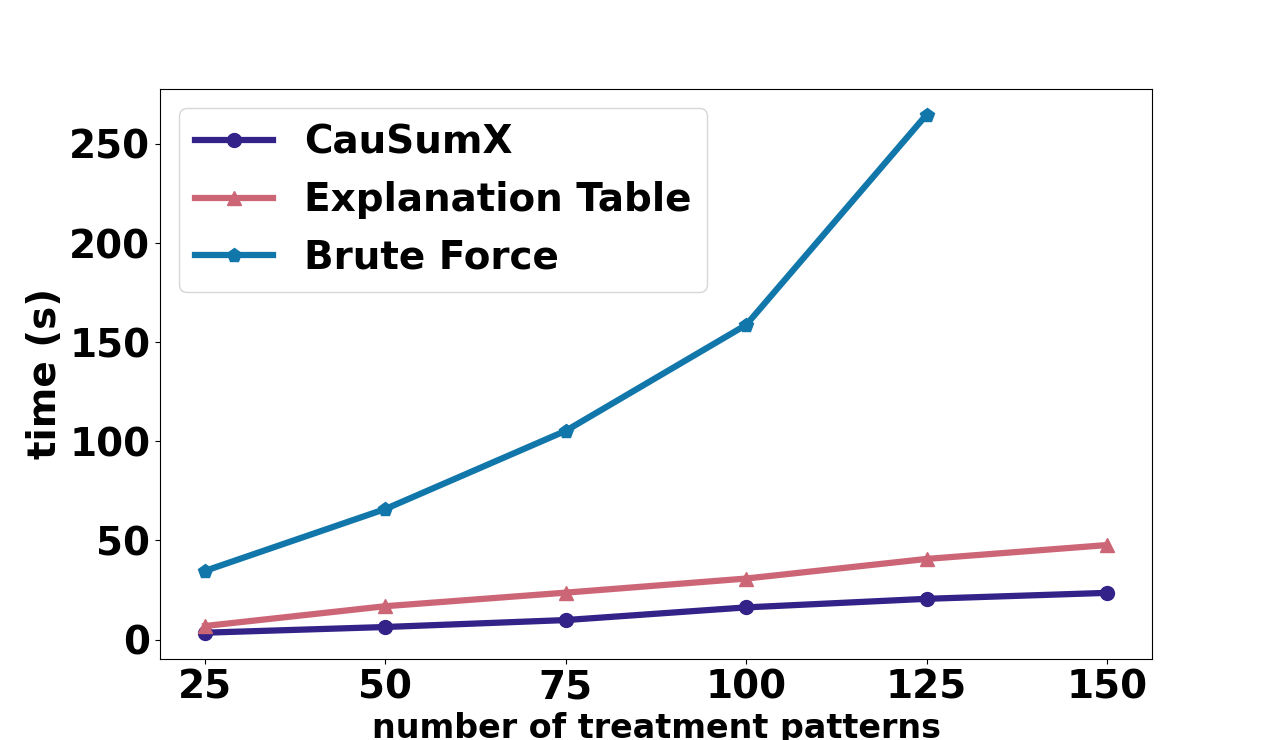}
    \caption{Adult}
    \label{fig:figure1}
  \end{subfigure}
  \begin{subfigure}[b]{0.43\textwidth}
    \centering
    \includegraphics[width=\textwidth]{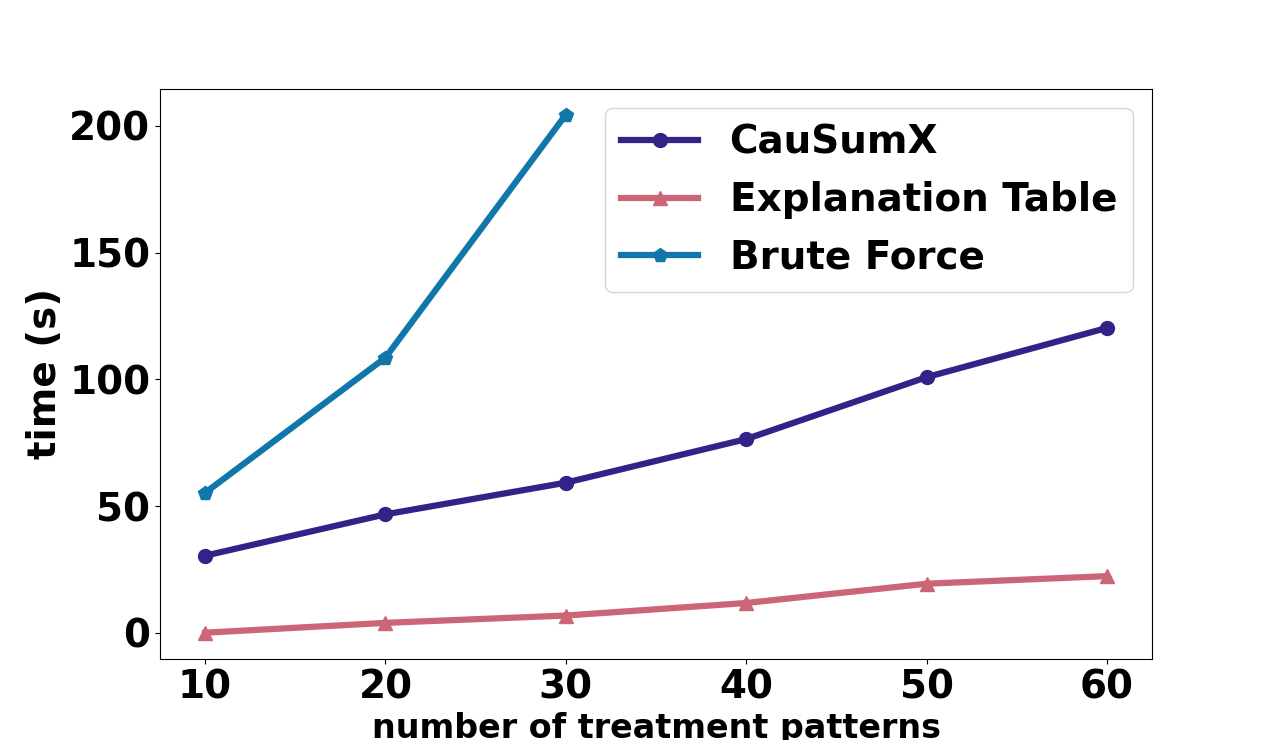}
    \caption{IMPUS-CPS}
    \label{fig:figure4}
  \end{subfigure}
  \caption{Time vs. the number of treatment patterns.}
  \vspace{-4mm}
  \label{fig:bins_num}
\end{figure}


\begin{figure}[t]
\centering
\includegraphics[scale = 0.2]{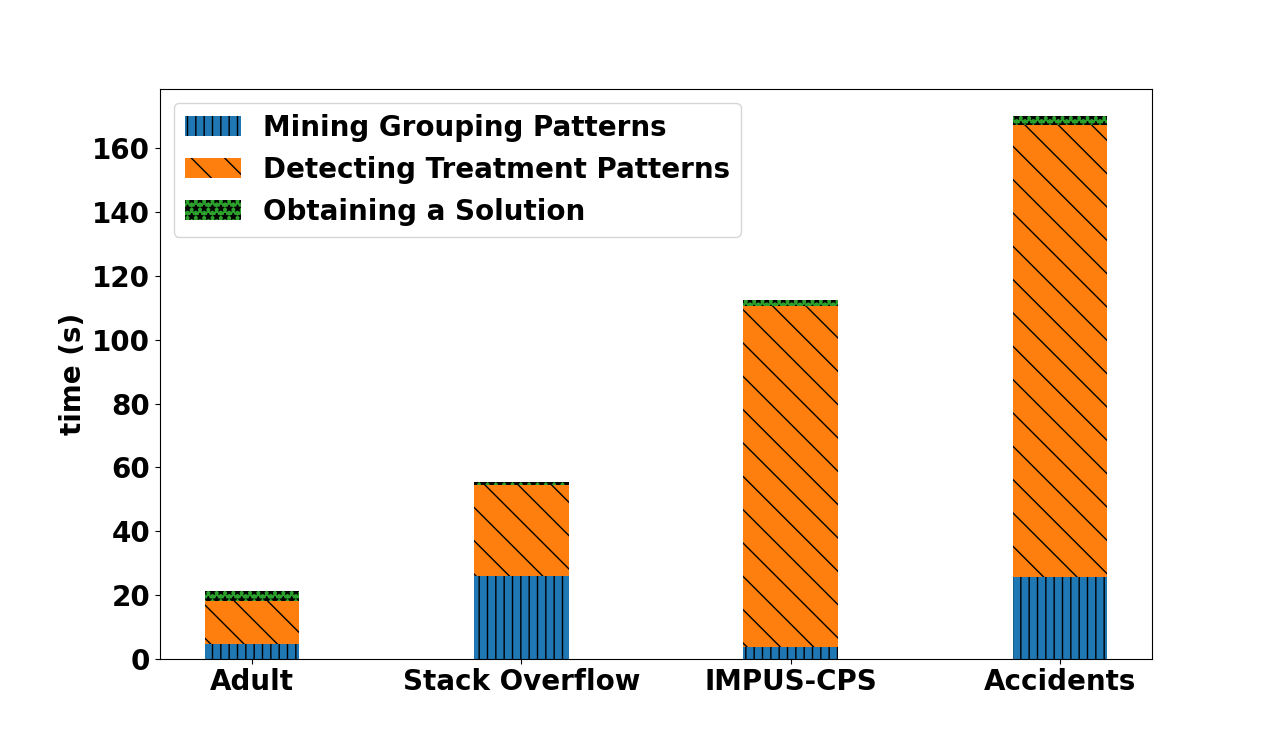}
\caption{\revc{Runtime by-step of the \algoName\ algorithm}} \label{fig:algo_steps}
\end{figure}

\subsection{Ablation Study ($Q_2$)}
\label{subsec:variants}
\revc{A breakdown analysis by step of \algoName\ (Figure \ref{fig:algo_steps}) shows that in all cases, mining the treatment pattern phase (Algorithm \ref{algo:treatment_patterns}) consumes most of the time. The first and last steps are relatively fast. This aligns with our time complexity analysis. }
We next analyze the impact of our grouping and treatment pattern mining algorithms, as well as the LP formulation, compared to the optimal solution determined by \bruteforce\ based on \cref{def:problem}.




\paratitle{Runtime}
Consider Figure \ref{fig:variants}(a). As expected, the algorithms that utilize our grouping and treatment patterns algorithms (\algoName, \greedy) exhibit significantly faster performance compared to \bruteforce variants. Both \bruteforce\ and \bruteforcelp\ could not process any dataset other than \german, as their runtimes exceeded our time cutoff. \emph{This clearly demonstrates the efficiency gained by our algorithms}.
The disparity in running times between \greedy\ and \algoName\ is negligible. This can be attributed to the fact that the treatment pattern detection phase consumes most of the execution time. Additionally, the relatively low number of grouping patterns explains the minimal difference in execution times during the last phase. Despite \greedy\ being faster than solving the LP, there are not too many explanation patterns to consider.

\paratitle{Coverage \& Explainability}
Figures \ref{fig:variants}(b) and (c) display the explainability and coverage of each baseline. In \german, all baselines achieve the same coverage determined by $k$, but the approaches considering all patterns have higher explainability. This minimal difference highlights the effectiveness of our pattern mining and pruning techniques, improving runtime without compromising explainability significantly.
In \accidents, \algoName\ and \greedy\ yield identical solutions. However, in \adult, \greedy\ achieves higher explainability but falls short in coverage compared to \algoName. This emphasizes the better balance achieved by our LP formulation in satisfying coverage and maximizing the objective compared to the greedy approach.

\paratitle{In-depth Analysis}
In a comprehensive analysis comparing \algoName\ and \greedy, we explored their performance by varying the solution size $k$ on the \so\ dataset. The findings are presented in Figure \ref{fig:cpe_vs_greedy}.
As $k$ increases, both algorithms exhibit improved overall explainability, with similar performance in this aspect. However, their behavior diverges in terms of coverage. \algoName\ demonstrates a faster ability to satisfy the coverage constraint (shown by the dashed horizontal line). This is because \algoName\ treats coverage as a constraint, while \greedy\ has no guarantees for coverage. As a result, \greedy\ only satisfies the coverage constraint for $k {=} 6$.
Based on our findings, \algoName\ surpasses \greedy, as both achieve comparable objective values, but \algoName\ has a higher likelihood of satisfying constraints.

\subsection{Efficiency Evaluation ($Q_3$)} \label{subsec:exp-runtime}
We showcase the scalability of \algoName. For brevity,
we excluded certain figures while discussing the trends they reveal within the text.
We omit the results for \ids\ and \frl\ from the presentation, as their response times exceed 10 minutes. \revc{We also exclude the results of \exptableG\ since they exhibited similar trends to those of \exptable.} \common{As we are not making a direct comparison with XInsight but instead with a variant that generates explanations for all pairs of groups, we have chosen to omit this baseline from presentation, as its runtime should be evaluated for the task of generating an explanation for a single pair.}

\paratitle{Data Size}
We analyze the impact of dataset size on runtime through random sampling of tuples. The results for the \adult\ and \impus\ datasets are shown in Figure \ref{fig:data_size}. 
\algoName\ and \bruteforce\ demonstrate a nearly linear increase in runtime for \adult\ (and \so\ - omitted from presentation) due to their full utilization of data for CATE value computation. However, \algoName\ employed a sampling optimization for the larger \impus\ (and \accidents) dataset, resulting in a more consistent runtime. \exptable's runtime is unaffected by dataset size due to sampling, but it is unable to handle datasets with more than $10$ attributes (e.g., \so).

\paratitle{\# Attributes}
We examine the impact of attribute quantity on runtime, by randomly excluding attributes from consideration. The results for the \so\ and \accidents\ datasets are shown in Figure \ref{fig:atts_num}.
\bruteforce\ and \exptable\ show exponential runtime increases with attribute number due to the growing number of patterns to consider. In contrast, \algoName\ exhibits linear growth in runtime, thanks to pruning techniques mentioned in Section \ref{subsec:treatment_patterns} that eliminate non-promising treatment patterns.

\paratitle{Treatment Patterns}
We analyze the impact of treatment pattern quantity on runtime. We vary the number of bins for ordinal attributes and randomly exclude values for non-ordinal attributes. The results for \adult\ and \impus\ are displayed in Figure \ref{fig:bins_num}. Runtime increases linearly for all algorithms, which is expected due to the increased solution space.


\paratitle{Grouping Patterns}
We examine the impact of grouping pattern quantity on runtime. By adjusting the threshold of the Apriori algorithm, we explore different numbers of grouping patterns. 
For \algoName, the runtime remains relatively unchanged across all scenarios due to its simultaneous exploration of promising treatment patterns for each grouping pattern. However, \bruteforce's runtime increases linearly with the number of grouping patterns. 

\paratitle{Solution Size}
Lastly, we explore the impact of varying $k$. As this parameter affects only the final phase of \algoName\ and \bruteforce, we observe negligible changes in their runtimes.

\subsection{Explanations Sensitivity ($Q_4$)} \label{subsec:exp-sensitivity}
We evaluate the impact of various parameters on the quality of the explanation summaries. The measures we focus on are overall explainability and coverage. 
Full details are provided in~\cite{fullversion}.



\paratitle{Apriori Threshold}
We investigate the effect of varying the threshold parameter $\tau$ in the Apriori algorithm. Increasing $\tau$ leads to a reduction in the number of grouping patterns considered. 
Our findings indicate that higher $\tau$ values lead to a decrease in both explainability and coverage. Based on our findings, we recommend using a default threshold of $0.1$, which provides satisfactory results in terms of runtime, explainability, and coverage. However, it can be adjusted  according to specific coverage requirements.

\begin{figure}[t]
  \centering
  \begin{subfigure}[b]{0.43\textwidth}
    \centering
    \includegraphics[width=\textwidth]{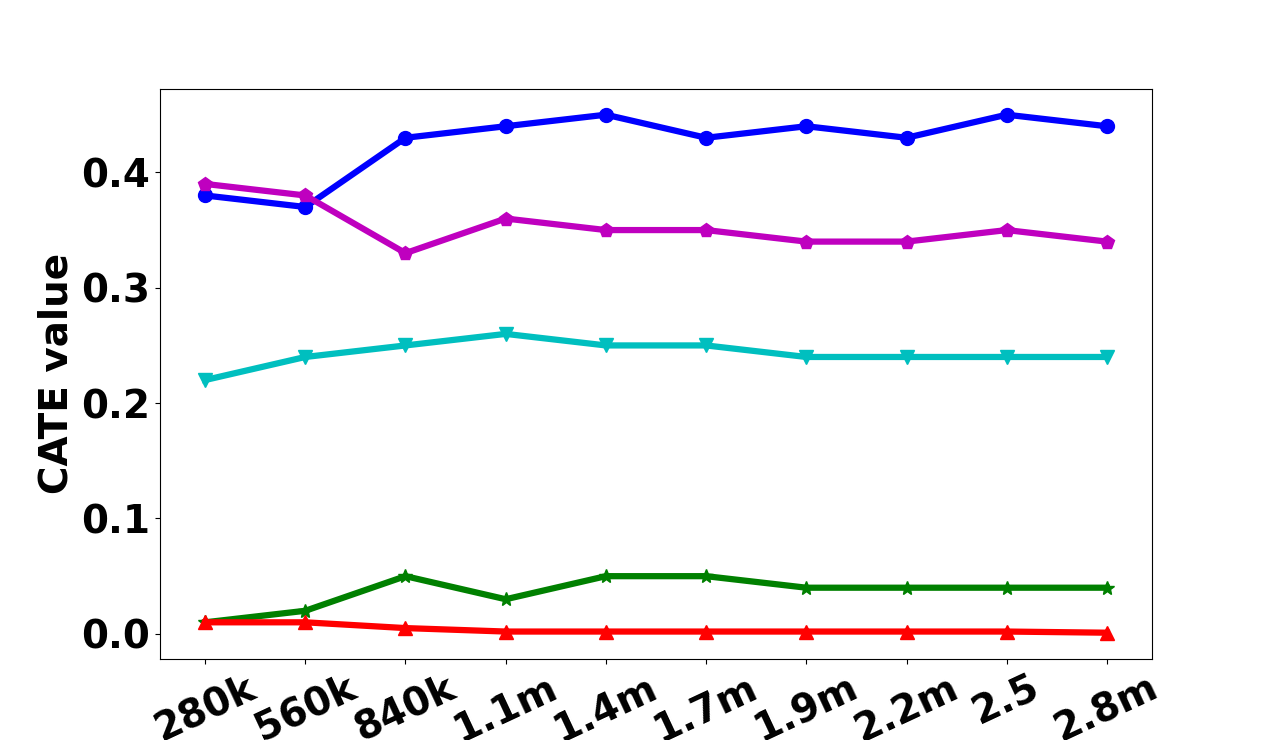}
    \caption{CATE value vs. sample size}
    \label{fig:figure4}
  \end{subfigure}
     \begin{subfigure}[b]{0.43\textwidth}
    \centering
    \includegraphics[width=\textwidth]{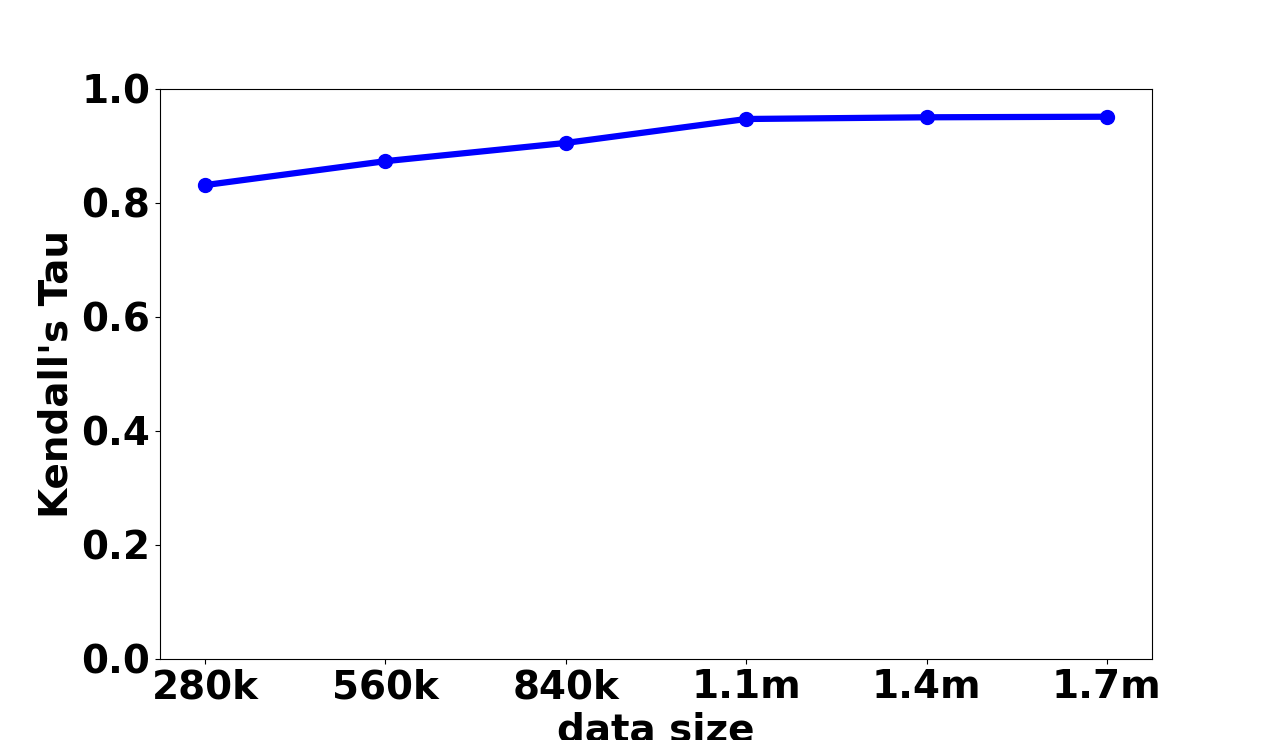}
    \caption{Kendall's $\tau$ vs. sample size}
    \label{fig:figure2}
  \end{subfigure}
  \caption{\revb{CATE Values Estimation (Accidents dataset)}.}
  \vspace{-3mm}
  \label{fig:cate_estimation}
\end{figure}

\paratitle{CATE Values Estimation}
\revb{Recall that we use fixed-size random sampling for estimating CATE values.
We investigate the impact of sample size on CATE value estimation.
Figure \ref{fig:cate_estimation} illustrates the results for the \accidents\ dataset. In Figure \ref{fig:cate_estimation}(a), we present the estimated CATE values for $5$ random treatments using various sample sizes. In Figure \ref{fig:cate_estimation}(b), we evaluate the agreement between rankings using Kendall's $\tau$ correlation coefficient. We randomly selected $20$ treatments and ranked them based on their CATE values, comparing this ranking with rankings obtained using different sample sizes. Notably, for a sample size of $1$m tuples, CATE values exhibit an error of no more than $5\%$, and Kendall's $\tau$ reaches a high and stable value of $0.95$. Similar trends were observed for the \impus\ dataset. Consequently, we conclude that a sample size of $1m$ tuples is suitable for accurate estimation of CATE values.}


\begin{figure}[t]
  \centering
  \begin{subfigure}[b]{0.43\textwidth}
    \centering
    \includegraphics[width=\textwidth]{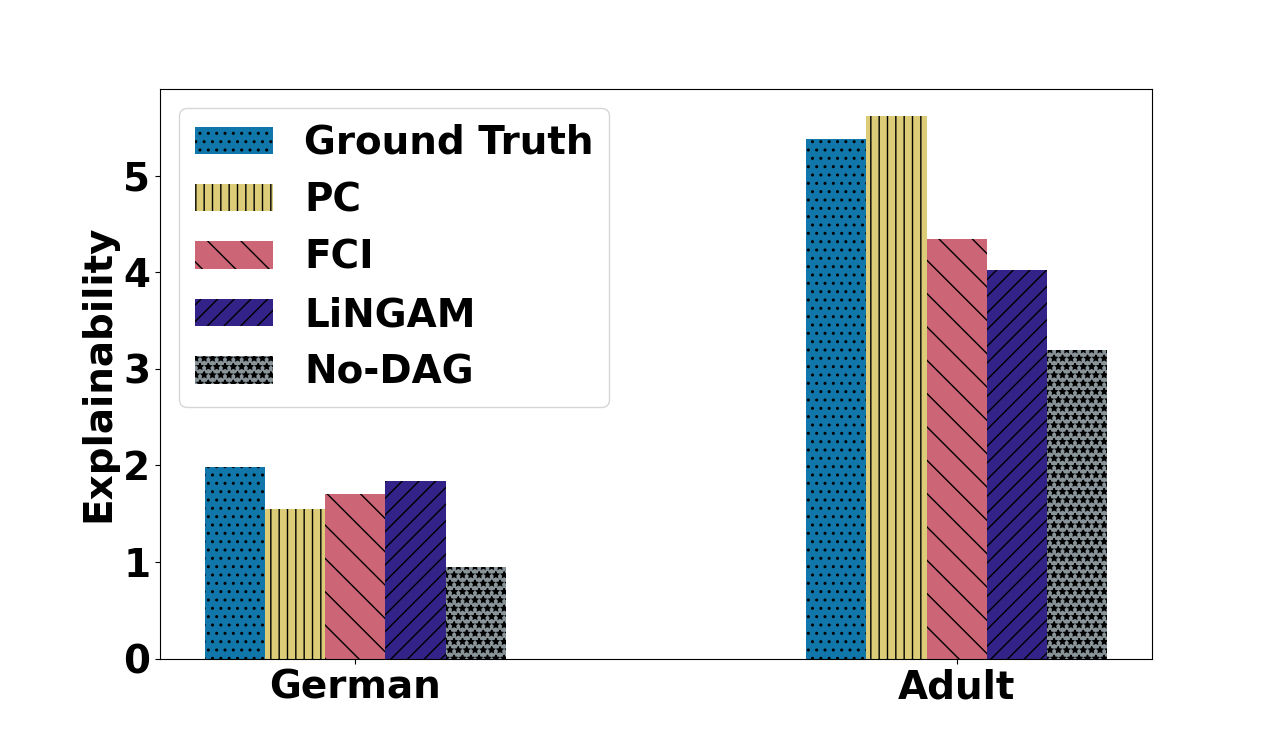}
\vspace{-6mm}
    \caption{Explainability}
    \label{fig:figure4}
  \end{subfigure}
     \begin{subfigure}[b]{0.43\textwidth}
    \centering
    \includegraphics[width=\textwidth]{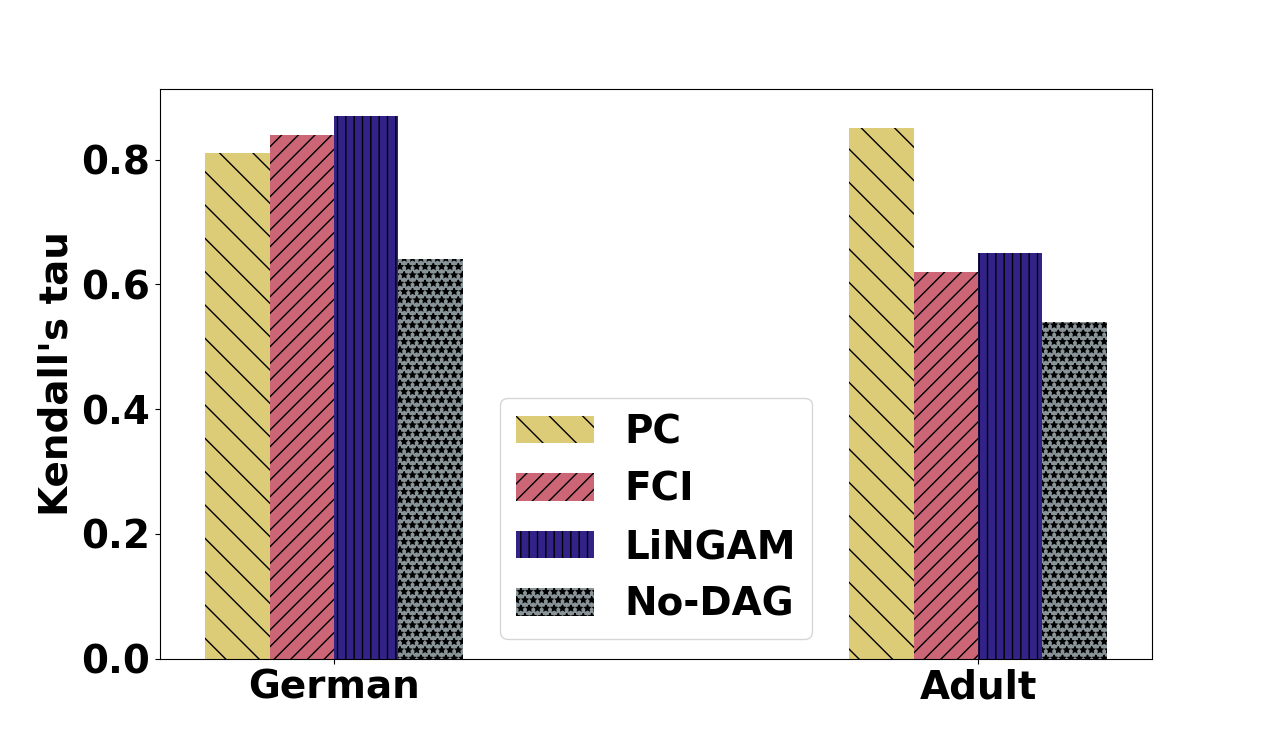}
 \vspace{-6mm}
    \caption{Kendall's $\tau$}
    \label{fig:figure2}
  \end{subfigure}

  \caption{\revc{Modifying the causal DAG.}}
  \label{fig:causal_dag_effect}
\end{figure}

\paratitle{Causal DAG}
\revc{
We depart from the assumption of a given causal DAG and instead use existing solutions to construct DAGs.
We conducted tests with multiple widely used causal discovery algorithms (the PC~\cite{spirtes2000causation}, FCI~\cite{spirtes2000causation}, and LiNGAM~\cite{shimizu2006linear} algorithms), as well as with a simple, straightforward solution. 
In particular, we considered a causal DAG (referred to as No-DAG) wherein all attributes are directly linked to the outcome variable via edges, and no other edges exist within the DAG, similar to the approach taken in~\cite{galhotra2022hyper}.
Our findings reveal that even the employment of basic causal discovery algorithms produces superior results when compared to the absence of any assumed causal DAG (Figure \ref{fig:causal_dag_effect}). We examined the effects on overall explainability and the ranking of treatment patterns when using different causal DAGs (we report the Kendall tau values, comparing the ranking of top-20 treatments based on their CATE values with a ranking obtained using the ground truth causal DAG). Notably, no single causal discovery algorithm outperforms all others, however, all of them outperform the No-DAG baseline.
We note that causal DAGs can originate from various sources, including domain knowledge, GPT, or existing causal discovery methods. This experiment illustrates that while our results do rely on the input causal DAG, the currently available methods for deriving such a DAG are capable of producing meaningful results.}

\section{Limitations and Future Work} 
\label{sec:conc}

\reva{Our framework offers explanations to group-by-avg queries. Restricting to AVG is fundamental for causal explanations unlike non-causal methods.
Our framework aims to find summarized \underline{\em causal explanations} for the query answers in $Q(D)$ using the concept of CATE discussed in Section~\ref{sec:prelim}. The causal estimates by CATE (as shown in Eq.~(\ref{eq:cate}) and (\ref{eq:conf-ate})), inherently use {\em expectations, i.e., weighted average}. Thus, CATEs can be used when considering the aggregate average on the outcome column. 
Aggregate functions like {\tt SUM} or {\tt COUNT} depend on the {\em number of units} satisfying the grouping and treatment patterns, 
which does not have a correspondence with the estimate of causal effect (except that larger groups might reduce variance in the estimate). While other non-causal work on explanations or summarizations for query answers \cite{wu2013scorpion, roy2014formal, miao2019going, li2021putting, lakshmanan2002quotient, wen2018interactive} can support other aggregate functions, methods based on causal estimates focus on average as well \cite{salimi2018bias,youngmann2022explaining,abs-2207-12718}.}

\common{
Our framework currently supports a single-relation database with no dependencies among the tuples. The rationale behind this assumption is to ensure the SUTVA assumption~\cite{rubin2005causal} (as mentioned in Section \cref{sec:prelim}) holds.
Even for a single-table database with dependencies among tuples, this assumption no longer holds. For example, in a {\em Flights} dataset for flight delay, the delay of one flight has an impact on subsequent flights using the same aircraft, and is also dependent on flights leaving and arriving the same airport. Consequently, attempting to compute causal effects on such datasets would result in invalid results. 
When dealing with a single table, treatment patterns are well-defined, and grouping patterns are defined with attributes with FDs to the grouping attributes. However, extending these concepts to multi-table scenarios, where grouping attributes, treatments, and outcomes may originate from distinct tables, poses a challenging formalization task. This complexity arises from the need to address many-to-many relationships and patterns spanning multiple tables. Additionally, in multi-relational databases, intricate dependencies among tuples can exist, potentially violating the SUTVA assumption.
While previous work \cite{SalimiPKGRS20} and \cite{galhotra2022hyper} have extended causal models to accommodate multi-table data and employed them for hypothetical query answering, they have not specifically addressed explanations for groupby-avg queries. The extension of our framework to support multi-relational databases with complex dependencies is an important future work. It is worth noting that prior work on causal explanations \cite{abs-2207-12718, youngmann2022explaining, salimi2018bias} has primarily focused on single tables as well.}

\revc{Lastly, note that attributes used for defining the grouping patterns must be categorical, as they need to exhibit a functional dependency with the grouping attribute. However, attributes used to define treatment patterns can take either continuous or categorical values. 
Handling continuous treatment variables poses specific challenges, but there are standard approaches in causal inference to address them, such as propensity weighting~\cite{rosenbaum1983central} or variable discretization. 
We note that with continuous treatment variables, the search space for potential treatment patterns significantly expands, and thus, while the operation of the treatment mining algorithm remains consistent, the execution times increase.
Here, variable discretization may also be helpful, although this should be done cautiously to preserve the integrity of causal analysis. In our implementation, we did not discretize continuous variables, and we leave this optimization for future work.}

\begin{acks}
This work was partially supported by NSF awards IIS-1552538, IIS-1703431, IIS-2008107, IIS-2147061, and the NSF Convergence Accelerator Program award number 2132318.
\end{acks}

\bibliographystyle{ACM-Reference-Format}
\bibliography{vldb_sample}


\begin{thebibliography}{85}


\ifx \showCODEN    \undefined \def \showCODEN     #1{\unskip}     \fi
\ifx \showDOI      \undefined \def \showDOI       #1{#1}\fi
\ifx \showISBNx    \undefined \def \showISBNx     #1{\unskip}     \fi
\ifx \showISBNxiii \undefined \def \showISBNxiii  #1{\unskip}     \fi
\ifx \showISSN     \undefined \def \showISSN      #1{\unskip}     \fi
\ifx \showLCCN     \undefined \def \showLCCN      #1{\unskip}     \fi
\ifx \shownote     \undefined \def \shownote      #1{#1}          \fi
\ifx \showarticletitle \undefined \def \showarticletitle #1{#1}   \fi
\ifx \showURL      \undefined \def \showURL       {\relax}        \fi
\providecommand\bibfield[2]{#2}
\providecommand\bibinfo[2]{#2}
\providecommand\natexlab[1]{#1}
\providecommand\showeprint[2][]{arXiv:#2}

\bibitem[sta(2021)]%
        {stackoverflowreport}
 \bibinfo{year}{2021}\natexlab{}.
\newblock \bibinfo{title}{2021 Stackoverflow Developer Survey}.
\newblock
\newblock
\newblock
\shownote{\url{https://insights.stackoverflow.com/survey/2021}}.


\bibitem[adu(2021)]%
        {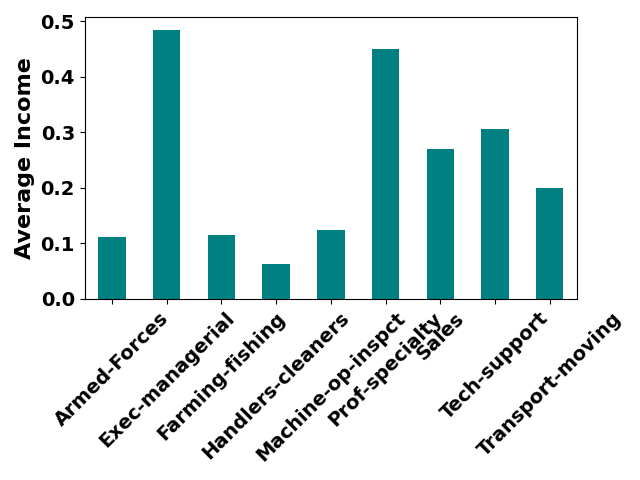}
 \bibinfo{year}{2021}\natexlab{}.
\newblock \bibinfo{title}{Adult Census Income Dataset}.
\newblock
\newblock
\newblock
\shownote{\url{https://www.kaggle.com/datasets/uciml/adult-census-income}}.


\bibitem[dad(2023)]%
        {dads}
 \bibinfo{year}{2023}\natexlab{}.
\newblock \bibinfo{title}{The 19th*}.
\newblock \bibinfo{howpublished}{\url{https://19thnews.org/2023/03/parenthood-stereotypes-gender-pay-gap/}}.
\newblock
\newblock
\shownote{Accessed: 2023-05-18}.


\bibitem[cha(2023)]%
        {chatgpt}
 \bibinfo{year}{2023}\natexlab{}.
\newblock \bibinfo{title}{{OpenAI} Introducing ChatGPT}.
\newblock
\newblock
\newblock
\shownote{\url{https://openai.com/blog/chatgpt}}.


\bibitem[viz(2023)]%
        {vizajobs}
 \bibinfo{year}{2023}\natexlab{}.
\newblock \bibinfo{title}{Viza jobs}.
\newblock \bibinfo{howpublished}{\url{https://vizajobs.com/what-do-technology-jobs-paysalary-insights-and-compensation-factors/}}.
\newblock


\bibitem[sch(2023)]%
        {schufa}
 \bibinfo{year}{2023}\natexlab{}.
\newblock \bibinfo{title}{What is Schufa Score?}
\newblock \bibinfo{howpublished}{\url{https://www.settle-in-berlin.com/what-is-schufa/}}.
\newblock
\newblock
\shownote{Last Updated: 2023-07-15}.


\bibitem[bdt(2029)]%
        {bdtechtalks}
 \bibinfo{year}{2029}\natexlab{}.
\newblock \bibinfo{title}{Tech Talks}.
\newblock \bibinfo{howpublished}{\url{https://bdtechtalks.com/2019/03/29/ageism-in-tech-age-limit-software-developers-face/}}.
\newblock


\bibitem[Agrawal et~al\mbox{.}(1994)]%
        {agrawal1994fast}
\bibfield{author}{\bibinfo{person}{Rakesh Agrawal}, \bibinfo{person}{Ramakrishnan Srikant}, {et~al\mbox{.}}} \bibinfo{year}{1994}\natexlab{}.
\newblock \showarticletitle{Fast algorithms for mining association rules}. In \bibinfo{booktitle}{\emph{Proc. 20th int. conf. very large data bases, VLDB}}, Vol.~\bibinfo{volume}{1215}. Santiago, Chile, \bibinfo{pages}{487--499}.
\newblock


\bibitem[Aljaban(2021)]%
        {aljaban2021analysis}
\bibfield{author}{\bibinfo{person}{Mohamed Aljaban}.} \bibinfo{year}{2021}\natexlab{}.
\newblock \showarticletitle{Analysis of car accidents causes in the usa}.
\newblock  (\bibinfo{year}{2021}).
\newblock


\bibitem[Amsterdamer et~al\mbox{.}(2011)]%
        {DBLP:conf/pods/AmsterdamerDT11}
\bibfield{author}{\bibinfo{person}{Yael Amsterdamer}, \bibinfo{person}{Daniel Deutch}, {and} \bibinfo{person}{Val Tannen}.} \bibinfo{year}{2011}\natexlab{}.
\newblock \showarticletitle{Provenance for aggregate queries}. In \bibinfo{booktitle}{\emph{Proceedings of the 30th {ACM} {SIGMOD-SIGACT-SIGART} Symposium on Principles of Database Systems, {PODS} 2011, June 12-16, 2011, Athens, Greece}}, \bibfield{editor}{\bibinfo{person}{Maurizio Lenzerini} {and} \bibinfo{person}{Thomas Schwentick}} (Eds.). \bibinfo{publisher}{{ACM}}, \bibinfo{pages}{153--164}.
\newblock
\urldef\tempurl%
\url{https://doi.org/10.1145/1989284.1989302}
\showDOI{\tempurl}


\bibitem[Asudeh et~al\mbox{.}(2019)]%
        {asudeh2019assessing}
\bibfield{author}{\bibinfo{person}{Abolfazl Asudeh}, \bibinfo{person}{Zhongjun Jin}, {and} \bibinfo{person}{HV Jagadish}.} \bibinfo{year}{2019}\natexlab{}.
\newblock \showarticletitle{Assessing and remedying coverage for a given dataset}. In \bibinfo{booktitle}{\emph{2019 IEEE 35th International Conference on Data Engineering (ICDE)}}. IEEE, \bibinfo{pages}{554--565}.
\newblock


\bibitem[Asuncion and Newman(2007)]%
        {asuncion2007uci}
\bibfield{author}{\bibinfo{person}{Arthur Asuncion} {and} \bibinfo{person}{David Newman}.} \bibinfo{year}{2007}\natexlab{}.
\newblock \bibinfo{title}{UCI machine learning repository}.
\newblock
\newblock


\bibitem[Banerjee et~al\mbox{.}(2011)]%
        {banerjee2011poor}
\bibfield{author}{\bibinfo{person}{Abhijit~V Banerjee}, \bibinfo{person}{Abhijit Banerjee}, {and} \bibinfo{person}{Esther Duflo}.} \bibinfo{year}{2011}\natexlab{}.
\newblock \bibinfo{booktitle}{\emph{Poor economics: A radical rethinking of the way to fight global poverty}}.
\newblock \bibinfo{publisher}{Public Affairs}.
\newblock


\bibitem[Basu~Roy et~al\mbox{.}(2010)]%
        {basu2010constructing}
\bibfield{author}{\bibinfo{person}{Senjuti Basu~Roy}, \bibinfo{person}{Sihem Amer-Yahia}, \bibinfo{person}{Ashish Chawla}, \bibinfo{person}{Gautam Das}, {and} \bibinfo{person}{Cong Yu}.} \bibinfo{year}{2010}\natexlab{}.
\newblock \showarticletitle{Constructing and exploring composite items}. In \bibinfo{booktitle}{\emph{Proceedings of the 2010 ACM SIGMOD International Conference on Management of data}}. \bibinfo{pages}{843--854}.
\newblock


\bibitem[Bertrand and Mullainathan(2004)]%
        {bertrand2004emily}
\bibfield{author}{\bibinfo{person}{Marianne Bertrand} {and} \bibinfo{person}{Sendhil Mullainathan}.} \bibinfo{year}{2004}\natexlab{}.
\newblock \showarticletitle{Are Emily and Greg more employable than Lakisha and Jamal? A field experiment on labor market discrimination}.
\newblock \bibinfo{journal}{\emph{American economic review}} \bibinfo{volume}{94}, \bibinfo{number}{4} (\bibinfo{year}{2004}), \bibinfo{pages}{991--1013}.
\newblock


\bibitem[Bessa et~al\mbox{.}(2020)]%
        {bessa2020effective}
\bibfield{author}{\bibinfo{person}{Aline Bessa}, \bibinfo{person}{Juliana Freire}, \bibinfo{person}{Tamraparni Dasu}, {and} \bibinfo{person}{Divesh Srivastava}.} \bibinfo{year}{2020}\natexlab{}.
\newblock \showarticletitle{Effective Discovery of Meaningful Outlier Relationships}.
\newblock \bibinfo{journal}{\emph{ACM Transactions on Data Science}} \bibinfo{volume}{1}, \bibinfo{number}{2} (\bibinfo{year}{2020}), \bibinfo{pages}{1--33}.
\newblock


\bibitem[Bidoit et~al\mbox{.}(2014)]%
        {bidoit2014query}
\bibfield{author}{\bibinfo{person}{Nicole Bidoit}, \bibinfo{person}{Melanie Herschel}, {and} \bibinfo{person}{Katerina Tzompanaki}.} \bibinfo{year}{2014}\natexlab{}.
\newblock \showarticletitle{Query-based why-not provenance with nedexplain}. In \bibinfo{booktitle}{\emph{Extending database technology (EDBT)}}.
\newblock


\bibitem[Bu et~al\mbox{.}(2005)]%
        {bu2005mdl}
\bibfield{author}{\bibinfo{person}{Shaofeng Bu}, \bibinfo{person}{Laks~VS Lakshmanan}, {and} \bibinfo{person}{Raymond~T Ng}.} \bibinfo{year}{2005}\natexlab{}.
\newblock \showarticletitle{Mdl summarization with holes}. In \bibinfo{booktitle}{\emph{Proceedings of the 31st international conference on Very large data bases}}. Citeseer, \bibinfo{pages}{433--444}.
\newblock


\bibitem[Chapman and Jagadish(2009)]%
        {chapman2009not}
\bibfield{author}{\bibinfo{person}{Adriane Chapman} {and} \bibinfo{person}{HV Jagadish}.} \bibinfo{year}{2009}\natexlab{}.
\newblock \showarticletitle{Why not?}. In \bibinfo{booktitle}{\emph{Proceedings of the 2009 ACM SIGMOD International Conference on Management of data}}. \bibinfo{pages}{523--534}.
\newblock


\bibitem[Chen and Rudin(2018)]%
        {chen2018optimization}
\bibfield{author}{\bibinfo{person}{Chaofan Chen} {and} \bibinfo{person}{Cynthia Rudin}.} \bibinfo{year}{2018}\natexlab{}.
\newblock \showarticletitle{An optimization approach to learning falling rule lists}. In \bibinfo{booktitle}{\emph{International conference on artificial intelligence and statistics}}. PMLR, \bibinfo{pages}{604--612}.
\newblock


\bibitem[Chetty et~al\mbox{.}(2020)]%
        {chetty2020race}
\bibfield{author}{\bibinfo{person}{Raj Chetty}, \bibinfo{person}{Nathaniel Hendren}, \bibinfo{person}{Maggie~R Jones}, {and} \bibinfo{person}{Sonya~R Porter}.} \bibinfo{year}{2020}\natexlab{}.
\newblock \showarticletitle{Race and economic opportunity in the United States: An intergenerational perspective}.
\newblock \bibinfo{journal}{\emph{The Quarterly Journal of Economics}} \bibinfo{volume}{135}, \bibinfo{number}{2} (\bibinfo{year}{2020}), \bibinfo{pages}{711--783}.
\newblock


\bibitem[Chiappa(2019)]%
        {chiappa2019path}
\bibfield{author}{\bibinfo{person}{Silvia Chiappa}.} \bibinfo{year}{2019}\natexlab{}.
\newblock \showarticletitle{Path-specific counterfactual fairness}. In \bibinfo{booktitle}{\emph{Proceedings of the AAAI Conference on Artificial Intelligence}}, Vol.~\bibinfo{volume}{33}. \bibinfo{pages}{7801--7808}.
\newblock


\bibitem[de~Moura and Bj{\o}rner(2008)]%
        {MouraB08}
\bibfield{author}{\bibinfo{person}{Leonardo~Mendon{\c{c}}a de Moura} {and} \bibinfo{person}{Nikolaj Bj{\o}rner}.} \bibinfo{year}{2008}\natexlab{}.
\newblock \showarticletitle{{Z3:} An Efficient {SMT} Solver}. In \bibinfo{booktitle}{\emph{TACAS}}. \bibinfo{pages}{337--340}.
\newblock


\bibitem[Deutch et~al\mbox{.}(2020)]%
        {DeutchFG20}
\bibfield{author}{\bibinfo{person}{Daniel Deutch}, \bibinfo{person}{Nave Frost}, {and} \bibinfo{person}{Amir Gilad}.} \bibinfo{year}{2020}\natexlab{}.
\newblock \showarticletitle{Explaining Natural Language query results}.
\newblock \bibinfo{journal}{\emph{{VLDB} J.}} \bibinfo{volume}{29}, \bibinfo{number}{1} (\bibinfo{year}{2020}), \bibinfo{pages}{485--508}.
\newblock


\bibitem[Deutch and Gilad(2019)]%
        {DeutchG19}
\bibfield{author}{\bibinfo{person}{Daniel Deutch} {and} \bibinfo{person}{Amir Gilad}.} \bibinfo{year}{2019}\natexlab{}.
\newblock \showarticletitle{Reverse-Engineering Conjunctive Queries from Provenance Examples}. In \bibinfo{booktitle}{\emph{Advances in Database Technology - 22nd International Conference on Extending Database Technology, {EDBT} 2019, Lisbon, Portugal, March 26-29, 2019}}, \bibfield{editor}{\bibinfo{person}{Melanie Herschel}, \bibinfo{person}{Helena Galhardas}, \bibinfo{person}{Berthold Reinwald}, \bibinfo{person}{Irini Fundulaki}, \bibinfo{person}{Carsten Binnig}, {and} \bibinfo{person}{Zoi Kaoudi}} (Eds.). \bibinfo{publisher}{OpenProceedings.org}, \bibinfo{pages}{277--288}.
\newblock
\urldef\tempurl%
\url{https://doi.org/10.5441/002/edbt.2019.25}
\showDOI{\tempurl}


\bibitem[Deutch et~al\mbox{.}(2022)]%
        {DBLP:journals/pvldb/DeutchGMMS22}
\bibfield{author}{\bibinfo{person}{Daniel Deutch}, \bibinfo{person}{Amir Gilad}, \bibinfo{person}{Tova Milo}, \bibinfo{person}{Amit Mualem}, {and} \bibinfo{person}{Amit Somech}.} \bibinfo{year}{2022}\natexlab{}.
\newblock \showarticletitle{{FEDEX:} An Explainability Framework for Data Exploration Steps}.
\newblock \bibinfo{journal}{\emph{Proc. {VLDB} Endow.}} \bibinfo{volume}{15}, \bibinfo{number}{13} (\bibinfo{year}{2022}), \bibinfo{pages}{3854--3868}.
\newblock
\urldef\tempurl%
\url{https://www.vldb.org/pvldb/vol15/p3854-gilad.pdf}
\showURL{%
\tempurl}


\bibitem[El~Gebaly et~al\mbox{.}(2014a)]%
        {el2014interpretable}
\bibfield{author}{\bibinfo{person}{Kareem El~Gebaly}, \bibinfo{person}{Parag Agrawal}, \bibinfo{person}{Lukasz Golab}, \bibinfo{person}{Flip Korn}, {and} \bibinfo{person}{Divesh Srivastava}.} \bibinfo{year}{2014}\natexlab{a}.
\newblock \showarticletitle{Interpretable and informative explanations of outcomes}.
\newblock \bibinfo{journal}{\emph{Proceedings of the VLDB Endowment}} \bibinfo{volume}{8}, \bibinfo{number}{1} (\bibinfo{year}{2014}), \bibinfo{pages}{61--72}.
\newblock


\bibitem[El~Gebaly et~al\mbox{.}(2014b)]%
        {Gebaly+2014-expltable}
\bibfield{author}{\bibinfo{person}{Kareem El~Gebaly}, \bibinfo{person}{Parag Agrawal}, \bibinfo{person}{Lukasz Golab}, \bibinfo{person}{Flip Korn}, {and} \bibinfo{person}{Divesh Srivastava}.} \bibinfo{year}{2014}\natexlab{b}.
\newblock \showarticletitle{Interpretable and Informative Explanations of Outcomes}.
\newblock \bibinfo{journal}{\emph{Proc. VLDB Endow.}} \bibinfo{volume}{8}, \bibinfo{number}{1} (\bibinfo{date}{sep} \bibinfo{year}{2014}), \bibinfo{pages}{61–72}.
\newblock
\showISSN{2150-8097}
\urldef\tempurl%
\url{https://doi.org/10.14778/2735461.2735467}
\showDOI{\tempurl}


\bibitem[Flood et~al\mbox{.}(2015)]%
        {flood2015integrated}
\bibfield{author}{\bibinfo{person}{Sarah Flood}, \bibinfo{person}{Miriam King}, \bibinfo{person}{Steven Ruggles}, {and} \bibinfo{person}{J~Robert Warren}.} \bibinfo{year}{2015}\natexlab{}.
\newblock \showarticletitle{Integrated public use microdata series, current population survey: Version 9.0.[Machine-readable database]}.
\newblock \bibinfo{journal}{\emph{Minneapolis: University of Minnesota}}  \bibinfo{volume}{1} (\bibinfo{year}{2015}).
\newblock


\bibitem[Galhotra et~al\mbox{.}(2022)]%
        {galhotra2022hyper}
\bibfield{author}{\bibinfo{person}{Sainyam Galhotra}, \bibinfo{person}{Amir Gilad}, \bibinfo{person}{Sudeepa Roy}, {and} \bibinfo{person}{Babak Salimi}.} \bibinfo{year}{2022}\natexlab{}.
\newblock \showarticletitle{Hyper: Hypothetical reasoning with what-if and how-to queries using a probabilistic causal approach}. In \bibinfo{booktitle}{\emph{Proceedings of the 2022 International Conference on Management of Data}}. \bibinfo{pages}{1598--1611}.
\newblock


\bibitem[Greenland and Robins(1999)]%
        {greenland1999epidemiology}
\bibfield{author}{\bibinfo{person}{Sander Greenland} {and} \bibinfo{person}{James~M Robins}.} \bibinfo{year}{1999}\natexlab{}.
\newblock \showarticletitle{Epidemiology, justice, and the probability of causation}.
\newblock \bibinfo{journal}{\emph{Jurimetrics}}  \bibinfo{volume}{40} (\bibinfo{year}{1999}), \bibinfo{pages}{321}.
\newblock


\bibitem[Holland(1986)]%
        {holland1986statistics}
\bibfield{author}{\bibinfo{person}{Paul~W Holland}.} \bibinfo{year}{1986}\natexlab{}.
\newblock \showarticletitle{Statistics and causal inference}.
\newblock \bibinfo{journal}{\emph{Journal of the American statistical Association}} \bibinfo{volume}{81}, \bibinfo{number}{396} (\bibinfo{year}{1986}), \bibinfo{pages}{945--960}.
\newblock


\bibitem[Kim et~al\mbox{.}(2020)]%
        {kim2020summarizing}
\bibfield{author}{\bibinfo{person}{Alexandra Kim}, \bibinfo{person}{Laks~VS Lakshmanan}, {and} \bibinfo{person}{Divesh Srivastava}.} \bibinfo{year}{2020}\natexlab{}.
\newblock \showarticletitle{Summarizing hierarchical multidimensional data}. In \bibinfo{booktitle}{\emph{2020 IEEE 36th International Conference on Data Engineering (ICDE)}}. IEEE, \bibinfo{pages}{877--888}.
\newblock


\bibitem[Kim et~al\mbox{.}(2014)]%
        {kim2014bayesian}
\bibfield{author}{\bibinfo{person}{Been Kim}, \bibinfo{person}{Cynthia Rudin}, {and} \bibinfo{person}{Julie~A Shah}.} \bibinfo{year}{2014}\natexlab{}.
\newblock \showarticletitle{The bayesian case model: A generative approach for case-based reasoning and prototype classification}.
\newblock \bibinfo{journal}{\emph{Advances in neural information processing systems}}  \bibinfo{volume}{27} (\bibinfo{year}{2014}).
\newblock


\bibitem[Lakkaraju et~al\mbox{.}(2016)]%
        {lakkaraju2016interpretable}
\bibfield{author}{\bibinfo{person}{Himabindu Lakkaraju}, \bibinfo{person}{Stephen~H Bach}, {and} \bibinfo{person}{Jure Leskovec}.} \bibinfo{year}{2016}\natexlab{}.
\newblock \showarticletitle{Interpretable decision sets: A joint framework for description and prediction}. In \bibinfo{booktitle}{\emph{Proceedings of the 22nd ACM SIGKDD international conference on knowledge discovery and data mining}}. \bibinfo{pages}{1675--1684}.
\newblock


\bibitem[Lakshmanan et~al\mbox{.}(2002a)]%
        {lakshmanan2002generalized}
\bibfield{author}{\bibinfo{person}{Laks~VS Lakshmanan}, \bibinfo{person}{Raymond~T Ng}, \bibinfo{person}{Christine~Xing Wang}, \bibinfo{person}{Xiaodong Zhou}, {and} \bibinfo{person}{Theodore~J Johnson}.} \bibinfo{year}{2002}\natexlab{a}.
\newblock \showarticletitle{The generalized MDL approach for summarization}. In \bibinfo{booktitle}{\emph{VLDB'02: Proceedings of the 28th International Conference on Very Large Databases}}. Elsevier, \bibinfo{pages}{766--777}.
\newblock


\bibitem[Lakshmanan et~al\mbox{.}(2002b)]%
        {lakshmanan2002quotient}
\bibfield{author}{\bibinfo{person}{Laks~VS Lakshmanan}, \bibinfo{person}{Jian Pei}, {and} \bibinfo{person}{Jiawei Han}.} \bibinfo{year}{2002}\natexlab{b}.
\newblock \showarticletitle{Quotient cube: How to summarize the semantics of a data cube}. In \bibinfo{booktitle}{\emph{VLDB'02: Proceedings of the 28th International Conference on Very Large Databases}}. Elsevier, \bibinfo{pages}{778--789}.
\newblock


\bibitem[Lee et~al\mbox{.}(2020)]%
        {lee2020approximate}
\bibfield{author}{\bibinfo{person}{Seokki Lee}, \bibinfo{person}{Bertram Lud{\"a}scher}, {and} \bibinfo{person}{Boris Glavic}.} \bibinfo{year}{2020}\natexlab{}.
\newblock \showarticletitle{Approximate summaries for why and why-not provenance (extended version)}.
\newblock \bibinfo{journal}{\emph{arXiv preprint arXiv:2002.00084}} (\bibinfo{year}{2020}).
\newblock


\bibitem[Li et~al\mbox{.}(2021)]%
        {li2021putting}
\bibfield{author}{\bibinfo{person}{Chenjie Li}, \bibinfo{person}{Zhengjie Miao}, \bibinfo{person}{Qitian Zeng}, \bibinfo{person}{Boris Glavic}, {and} \bibinfo{person}{Sudeepa Roy}.} \bibinfo{year}{2021}\natexlab{}.
\newblock \showarticletitle{Putting Things into Context: Rich Explanations for Query Answers using Join Graphs}. In \bibinfo{booktitle}{\emph{Proceedings of the 2021 International Conference on Management of Data}}. \bibinfo{pages}{1051--1063}.
\newblock


\bibitem[Lin et~al\mbox{.}(2021)]%
        {lin2021detecting}
\bibfield{author}{\bibinfo{person}{Yin Lin}, \bibinfo{person}{Brit Youngmann}, \bibinfo{person}{Yuval Moskovitch}, \bibinfo{person}{HV Jagadish}, {and} \bibinfo{person}{Tova Milo}.} \bibinfo{year}{2021}\natexlab{}.
\newblock \showarticletitle{On detecting cherry-picked generalizations}.
\newblock \bibinfo{journal}{\emph{Proceedings of the VLDB Endowment}} \bibinfo{volume}{15}, \bibinfo{number}{1} (\bibinfo{year}{2021}), \bibinfo{pages}{59--71}.
\newblock


\bibitem[Livshits et~al\mbox{.}(2020)]%
        {LivshitsBKS20}
\bibfield{author}{\bibinfo{person}{Ester Livshits}, \bibinfo{person}{Leopoldo~E. Bertossi}, \bibinfo{person}{Benny Kimelfeld}, {and} \bibinfo{person}{Moshe Sebag}.} \bibinfo{year}{2020}\natexlab{}.
\newblock \showarticletitle{The Shapley Value of Tuples in Query Answering}. In \bibinfo{booktitle}{\emph{ICDT}}, Vol.~\bibinfo{volume}{155}. \bibinfo{pages}{20:1--20:19}.
\newblock


\bibitem[Lou et~al\mbox{.}(2013)]%
        {lou2013accurate}
\bibfield{author}{\bibinfo{person}{Yin Lou}, \bibinfo{person}{Rich Caruana}, \bibinfo{person}{Johannes Gehrke}, {and} \bibinfo{person}{Giles Hooker}.} \bibinfo{year}{2013}\natexlab{}.
\newblock \showarticletitle{Accurate intelligible models with pairwise interactions}. In \bibinfo{booktitle}{\emph{Proceedings of the 19th ACM SIGKDD international conference on Knowledge discovery and data mining}}. \bibinfo{pages}{623--631}.
\newblock


\bibitem[Ma et~al\mbox{.}(2023)]%
        {abs-2207-12718}
\bibfield{author}{\bibinfo{person}{Pingchuan Ma}, \bibinfo{person}{Rui Ding}, \bibinfo{person}{Shuai Wang}, \bibinfo{person}{Shi Han}, {and} \bibinfo{person}{Dongmei Zhang}.} \bibinfo{year}{2023}\natexlab{}.
\newblock \showarticletitle{XInsight: EXplainable Data Analysis Through The Lens of Causality}.
\newblock \bibinfo{journal}{\emph{Proc. ACM Manag. Data}}, Article \bibinfo{articleno}{156} (\bibinfo{date}{jun} \bibinfo{year}{2023}), \bibinfo{numpages}{27}~pages.
\newblock


\bibitem[Meliou et~al\mbox{.}(2009)]%
        {meliou2009so}
\bibfield{author}{\bibinfo{person}{Alexandra Meliou}, \bibinfo{person}{Wolfgang Gatterbauer}, \bibinfo{person}{Katherine~F Moore}, {and} \bibinfo{person}{Dan Suciu}.} \bibinfo{year}{2009}\natexlab{}.
\newblock \showarticletitle{Why so? or why no? functional causality for explaining query answers}.
\newblock \bibinfo{journal}{\emph{arXiv preprint arXiv:0912.5340}} (\bibinfo{year}{2009}).
\newblock


\bibitem[Meliou et~al\mbox{.}(2010)]%
        {meliou2010complexity}
\bibfield{author}{\bibinfo{person}{Alexandra Meliou}, \bibinfo{person}{Wolfgang Gatterbauer}, \bibinfo{person}{Katherine~F Moore}, {and} \bibinfo{person}{Dan Suciu}.} \bibinfo{year}{2010}\natexlab{}.
\newblock \showarticletitle{The complexity of causality and responsibility for query answers and non-answers}.
\newblock \bibinfo{journal}{\emph{arXiv preprint arXiv:1009.2021}} (\bibinfo{year}{2010}).
\newblock


\bibitem[Miao et~al\mbox{.}(2019)]%
        {miao2019going}
\bibfield{author}{\bibinfo{person}{Zhengjie Miao}, \bibinfo{person}{Qitian Zeng}, \bibinfo{person}{Boris Glavic}, {and} \bibinfo{person}{Sudeepa Roy}.} \bibinfo{year}{2019}\natexlab{}.
\newblock \showarticletitle{Going beyond provenance: Explaining query answers with pattern-based counterbalances}. In \bibinfo{booktitle}{\emph{Proceedings of the 2019 International Conference on Management of Data}}. \bibinfo{pages}{485--502}.
\newblock


\bibitem[Milo et~al\mbox{.}(2020)]%
        {milo2020contribution}
\bibfield{author}{\bibinfo{person}{Tova Milo}, \bibinfo{person}{Yuval Moskovitch}, {and} \bibinfo{person}{Brit Youngmann}.} \bibinfo{year}{2020}\natexlab{}.
\newblock \showarticletitle{Contribution Maximization in Probabilistic Datalog}. In \bibinfo{booktitle}{\emph{2020 IEEE 36th International Conference on Data Engineering (ICDE)}}. IEEE, \bibinfo{pages}{817--828}.
\newblock


\bibitem[Moosavi et~al\mbox{.}(2019)]%
        {moosavi2019countrywide}
\bibfield{author}{\bibinfo{person}{Sobhan Moosavi}, \bibinfo{person}{Mohammad~Hossein Samavatian}, \bibinfo{person}{Srinivasan Parthasarathy}, {and} \bibinfo{person}{Rajiv Ramnath}.} \bibinfo{year}{2019}\natexlab{}.
\newblock \showarticletitle{A countrywide traffic accident dataset}.
\newblock \bibinfo{journal}{\emph{arXiv preprint arXiv:1906.05409}} (\bibinfo{year}{2019}).
\newblock


\bibitem[Nilsson(1982)]%
        {nilsson1982effects}
\bibfield{author}{\bibinfo{person}{G{\"o}ran Nilsson}.} \bibinfo{year}{1982}\natexlab{}.
\newblock \bibinfo{title}{Effects of speed limits on traffic accidents in Sweden}.
\newblock
\newblock


\bibitem[Pardon and Average(2013)]%
        {pardon2013effectiveness}
\bibfield{author}{\bibinfo{person}{Ndhlovu Pardon} {and} \bibinfo{person}{Chigwenya Average}.} \bibinfo{year}{2013}\natexlab{}.
\newblock \showarticletitle{The effectiveness of traffic calming measures in reducing road carnage in masvingo urban}.
\newblock \bibinfo{journal}{\emph{International Journal}} \bibinfo{volume}{3}, \bibinfo{number}{2} (\bibinfo{year}{2013}), \bibinfo{pages}{2305--1493}.
\newblock


\bibitem[Pearl(2009)]%
        {pearl2009causal}
\bibfield{author}{\bibinfo{person}{Judea Pearl}.} \bibinfo{year}{2009}\natexlab{}.
\newblock \showarticletitle{Causal inference in statistics: An overview}.
\newblock  (\bibinfo{year}{2009}).
\newblock


\bibitem[Raghavan and Tompson(1987)]%
        {raghavan1987randomized}
\bibfield{author}{\bibinfo{person}{Prabhakar Raghavan} {and} \bibinfo{person}{Clark~D Tompson}.} \bibinfo{year}{1987}\natexlab{}.
\newblock \showarticletitle{Randomized rounding: a technique for provably good algorithms and algorithmic proofs}.
\newblock \bibinfo{journal}{\emph{Combinatorica}} \bibinfo{volume}{7}, \bibinfo{number}{4} (\bibinfo{year}{1987}), \bibinfo{pages}{365--374}.
\newblock


\bibitem[Reshef et~al\mbox{.}(2020)]%
        {ReshefKL20}
\bibfield{author}{\bibinfo{person}{Alon Reshef}, \bibinfo{person}{Benny Kimelfeld}, {and} \bibinfo{person}{Ester Livshits}.} \bibinfo{year}{2020}\natexlab{}.
\newblock \showarticletitle{The Impact of Negation on the Complexity of the Shapley Value in Conjunctive Queries}. In \bibinfo{booktitle}{\emph{PODS}}, \bibfield{editor}{\bibinfo{person}{Dan Suciu}, \bibinfo{person}{Yufei Tao}, {and} \bibinfo{person}{Zhewei Wei}} (Eds.). \bibinfo{pages}{285--297}.
\newblock


\bibitem[Robins et~al\mbox{.}(2000)]%
        {robins2000marginal}
\bibfield{author}{\bibinfo{person}{James~M Robins}, \bibinfo{person}{Miguel~Angel Hernan}, {and} \bibinfo{person}{Babette Brumback}.} \bibinfo{year}{2000}\natexlab{}.
\newblock \bibinfo{title}{Marginal structural models and causal inference in epidemiology}.
\newblock
\newblock


\bibitem[Rosenbaum and Rubin(1983)]%
        {rosenbaum1983central}
\bibfield{author}{\bibinfo{person}{Paul~R Rosenbaum} {and} \bibinfo{person}{Donald~B Rubin}.} \bibinfo{year}{1983}\natexlab{}.
\newblock \showarticletitle{The central role of the propensity score in observational studies for causal effects}.
\newblock \bibinfo{journal}{\emph{Biometrika}} \bibinfo{volume}{70}, \bibinfo{number}{1} (\bibinfo{year}{1983}), \bibinfo{pages}{41--55}.
\newblock


\bibitem[Roy et~al\mbox{.}(2015)]%
        {roy2015explaining}
\bibfield{author}{\bibinfo{person}{Sudeepa Roy}, \bibinfo{person}{Laurel Orr}, {and} \bibinfo{person}{Dan Suciu}.} \bibinfo{year}{2015}\natexlab{}.
\newblock \showarticletitle{Explaining query answers with explanation-ready databases}.
\newblock \bibinfo{journal}{\emph{Proceedings of the VLDB Endowment}} \bibinfo{volume}{9}, \bibinfo{number}{4} (\bibinfo{year}{2015}), \bibinfo{pages}{348--359}.
\newblock


\bibitem[Roy and Suciu(2014)]%
        {roy2014formal}
\bibfield{author}{\bibinfo{person}{Sudeepa Roy} {and} \bibinfo{person}{Dan Suciu}.} \bibinfo{year}{2014}\natexlab{}.
\newblock \showarticletitle{A formal approach to finding explanations for database queries}. In \bibinfo{booktitle}{\emph{Proceedings of the 2014 ACM SIGMOD international conference on Management of data}}. \bibinfo{pages}{1579--1590}.
\newblock


\bibitem[Rubin(1971)]%
        {rubin1971use}
\bibfield{author}{\bibinfo{person}{Donald~Bruce Rubin}.} \bibinfo{year}{1971}\natexlab{}.
\newblock \emph{\bibinfo{title}{The use of matched sampling and regression adjustment in observational studies}}.
\newblock \bibinfo{thesistype}{Ph.\,D. Dissertation}. \bibinfo{school}{Harvard University}.
\newblock


\bibitem[Rubin(2005)]%
        {rubin2005causal}
\bibfield{author}{\bibinfo{person}{Donald~B Rubin}.} \bibinfo{year}{2005}\natexlab{}.
\newblock \showarticletitle{Causal inference using potential outcomes: Design, modeling, decisions}.
\newblock \bibinfo{journal}{\emph{J. Amer. Statist. Assoc.}} \bibinfo{volume}{100}, \bibinfo{number}{469} (\bibinfo{year}{2005}), \bibinfo{pages}{322--331}.
\newblock


\bibitem[Sagi and Rokach(2021)]%
        {sagi2021approximating}
\bibfield{author}{\bibinfo{person}{Omer Sagi} {and} \bibinfo{person}{Lior Rokach}.} \bibinfo{year}{2021}\natexlab{}.
\newblock \showarticletitle{Approximating XGBoost with an interpretable decision tree}.
\newblock \bibinfo{journal}{\emph{Information Sciences}}  \bibinfo{volume}{572} (\bibinfo{year}{2021}), \bibinfo{pages}{522--542}.
\newblock


\bibitem[Salimi et~al\mbox{.}(2018)]%
        {salimi2018bias}
\bibfield{author}{\bibinfo{person}{Babak Salimi}, \bibinfo{person}{Johannes Gehrke}, {and} \bibinfo{person}{Dan Suciu}.} \bibinfo{year}{2018}\natexlab{}.
\newblock \showarticletitle{Bias in olap queries: Detection, explanation, and removal}. In \bibinfo{booktitle}{\emph{Proceedings of the 2018 International Conference on Management of Data}}. \bibinfo{pages}{1021--1035}.
\newblock


\bibitem[Salimi et~al\mbox{.}(2020)]%
        {SalimiPKGRS20}
\bibfield{author}{\bibinfo{person}{Babak Salimi}, \bibinfo{person}{Harsh Parikh}, \bibinfo{person}{Moe Kayali}, \bibinfo{person}{Lise Getoor}, \bibinfo{person}{Sudeepa Roy}, {and} \bibinfo{person}{Dan Suciu}.} \bibinfo{year}{2020}\natexlab{}.
\newblock \showarticletitle{Causal Relational Learning}. In \bibinfo{booktitle}{\emph{Proceedings of the 2020 International Conference on Management of Data, {SIGMOD} Conference 2020, online conference [Portland, OR, USA], June 14-19, 2020}}, \bibfield{editor}{\bibinfo{person}{David Maier}, \bibinfo{person}{Rachel Pottinger}, \bibinfo{person}{AnHai Doan}, \bibinfo{person}{Wang{-}Chiew Tan}, \bibinfo{person}{Abdussalam Alawini}, {and} \bibinfo{person}{Hung~Q. Ngo}} (Eds.). \bibinfo{publisher}{{ACM}}, \bibinfo{pages}{241--256}.
\newblock
\urldef\tempurl%
\url{https://doi.org/10.1145/3318464.3389759}
\showDOI{\tempurl}


\bibitem[Sathe and Sarawagi(2001)]%
        {sathe2001intelligent}
\bibfield{author}{\bibinfo{person}{Gayatri Sathe} {and} \bibinfo{person}{Sunita Sarawagi}.} \bibinfo{year}{2001}\natexlab{}.
\newblock \showarticletitle{Intelligent rollups in multidimensional OLAP data}. In \bibinfo{booktitle}{\emph{VLDB}}. \bibinfo{pages}{307--316}.
\newblock


\bibitem[Schielzeth(2010)]%
        {schielzeth2010simple}
\bibfield{author}{\bibinfo{person}{Holger Schielzeth}.} \bibinfo{year}{2010}\natexlab{}.
\newblock \showarticletitle{Simple means to improve the interpretability of regression coefficients}.
\newblock \bibinfo{journal}{\emph{Methods in Ecology and Evolution}} \bibinfo{volume}{1}, \bibinfo{number}{2} (\bibinfo{year}{2010}), \bibinfo{pages}{103--113}.
\newblock


\bibitem[Sharma and Kiciman(2020)]%
        {dowhypaper}
\bibfield{author}{\bibinfo{person}{Amit Sharma} {and} \bibinfo{person}{Emre Kiciman}.} \bibinfo{year}{2020}\natexlab{}.
\newblock \showarticletitle{DoWhy: An End-to-End Library for Causal Inference}.
\newblock \bibinfo{journal}{\emph{arXiv preprint arXiv:2011.04216}} (\bibinfo{year}{2020}).
\newblock


\bibitem[Shimizu et~al\mbox{.}(2006)]%
        {shimizu2006linear}
\bibfield{author}{\bibinfo{person}{Shohei Shimizu}, \bibinfo{person}{Patrik~O Hoyer}, \bibinfo{person}{Aapo Hyv{\"a}rinen}, \bibinfo{person}{Antti Kerminen}, {and} \bibinfo{person}{Michael Jordan}.} \bibinfo{year}{2006}\natexlab{}.
\newblock \showarticletitle{A linear non-Gaussian acyclic model for causal discovery.}
\newblock \bibinfo{journal}{\emph{Journal of Machine Learning Research}} \bibinfo{volume}{7}, \bibinfo{number}{10} (\bibinfo{year}{2006}).
\newblock


\bibitem[Shipley(2016)]%
        {shipley2016cause}
\bibfield{author}{\bibinfo{person}{Bill Shipley}.} \bibinfo{year}{2016}\natexlab{}.
\newblock \bibinfo{booktitle}{\emph{Cause and correlation in biology: a user's guide to path analysis, structural equations and causal inference with R}}.
\newblock \bibinfo{publisher}{Cambridge university press}.
\newblock


\bibitem[Spirtes et~al\mbox{.}(2000)]%
        {spirtes2000causation}
\bibfield{author}{\bibinfo{person}{P. Spirtes} {et~al\mbox{.}}} \bibinfo{year}{2000}\natexlab{}.
\newblock \bibinfo{booktitle}{\emph{Causation, prediction, and search}}.
\newblock \bibinfo{publisher}{MIT press}.
\newblock


\bibitem[Steurer(2014)]%
        {maxcoverage}
\bibfield{author}{\bibinfo{person}{David Steurer}.} \bibinfo{year}{2014}\natexlab{}.
\newblock \bibinfo{title}{Max Coverage—Randomized LP Rounding}.
\newblock \bibinfo{howpublished}{\url{ http://www.cs.cornell.edu/courses/cs4820/2014sp/notes/maxcoverage.pdf}}.
\newblock


\bibitem[Tao et~al\mbox{.}(2022)]%
        {tao2022dpxplain}
\bibfield{author}{\bibinfo{person}{Yuchao Tao}, \bibinfo{person}{Amir Gilad}, \bibinfo{person}{Ashwin Machanavajjhala}, {and} \bibinfo{person}{Sudeepa Roy}.} \bibinfo{year}{2022}\natexlab{}.
\newblock \showarticletitle{DPXPlain: Privately Explaining Aggregate Query Answers}.
\newblock \bibinfo{journal}{\emph{Proc. {VLDB} Endow.}} \bibinfo{volume}{16}, \bibinfo{number}{1} (\bibinfo{year}{2022}), \bibinfo{pages}{113--126}.
\newblock
\urldef\tempurl%
\url{https://www.vldb.org/pvldb/vol16/p113-tao.pdf}
\showURL{%
\tempurl}


\bibitem[ten Cate et~al\mbox{.}(2015)]%
        {ten2015high}
\bibfield{author}{\bibinfo{person}{Balder ten Cate}, \bibinfo{person}{Cristina Civili}, \bibinfo{person}{Evgeny Sherkhonov}, {and} \bibinfo{person}{Wang-Chiew Tan}.} \bibinfo{year}{2015}\natexlab{}.
\newblock \showarticletitle{High-level why-not explanations using ontologies}. In \bibinfo{booktitle}{\emph{Proceedings of the 34th ACM SIGMOD-SIGACT-SIGAI Symposium on Principles of Database Systems}}. \bibinfo{pages}{31--43}.
\newblock


\bibitem[Tian and Pearl(2000)]%
        {tian2000probabilities}
\bibfield{author}{\bibinfo{person}{Jin Tian} {and} \bibinfo{person}{Judea Pearl}.} \bibinfo{year}{2000}\natexlab{}.
\newblock \showarticletitle{Probabilities of causation: Bounds and identification}.
\newblock \bibinfo{journal}{\emph{Annals of Mathematics and Artificial Intelligence}} \bibinfo{volume}{28}, \bibinfo{number}{1-4} (\bibinfo{year}{2000}), \bibinfo{pages}{287--313}.
\newblock


\bibitem[Tramer et~al\mbox{.}(2017)]%
        {tramer2017fairtest}
\bibfield{author}{\bibinfo{person}{Florian Tramer}, \bibinfo{person}{Vaggelis Atlidakis}, \bibinfo{person}{Roxana Geambasu}, \bibinfo{person}{Daniel Hsu}, \bibinfo{person}{Jean-Pierre Hubaux}, \bibinfo{person}{Mathias Humbert}, \bibinfo{person}{Ari Juels}, {and} \bibinfo{person}{Huang Lin}.} \bibinfo{year}{2017}\natexlab{}.
\newblock \showarticletitle{Fairtest: Discovering unwarranted associations in data-driven applications}. In \bibinfo{booktitle}{\emph{2017 IEEE European Symposium on Security and Privacy (EuroS\&P)}}. IEEE, \bibinfo{pages}{401--416}.
\newblock


\bibitem[Vardi(1982)]%
        {Vardi82}
\bibfield{author}{\bibinfo{person}{Moshe~Y. Vardi}.} \bibinfo{year}{1982}\natexlab{}.
\newblock \showarticletitle{The Complexity of Relational Query Languages (Extended Abstract)}. In \bibinfo{booktitle}{\emph{Proceedings of the Fourteenth Annual ACM Symposium on Theory of Computing}} (San Francisco, California, USA) \emph{(\bibinfo{series}{STOC '82})}. \bibinfo{publisher}{ACM}, \bibinfo{address}{New York, NY, USA}, \bibinfo{pages}{137--146}.
\newblock
\showISBNx{0-89791-070-2}
\urldef\tempurl%
\url{https://doi.org/10.1145/800070.802186}
\showDOI{\tempurl}


\bibitem[Wager and Athey(2018)]%
        {wager2018estimation}
\bibfield{author}{\bibinfo{person}{Stefan Wager} {and} \bibinfo{person}{Susan Athey}.} \bibinfo{year}{2018}\natexlab{}.
\newblock \showarticletitle{Estimation and inference of heterogeneous treatment effects using random forests}.
\newblock \bibinfo{journal}{\emph{J. Amer. Statist. Assoc.}} \bibinfo{volume}{113}, \bibinfo{number}{523} (\bibinfo{year}{2018}), \bibinfo{pages}{1228--1242}.
\newblock


\bibitem[Wen et~al\mbox{.}(2018)]%
        {wen2018interactive}
\bibfield{author}{\bibinfo{person}{Yuhao Wen}, \bibinfo{person}{Xiaodan Zhu}, \bibinfo{person}{Sudeepa Roy}, {and} \bibinfo{person}{Jun Yang}.} \bibinfo{year}{2018}\natexlab{}.
\newblock \showarticletitle{Interactive summarization and exploration of top aggregate query answers}. In \bibinfo{booktitle}{\emph{Proceedings of the VLDB Endowment. International Conference on Very Large Data Bases}}, Vol.~\bibinfo{volume}{11}. NIH Public Access, \bibinfo{pages}{2196}.
\newblock


\bibitem[Wu and Madden(2013)]%
        {wu2013scorpion}
\bibfield{author}{\bibinfo{person}{Eugene Wu} {and} \bibinfo{person}{Samuel Madden}.} \bibinfo{year}{2013}\natexlab{}.
\newblock \showarticletitle{Scorpion: Explaining away outliers in aggregate queries}.
\newblock  (\bibinfo{year}{2013}).
\newblock


\bibitem[Xie et~al\mbox{.}(2012)]%
        {xie2012estimating}
\bibfield{author}{\bibinfo{person}{Yu Xie}, \bibinfo{person}{Jennie~E Brand}, {and} \bibinfo{person}{Ben Jann}.} \bibinfo{year}{2012}\natexlab{}.
\newblock \showarticletitle{Estimating heterogeneous treatment effects with observational data}.
\newblock \bibinfo{journal}{\emph{Sociological methodology}} \bibinfo{volume}{42}, \bibinfo{number}{1} (\bibinfo{year}{2012}), \bibinfo{pages}{314--347}.
\newblock


\bibitem[Yang et~al\mbox{.}(2017)]%
        {yang2017scalable}
\bibfield{author}{\bibinfo{person}{Hongyu Yang}, \bibinfo{person}{Cynthia Rudin}, {and} \bibinfo{person}{Margo Seltzer}.} \bibinfo{year}{2017}\natexlab{}.
\newblock \showarticletitle{Scalable Bayesian rule lists}. In \bibinfo{booktitle}{\emph{International conference on machine learning}}. PMLR, \bibinfo{pages}{3921--3930}.
\newblock


\bibitem[Youngmann et~al\mbox{.}(2022)]%
        {DBLP:journals/pvldb/YoungmannAP22}
\bibfield{author}{\bibinfo{person}{Brit Youngmann}, \bibinfo{person}{Sihem Amer{-}Yahia}, {and} \bibinfo{person}{Aur{\'{e}}lien Personnaz}.} \bibinfo{year}{2022}\natexlab{}.
\newblock \showarticletitle{Guided Exploration of Data Summaries}.
\newblock \bibinfo{journal}{\emph{Proc. {VLDB} Endow.}} \bibinfo{volume}{15}, \bibinfo{number}{9} (\bibinfo{year}{2022}).
\newblock


\bibitem[Youngmann et~al\mbox{.}(2023a)]%
        {fullversion}
\bibfield{author}{\bibinfo{person}{Brit Youngmann}, \bibinfo{person}{Michael Cafarella}, \bibinfo{person}{Amir Gilad}, {and} \bibinfo{person}{Sudeepa Roy}.} \bibinfo{year}{2023}\natexlab{a}.
\newblock \bibinfo{title}{Techinical Report}.
\newblock
\newblock
\urldef\tempurl%
\url{https://anonymous.4open.science/r/Explanation_Summarization-F736}
\showURL{%
\tempurl}


\bibitem[Youngmann et~al\mbox{.}(2023b)]%
        {youngmann2023nexus}
\bibfield{author}{\bibinfo{person}{Brit Youngmann}, \bibinfo{person}{Michael Cafarella}, \bibinfo{person}{Yuval Moskovitch}, {and} \bibinfo{person}{Babak Salimi}.} \bibinfo{year}{2023}\natexlab{b}.
\newblock \showarticletitle{NEXUS: On Explaining Confounding Bias}. In \bibinfo{booktitle}{\emph{Companion of the 2023 International Conference on Management of Data}}. \bibinfo{pages}{171--174}.
\newblock


\bibitem[Youngmann et~al\mbox{.}(2023c)]%
        {youngmann2022explaining}
\bibfield{author}{\bibinfo{person}{Brit Youngmann}, \bibinfo{person}{Michael Cafarella}, \bibinfo{person}{Yuval Moskovitch}, {and} \bibinfo{person}{Babak Salimi}.} \bibinfo{year}{2023}\natexlab{c}.
\newblock \showarticletitle{On Explaining Confounding Bias}.
\newblock \bibinfo{journal}{\emph{2023 IEEE 39th International Conference on Data Engineering (ICDE)}} (\bibinfo{year}{2023}).
\newblock


\bibitem[Youngmann et~al\mbox{.}(2023d)]%
        {youngmann2023causal}
\bibfield{author}{\bibinfo{person}{Brit Youngmann}, \bibinfo{person}{Michael~J. Cafarella}, \bibinfo{person}{Babak Salimi}, {and} \bibinfo{person}{Anna Zeng}.} \bibinfo{year}{2023}\natexlab{d}.
\newblock \showarticletitle{Causal Data Integration}.
\newblock \bibinfo{journal}{\emph{Proc. {VLDB} Endow.}} \bibinfo{volume}{16}, \bibinfo{number}{10} (\bibinfo{year}{2023}), \bibinfo{pages}{2659--2665}.
\newblock


\bibitem[Yu et~al\mbox{.}(2009)]%
        {yu2009takes}
\bibfield{author}{\bibinfo{person}{Cong Yu}, \bibinfo{person}{Laks Lakshmanan}, {and} \bibinfo{person}{Sihem Amer-Yahia}.} \bibinfo{year}{2009}\natexlab{}.
\newblock \showarticletitle{It takes variety to make a world: diversification in recommender systems}. In \bibinfo{booktitle}{\emph{Proceedings of the 12th international conference on extending database technology: Advances in database technology}}. \bibinfo{pages}{368--378}.
\newblock


\end{thebibliography}


\appendix
In this part, we provide missing proofs and additional experiments. 
\section{Proofs}
\label{app:proofs}
\begin{figure}[!htb]
\begin{footnotesize}
\centering
\begin{tabular}{| c | c | c | c | c | c | c | c |}
\hline {\bf id} & {\bf $A_1$} & {\bf $A_2$} & {\bf $A_3$} & {\bf $O$} \\
\hline $t_1$ & $1$ & $0$ & $0$ & $0$ \\
\hline $t_2$ & $1$ & $0$ & $0$ & $0$ \\
\hline $t_3$ & $1$ & $1$ & $0$ & $0$ \\
\hline $t_4$ & $0$ & $0$ & $1$ & $0$ \\
\hline $t_5$ & $0$ & $1$ & $1$ & $0$ \\
\hline $t_{S_1}$ & $1$ & $-35$ & $7$ & $0$ \\
\hline $t_{S_2}$ & $12$ & $1$ & $-4$ & $0$ \\
\hline $t_{S_3}$ & $55$ & $97$ & $1$ & $0$ \\
\hline
\end{tabular}
\end{footnotesize}
\caption{Example query view for the reduction from set cover in the proof of \cref{prop:hardness}. The sets are $S_1 = \{1,2,3\}$, $S_2 = \{3,5\}$, and $S_3 = \{4,5\}$. If $k=2$, the patterns that cover $\frac{5+2}{5+3} = \frac{7}{8}$ of the tuples are $A_1 = 1$ and $A_3 = 1$, indicating that $S_1, S_3$ is a cover.}
\end{figure}

\begin{proof}[Proof of Proposition \ref{prop:approximation}] 
In the decision version of the Set Cover problem we are given with a universe of elements $U = \{x_1, \ldots, x_{n'}\}$, a collection of $m$ subsets $S_1, \ldots, S_{m'} \subseteq U$ and a number $k$. The question is whether there exists a cover of $U$ of at most $k'$ subsets.

Given an instance of the set cover problem, we build an instance of the \probName\ problem as follows. 
We build a relation $R$ with $m'+1$ attributes, $\attrset = (A_1, \ldots, A_{m'}, O)$, and containing $n'+m'$ tuples. For each element $x_i \in U$, we create a tuple $t_i$, such that $t_i[A_j]=1$ iff $x_i \in S_j$. 
We further add $m'$ tuples $t_{S_j}$ such that $t_{S_j}[A_j] = 1$, $t_{S_j}[O] = 0$, and $t_{S_j}[A_p] = l \neq 0$ for all $p \neq j$ where $l$ is a unique number not used anywhere else in an attribute of $R$. 
The query $\Qagg$ is the query that returns all the tuples in the relation $R$, i.e., 
$\Qagg$ = \texttt{SELECT} $A_1 \ldots A_{m'}$, \texttt{COUNT(*) FROM R GROUP BY} $A_1, \ldots, A_{m'}$.
Here, $\pattern_g$ can be any predicate, since the FD that needs to hold is $id \rightarrow \pattern_g $.  
Note that for each set of tuples defined by a pattern can only have an outcome of $0$, as the outcome of all tuples is $0$. 
Therefore, the explainability of all explanation patterns (\cref{def:explainability}) is $0$. 
For \probName, we further define $\theta = \frac{n'+k'}{n'+m'}$, $k = k'$, and $\tau = 0$. The underlying causal DAG, $G$, only contains the edges of the form $A_j \to O$ for all $1\leq j \leq m'$. 
We claim that there exists a cover of $U$ with at most $k$ sets iff there exists a solution $\Phi$ to \probName\ such that $|\Phi| \leq k$, all tuples are covered, and $\sum_{\varphi \in \Phi}explainability(\varphi) \geq 0$. 

($\Rightarrow$) Assume that we have a collection $S_{j_1}, \ldots, S_{j_k}$ such that $\cup_{j = j_1}^{j_k} S_j = U$. We show that there is a solution for \probName\ as follows. For each $S_{j_1}$, we choose for the solution the pattern $\pattern_g^{j_i}: A_{j_i} = 1$. 
We show that $\Phi = \{(\pattern_g^{j_1}, \emptyset), \ldots, (\pattern_g^{j_k}, \emptyset)\}$ is a solution to \probName. First, we note that all tuples of the form $t_i$ are covered by at least one explanation pattern by their definition. For the $m$ remaining tuples, we have coverage of at most $k$ tuples. These are the tuples $t_{S_{j_i}}$ that have $A_{j_i} = 1$. 
Thus, the number of covered tuples is exactly $n'+k'$ out of $n'+m'$ tuples in $R$. 
If there are fewer than $k$ tuples we can augment the original cover with arbitrary sets to obtain a cover of size $k$.

($\Leftarrow$) Assume that we have a solution to \probName\ with the aforementioned parameters. We show that we can find a solution to the set cover problem. 
Suppose the cover is $\Phi = \{(\pattern_g^{j_1}, \emptyset), \ldots, (\pattern_g^{j_k}, \emptyset)\}$. 
We first claim that no grouping pattern that includes $A_i = 0$ in a conjunction can be included in $\Phi$ as such a pattern will not cover any tuple $t_{S_j}$ since these tuples do not have an attribute with value $0$ by definition (and any other number other than $1$ will only cover a single tuple). Thus, the number of covered tuples will be $< \frac{n'+k'}{n'+m'} = \theta$, which would contradict the assumption that this is a valid solution to \probName. 
Hence, all patterns are of conjunctions of $A_i = 1$. For each treatment pattern of the form $\pattern_g = \wedge_{j=i_a}^{i_b} (A_j = 1) \land (A_p = l)$, we choose an arbitrary attribute in the conjunction $A_{j}$ if $(A_j = 1) \in \pattern_g$ and choose $S_{j}$ for the cover. Finally, if there is an uncovered element $x$ in $U$ and $\Phi$ includes a pattern in of the form $\pattern_g = (A_j = l)$ where $l \neq 1$, we choose for the cover a set $S$ that covers $x$ arbitrarily.  
We claim that the chosen collection of $k$ sets is a cover of $U$. 
To see this, 
recall that we claimed that the coverage of $\Phi$ is at least $n+k$. 
If the coverage includes tuples of the form $t_{S_j}$, then each pattern covers a single tuple. Suppose these patterns are $\pattern_a, \ldots, \pattern_b$. 
When building the coverage, instead of these patterns, we add a set that covers elements that are not yet covered by existing patterns. Thus, there are at least $b-a$ covered elements from $U$ in addition to the $n+k-(b-a)$ tuples covered by the patterns. Thus, the set cover we have assembled contains $n-(b-a) + (b-a) = n$ elements and covers all elements in $U$. 
\end{proof}

\paratitle{Randomized rounding algorithm}
The solution of the ILP formulation (\cref{subsec:step_3}) is determined by the values
of the variables $g_i$, indicating the selected explanation patterns. We
compute a solution by using any LP solver, then apply the following
randomized rounding procedure~\cite{raghavan1987randomized}: \\
{\bf(1)} If no solution is returned (LP is infeasible), return {\em ``no solution''}.\\ 
{\bf(2)} Let $g_1, \ldots, g_l$ and $t_1, \ldots, t_m$ be a solution to the LP. \\
{\bf(3)} Interpret the numbers for $g_1/k, \ldots, g_l/k$ as probabilities for the explanation patterns $\pattern_j, j = 1$ to $l$.\\
{\bf(4)} Choose $k$ explanation patterns independently at random according to these probabilities.\\
{\bf(5)} Return the collection $\Phi$ with the $k$ chosen explanation patterns.

We can show that if all grouping and treatment patterns are considered, this procedure yields a solution that covers at least $(1 - \frac{1}{e})\times \theta m$ groups from $\Qagg(\db)$ in expectation, and its overall explainability, in expectation,
is at least a ($\frac{1}{k}$) fraction of the corresponding optimal solution. 

\begin{proposition}\label{prop:lp_rounding}
The following holds for the LP-rounding algorithm of the ILP in \cref{fig:lp}:
\begin{enumerate}[leftmargin=*]
    \item If the LP-rounding Algorithm returns ``no solution'', then it is correct, i.e., no solutions exist for the ILP.
    \item Otherwise, the algorithm returns a collection $\Phi$ of $k$ explanation patterns that covers at least $(1 - \frac{1}{e})\times \theta \cdot m$ groups in expectation, and has an expected overall explainability of $\geq \frac{1}{k}\times OPT \ ILP$.
\end{enumerate}
  \end{proposition}


The proof adapts the proof from~\cite{maxcoverage} for the randomized rounding algorithm for the Maximum Coverage problem to our setting. 

\begin{proof}[Proof of \ref{prop:lp_rounding}]
    (1) A feasible solution of ILP($\mathcal{I}$) is also a feasible solution of LP($\mathcal{I}$). Hence if there are no fractional solution to LP($\mathcal{I}$), there are no integral solutions as well. \\
    
    (2) In this case the LP($\mathcal{I}$) returns some solution. \\
    
    (a) {\bf Claim: for every group $T_i$, the probability that $G$ covers $T_i$ is at least $(1 - \frac{1}{e}).t_i$.} If we choose a random pattern according to the probabilities $g_1/k, \cdots, g_l/k$, it covers group $T_i$ with probability $\sum_{j: T_i \in \pattern_j} g_j/k \geq t_i/k$ by (\ref{eq:lp-coverage}). Therefore, the probability that none of the $k$ patterns chosen by the rounding algorithm in case (2) covers $T_i$ is at most $(1 - t_i/k)^k$ and the probability that $T_i$ is covered is $\geq 1 - (1 - t_i/k)^k$ which is $\geq (1 - \frac{1}{e}).t_i$, since the left is concave and the right is linear, and the inequality holds at the end points $t_i = 0, 1$ in the interval $[0, 1]$. \\
    
    (b) {\bf Claim: at least  $(1 - \frac{1}{e})\times \theta m$ groups are covered in expectation.} Let $M_i$ be a random variable such that $M_i = 1$ if $T_i$ is covered by some pattern in $G$ and $= 0$ otherwise. The number of groups covered by $G$ is $\sum_{i = 1}^m M_i$. Expected number of groups covered by $G$ is $E [\sum_{i = 1}^m M_i] = \sum_i E[M_i]$ (by linearity of expectation) = $\sum_i Pr[M_i = 1]$ $\geq \sum_i (1 - \frac{1}{e}).t_i$ (by claim (a)) = $(1 - \frac{1}{e})\sum_{i = 1}^m t_i$  $\geq (1 - \frac{1}{e}) . \theta m$ by (\ref{eq:lp-coverage}) in the LP.\\

    (c) {\bf Claim: The total weight of patterns in $G$ is $\geq \frac{1}{k}\times OPT \ LP(\mathcal{I})$.} Let $C_j$ be a random variable such that $C_j = 1$ if $\pattern_j \in G$ and $= 0$ otherwise. The total weight of patterns in $G$ is $\sum_{j = 1}^l C_j . w_j$. Expected weights of patterns in $G$ is $E [\sum_{j = 1}^l C_j] = \sum_j E[C_j]$ (by linearity of expectation) = $\sum_j Pr[C_j = 1]$ $= \sum_j g_j/k$  = $\frac{1}{k} . OPT \ LP(\mathcal{I})$. \\
        
     (d) {\bf Claim: The total weight of groups in $G$ is $\geq \frac{1}{k}\times OPT \ ILP(\mathcal{I})$.} 
    Since an optimal solution of ILP($\mathcal{I}$) is a feasible solution of LP($\mathcal{I}$), we have 
    \begin{equation}
        OPT \ ILP(\mathcal{I}) \leq OPT \ LP(\mathcal{I})\label{eq:LP-ILP}
    \end{equation}
    Combining (\ref{eq:LP-ILP}) with claim (c), (d) follows. (a) and (d) prove the proposition.
\end{proof}


\section{Additional Experiments}
\label{app:exp}
Here we provide missing details for the experiments as well as additional experiments.

 \begin{figure}[t]
        \centering
        \begin{minipage}[b]{1.0\linewidth}
            \small
            \begin{tcolorbox}[colback=white]
            \vspace{-2mm}
   \textsf{$\bullet$ To buy a new car, having a \textcolor{blue}{caching account with at least 200 DM and paying back all credits at this bank duly} has the most significant positive effect on the risk score (effect size of 0.56, $p {<}$ 1e-3). Conversely, \textcolor{red}{requesting a loan with a duration exceeding 48 months} has the largest adverse impact on credit risk (effect size of -0.49, $p {<}$ 1e-5).}\\
                  \textsf{$\bullet$ To buy domestic appliances, \textcolor{blue}{requesting a loan with a duration not exceeding 12 months and paying back all credits at this bank duly} has the most significant positive effect on the risk score (effect size of 0.34,$p {<}$ 1e-3). Conversely, \textcolor{red}{requesting a loan with a duration exceeding 48 months} has the largest adverse impact on credit risk (effect size of -0.69, $p {<}$ 1e-5).}\\
                 \textsf{$\bullet$ To buy furniture or equipment, having a \textcolor{blue}{caching account with at least 200 DM} has the most significant positive effect on the risk score (effect size of 0.3, $p {<}$ 1e-5). Conversely, \textcolor{red}{requesting a loan with a duration exceeding 45 months} has the largest adverse impact on credit risk (effect size of -0.78, $p {<}$ 1e-3).}\\
        \textsf{$\bullet$ To get a loan for repairs, have a \textcolor{blue}{caching account with at least 200 DM and a saving account with at least 1000 DM} has the most significant positive effect on the risk score (effect size of 0.5, $p {<}$ 1e-3). Conversely, \textcolor{red}{not having a checking account and renting a house} has the largest adverse impact on credit risk (effect size of -0.66, $p {<}$ 1e-4).}\\
           \textsf{$\bullet$ To get a loan for retraining, having a \textcolor{blue}{owning a house} has the most significant positive effect on the risk score (effect size of 0.4, $p {<}$ 1e-2). Conversely, \textcolor{red}{requesting a loan with a duration exceeding 60 months} has the largest adverse impact on credit risk (effect size of -0.66, $p {<}$ 1e-3).}
          \vspace{-2mm}
            \end{tcolorbox}
        \end{minipage}
        \caption{German use-case example.}
        \label{fig:german}
    \end{figure}

 \begin{figure}[t]
        \centering
        \begin{minipage}[b]{1.0\linewidth}
            \small
            \begin{tcolorbox}[colback=white]
            \vspace{-2mm}
\textsf{$\bullet$ For blue-collar occupations (e.g., Machine-op-inspct, Craft-repair, Transport-moving), being an \textcolor{blue}{adult who is married} has the most significant positive effect on the death rate (0.25,$p {<}$ 1e-3). Conversely, \textcolor{red}{being unmarried} has the largest adverse impact on income (effect size of -0.2, $p {<}$ 1e-4).}\\
                \textsf{$\bullet$ For white-collar occupations (e.g., Exec-managerial, Prof-specialty, Adm-clerical), being a \textcolor{blue}{male with a bachelor's degree or higher} has the most significant positive effect on income (effect size of 0.38, $p {<}$ 1e-4). Conversely, \textcolor{red}{being unmarried} has the largest adverse impact on income (effect size of -0.23, $p {<}$ 1e-3).}\\
                 \textsf{$\bullet$ For service occupations (e.g., Sales, Other-service), being \textcolor{blue}{married} has the most significant positive effect on income (effect size of 0.53, $p {<}$ 1e-3). Conversely, \textcolor{red}{being unmarried female} has the largest adverse impact on income (effect size of -0.39, $p {<}$ 1e-4).}  
          \vspace{-2mm}
            \end{tcolorbox}
        \end{minipage}
        \caption{Adult use-case example.}
        \label{fig:adult}
    \end{figure}

\paratitle{\german}:
We analyze a query computing the average risk score for loan requests based on purpose. 
Due to the absence of functional dependencies in the dataset, each group in the aggregated view necessitates a distinct explanation, representing the individual groups in the result view. The explanation summary generated by our \sysName\ is presented in Figure \ref{fig:german_use_case}. Among the ten purposes examined, four purposes could not be explained, as none of the considered treatments were found to be statistically significant. This outcome can be attributed to the low number of tuples associated with those particular purposes in the dataset.

Our findings reveal the significant impact of checking and saving accounts' status and credit history on the risk score for all loan purposes.
These results align with previous research~\cite{GalhotraGRS22} and the association of these attributes with the Schufa score~\cite{schufa}, a widely used credit rating score in Germany. The Schufa score considers factors like credit history, existing loans, and negative incidents. However, specific details remain undisclosed.
Our system identifies influential attributes that affect desired outcomes, potentially shedding light on opaque algorithms like Schufa score.
Rules generated by \ids, \frl, and \exptable\ rely on correlations rather than causal relationships. For instance, a rule generated by \ids\ (and \exptable) suggests that a loan for a vacation corresponds to a credit risk score of 1. However, not all vacation loan requests should be approved. Here, the purpose of vacation loans constituted a small number of records, and coincidentally, all of them were associated with high saving account balances. These approaches failed to appropriately prioritize this pattern, resulting in a rule that lacks sensibility in terms of causality but is accurate in terms of prediction.

\paratitle{\adult}:
We analyzed average income across occupations using an aggregated query.
The income variable is a binary attribute, where $1$ represents an annual income greater than $\$50k$ and $0$ indicates otherwise.
To establish the grouping patterns, we utilized the attribute \textsc{occupation category}, which exhibited a functional dependency with the \textsc{occupation} variable. All other variables were utilized to define the treatment patterns.
The generated explanation summary can be found in Figure \ref{fig:adult_use_case}.
Marital status, education, and gender were identified as significant factors affecting income across all occupations, aligning with prior research~\cite{tramer2017fairtest,salimi2018bias}.
Salimi et al.\cite{salimi2018bias} found a high representation of married males and a strong positive association between marriage and high income in this dataset. Although there is no direct causal link between marital status and income, dataset inconsistencies resulted in marital status having the strongest impact in our analysis due to adjusted gross income based on filing status, which reflects household income. Additionally, \cite{salimi2018bias} demonstrated that males tend to have higher education levels and higher education is associated with higher incomes. Our results support these findings, revealing variations across occupation categories, particularly in white-collar occupations where higher education predominantly influences income. This highlights the benefit of our approach in providing detailed causal explanations.
\ids, \frl, and \exptable\ yielded comparable outcomes, with marital status being the best predictor of income, followed by gender, and age.  However, they fail to explain variations among occupations.


\paragraph*{\bf Breakdown Analysis of the \algoName Algorithm}
We analyze the operation of \algoName\ by step. The runtime analysis is depicted in Figure \ref{fig:algo_steps}. We observe that in all cases, mining the treatment pattern phase (Algorithm \ref{algo:treatment_patterns}) consumes most of the time. Since the number of grouping patterns is relatively small (as there are not too many FDs in the examined cases), the first and last steps are relatively fast. This aligns with our time complexity analysis and demonstrates the need for an efficient approach for avoiding iterating over all possible treatment patterns.

\begin{figure}[t]
\centering
\includegraphics[scale = 0.2]{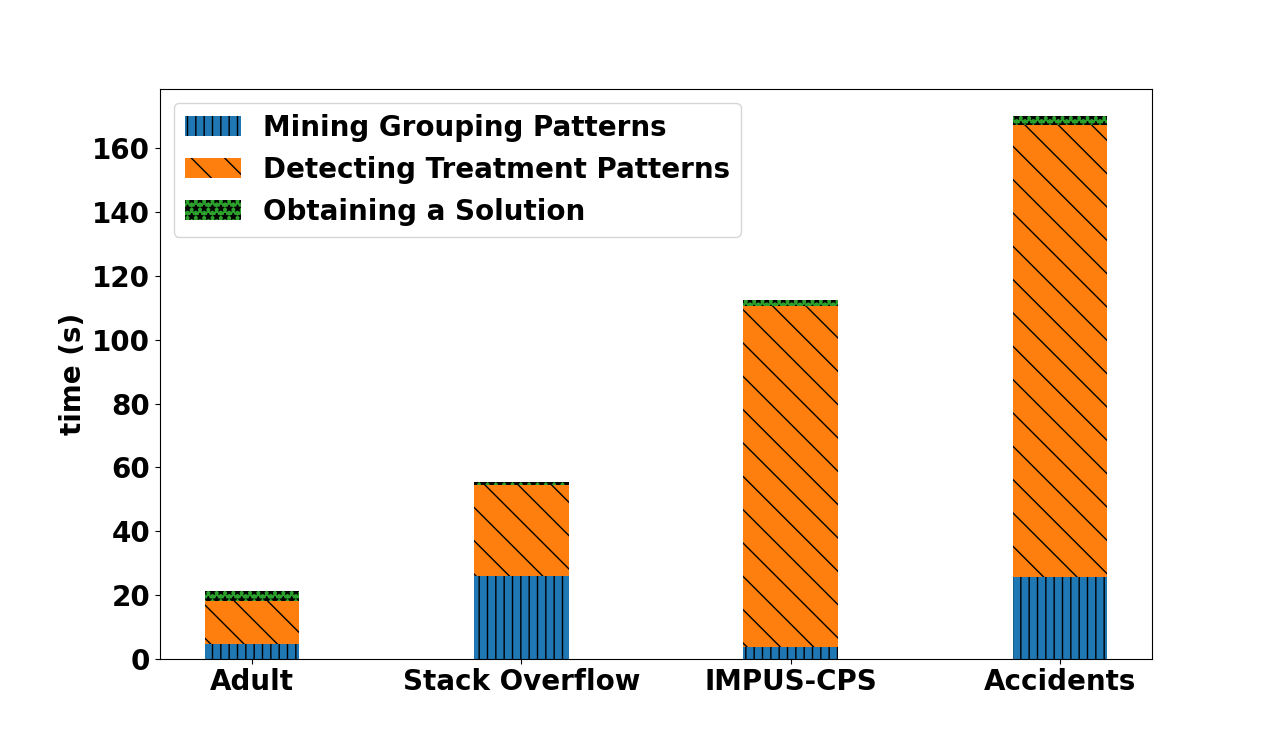}
\caption{Runtime by-step of the \algoName\ algorithm} \label{fig:algo_steps}
\end{figure}

\begin{figure}[t]
  \centering
  \begin{subfigure}[b]{0.23\textwidth}
    \centering
    \includegraphics[width=\textwidth]{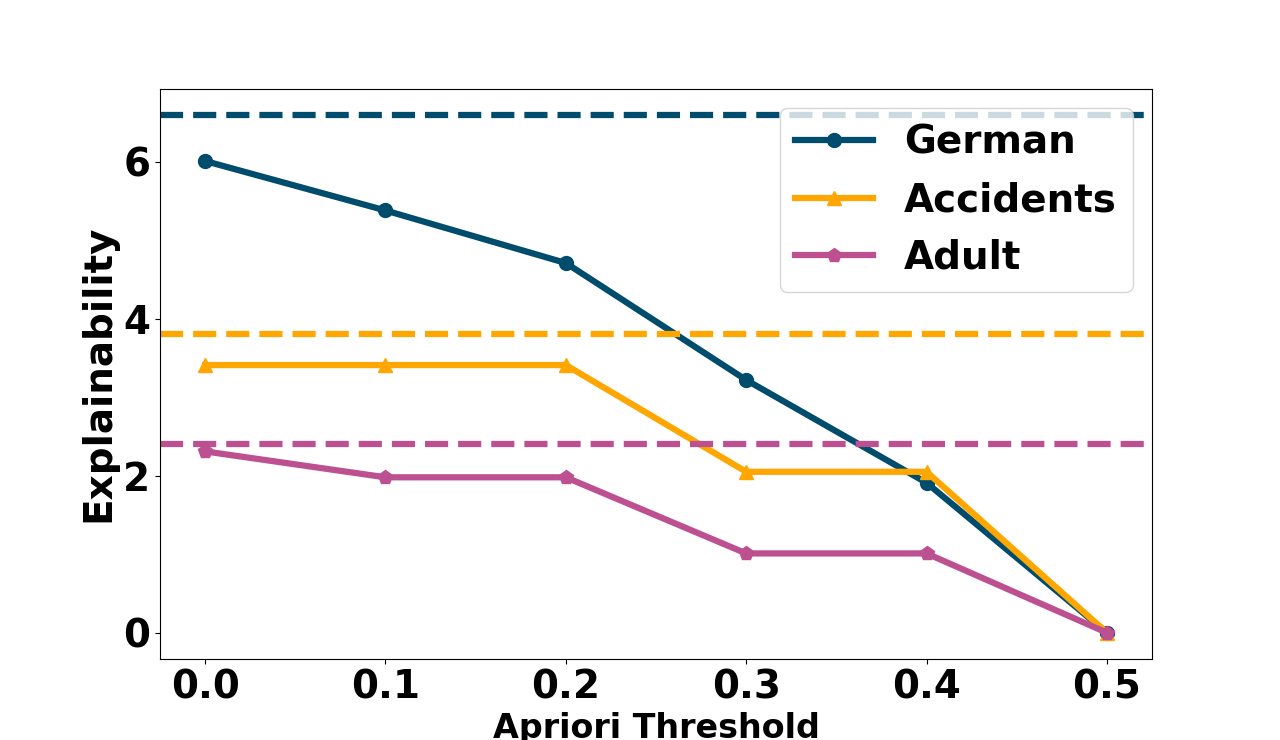}
    \caption{Explainability vs. Threshold}
    \label{fig:figure4}
  \end{subfigure}
     \begin{subfigure}[b]{0.23\textwidth}
    \centering
    \includegraphics[width=\textwidth]{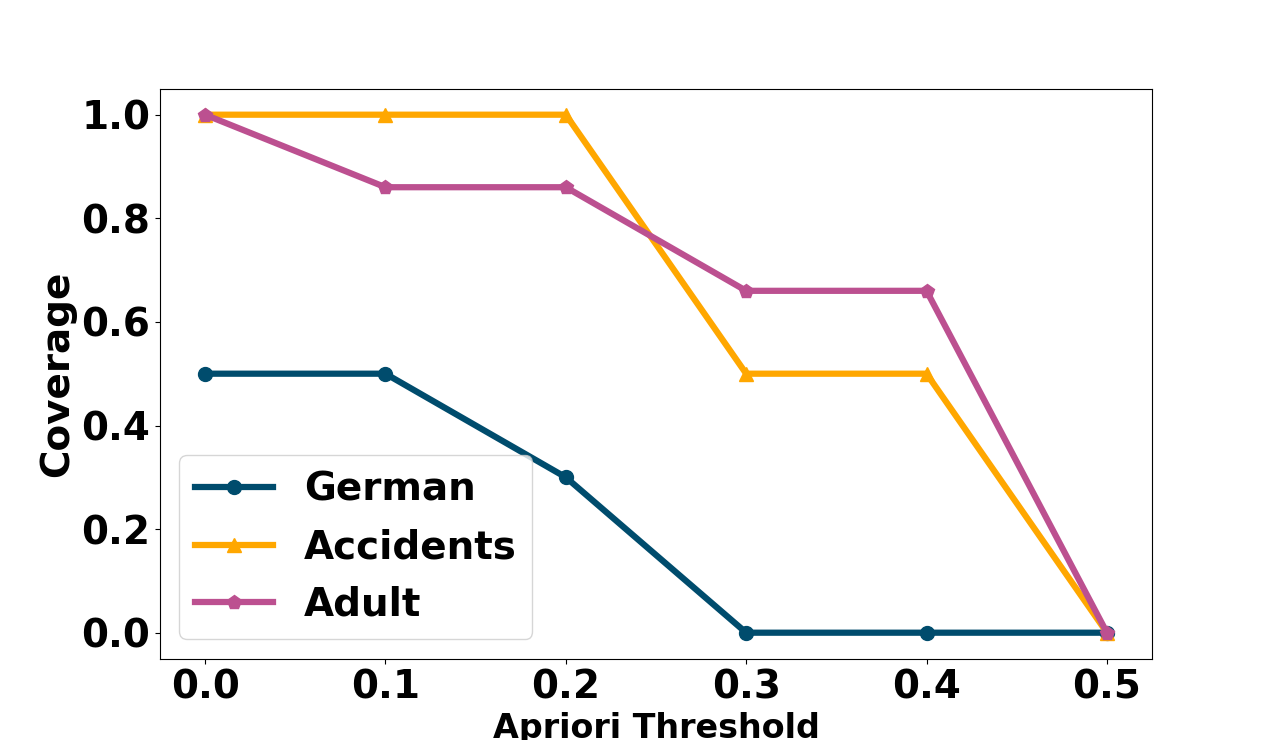}
    \caption{Coverage vs. Threshold}
    \label{fig:figure2}
  \end{subfigure}
  \caption{The Effect of Apriori Threshold.}
  \label{fig:apriori}
\end{figure}

\paragraph*{\bf Apriori Threshold}
We investigate the effect of varying the threshold parameter $\tau$ in the Apriori algorithm. Increasing $\tau$ leads to a reduction in the number of grouping patterns considered. Recall that Brute Force examines all possible grouping patterns, equivalent to setting $\tau {=} 0$. However, even when $\tau {=} 0$, \algoName\ and Brute Force yield different explainability scores. This is because \algoName\ does not explore all treatment patterns and therefore is not guaranteed to find the optimal ones. 
Our findings for the German, Adult, and Accident datasets are shown in Figure \ref{fig:apriori} (similar trends were observed for the other datasets). Figure \ref{fig:apriori}(a) presents the impact of the Apriori threshold on explainability. The dashed lines represent the explainability achieved by the Brute Force Algorithm. In Figure \ref{fig:apriori}(b), we illustrate the effect of the threshold on coverage.
Note that even when setting $\tau{=}0$, it is still not possible to cover all groups in the German dataset. This limitation arises because each group requires a separate explanation (due to the absence of FDs in this dataset), and the fact that the coverage is further restricted by the size constraint (which is $5$ in our setting).
As expected, higher threshold values lead to a decrease in both explainability and coverage. Based on our findings, we recommend using a default threshold of $0.1$, which provides satisfactory results in terms of runtime, explainability, and coverage. However, the analyst can adjust this threshold according to specific coverage requirements.

\begin{figure}[t]
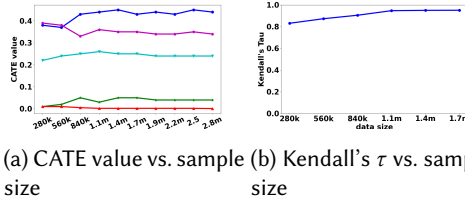

  \centering
  \begin{subfigure}[b]{0.23\textwidth}
    \centering
    \includegraphics[width=\textwidth]{figs/CATE_accidents.png}
    \caption{CATE value vs. sample size}
    \label{fig:figure4}
  \end{subfigure}
     \begin{subfigure}[b]{0.23\textwidth}
    \centering
    \includegraphics[width=\textwidth]{figs/order_accidents.png}
    \caption{Kendall's $\tau$ vs. sample size}
    \label{fig:figure2}
  \end{subfigure}
  \caption{CATE Values Estimation (Accidents dataset).}
  \label{fig:cate_estimation}
\end{figure}

\paragraph*{\bf CATE Values Estimation}
We conducted an investigation to assess how the sample size affects the estimation of CATE values. As the sample size increases, the running time also increases, but the accuracy of the estimations improves. 
Figure \ref{fig:cate_estimation} illustrates the results for the Accident dataset. In Figure \ref{fig:cate_estimation}(a), we present the estimated CATE values for $5$ random treatments using various sample sizes. In Figure \ref{fig:cate_estimation}(b), we evaluate the agreement between rankings using Kendall's $\tau$ correlation coefficient. We randomly selected $20$ treatments and ranked them based on their CATE values, comparing this ranking with rankings obtained using different sample sizes. Notably, for a sample size of $1m$ tuples, the CATE values exhibit an error of no more than $5\%$, and the Kendall's $\tau$ reaches a high and stable value of $0.95$. Similar trends were observed for the IMPUS-CPS dataset. Consequently, we conclude that a sample size of $1m$ tuples is suitable for accurate estimation of the CATE values.

\begin{table}[t]
\centering
\small
\begin{tabular}{|c|c|c|c|}

\hline
\textbf{Dataset} 
 & \textbf{Graph Name}& \textbf{Number of Edges} & \textbf{Density} \\
\hline
\multirow{4}{*}{German} & Used causal DAG & 20 & 0.05 \\
& PC & 43 & 0.11 \\
& FCI &  12& 0.03 \\
& LiNGAM & 8 & 0.02\\
\hline
\multirow{4}{*}{Adult} & Used causal DAG & 36 & 0.23 \\
& PC & 38 & 0.24 \\
& FCI & 10 & 0.06 \\
& LiNGAM & 18 & 0.11 \\
\hline
\multirow{4}{*}{SO} & Used causal DAG & 28 & 0.07 \\
& PC & 75 & 0.19 \\
& FCI & 41 &0.1  \\
& LiNGAM & 7 & 0.01 \\
\hline
\end{tabular}
\caption{Causal DAG statistics.}
\label{tab:causal_dags}
\end{table}

\begin{figure}[t]
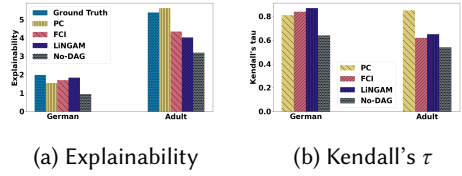

  \centering
  \begin{subfigure}[b]{0.23\textwidth}
    \centering
    \includegraphics[width=\textwidth]{figs/DAG_exp.png}

    \caption{Explainability}
    \label{fig:figure4}
  \end{subfigure}
     \begin{subfigure}[b]{0.23\textwidth}
    \centering
    \includegraphics[width=\textwidth]{figs/DAG_order.png}
 
    \caption{Kendall's $\tau$}
    \label{fig:figure2}
  \end{subfigure}

  \caption{Modifying the causal DAG.}
  \label{fig:causal_dag_effect}
\end{figure}

\paragraph*{\bf Causal DAG}
In this experiment, we deviate from the assumption of having a pre-defined causal DAG. Instead, we leverage existing causal discovery algorithms to construct the DAGs and examine their impact on the results. Modifying the DAG can lead to changes in the CATE values, consequently affecting their ranking based on CATE values. Hence, we present the effects on both overall explainability and the ranking of treatment patterns when using different causal DAGs. We conducted tests using three widely used causal discovery algorithms: PC~\cite{spirtes2000causation}, FCI~\cite{spirtes2000causation}, and LiNGAM~\cite{shimizu2006linear}. 
Statistics on the obtained causal DAGs are given in Table~\ref{tab:causal_dags}. 
Our results for the German and Adult and SO datasets are shown in Figure \ref{fig:causal_dag_effect} (similar patterns were observed for the other datasets). 
Figure \ref{fig:causal_dag_effect}(a) illustrates the obtained explainability scores for the German and Adult datasets (for SO, the explainability scores are on a different range and thus omitted from the presentation) using different causal DAGs. In Figure \ref{fig:causal_dag_effect}, we present the Kendall tau values, comparing the ranking of top-20 treatments based on their CATE values with the ranking obtained using the ground truth causal DAG. Notably, no single causal discovery algorithm outperforms all others.
While for the German dataset, LiNGAM yielded results that closely aligned with our ground truth, for the Adult dataset, the PC algorithm demonstrates superiority.
This is because each algorithm relies on specific assumptions to construct a causal DAG, leading to varying performance across datasets depending on the validity of those assumptions. We observe that the generated causal DAGs tend to be sparser than the ground truth DAGs. Hence, our key finding is that when a causal DAG is unavailable, one can utilize available algorithms to generate candidate causal DAGs, and then, by leveraging domain knowledge, determine the most suitable causal DAG that aligns with the domain expertise.

\end{document}